\documentclass[authoryear,11pt]{elsarticle}
\usepackage{blindtext}
\usepackage{amsfonts}
\usepackage{amssymb}
\usepackage{amsmath}
\usepackage{amsthm}
\usepackage[normalem]{ulem}
\usepackage{multirow}
\usepackage[flushleft]{threeparttable}
\usepackage{array}
\usepackage{setspace}
\usepackage{ragged2e}
\usepackage{lscape}
\usepackage{graphicx}
\graphicspath{ {./images/} }
\usepackage{amsfonts}
\usepackage{scalefnt}
\usepackage{natbib}
\usepackage{hyperref}
\usepackage{stix}
\usepackage{nameref}
\numberwithin{equation}{section}
\usepackage{appendix}

\providecommand{\U}[1]{\protect\rule{.1in}{.1in}}
\advance \topmargin by -\headheight
\advance \topmargin by -\headsep
\evensidemargin \oddsidemargin
\newtheorem{theorem}{Theorem}[section]

\newtheorem*{assumption_ahk}{Assumption AHK}
\newtheorem*{assumption_asm}{Assumption ASM}

\newtheorem*{assumption_se}{Assumption SE}
\newtheorem*{assumption_reg}{Assumption REG}
\newtheorem*{assumption_nd}{Assumption ND}
\newtheorem*{assumption_dml1}{Assumption DML1}
\newtheorem*{assumption_dml2}{Assumption DML2}
\newtheorem*{assumption_gmd}{Assumption GMD}

\newtheorem*{assumption_plm}{Assumption REG-P}

\newtheorem{claim}{Claim}[section]

\newtheorem{definition}{Definition}[section]

\newtheorem{lemma}{Lemma}[section]
\newtheorem{lemma*}{Lemma}

\newtheorem{remark}{Remark}[section]

\newcommand{\ind}{\perp\!\!\!\!\perp}

\begin{document}

\begin{frontmatter}

\title{Inference in High-Dimensional Panel Models:\\
Two-Way Dependence and Unobserved Heterogeneity \tnoteref{t1}}

\author[sufe]{Kaicheng Chen\fnref{fn1}}
\ead{chenkaicheng@sufe.edu.cn}
\fntext[fn1]{School of Economics, Shanghai University of Finance and Economics.}

\address[sufe]{School of Economics, Shanghai University of Finance and Economics}

\begin{abstract}
Panel data allows for the modeling of unobserved heterogeneity, significantly raising the number of nuisance parameters and making high dimensionality a practical issue. Meanwhile, temporal and cross-sectional dependence in panel data further complicates high-dimensional estimation and inference. This paper proposes a toolkit for high-dimensional panel models with large cross-sectional and time sample sizes. To reduce the dimensionality, I propose a variant of LASSO for two-way clustered panels. While being consistent, the convergence rate of LASSO is slow due to the cluster dependence, rendering inference challenging in general. Nevertheless, asymptotic normality can be established in a semiparametric moment-restriction model by leveraging a clustered-panel cross-fitting approach and, as a special case, in a partial linear model using the full sample. In an exercise of estimating multiplier using panel data, I demonstrate how high dimensionality could be hidden and the proposed toolkit enables flexible modeling and robust inference.
\end{abstract}
\begin{keyword}
high-dimensional panel regression \sep two-way cluster dependence \sep correlated time effects \sep unobservable heterogeneity \sep LASSO \sep double/debiased machine learning \sep cross-fitting. \\
\textit{JEL Classification:} C01, C14, C23, C33
\end{keyword}
\end{frontmatter}

\setcounter{page}{1} \thispagestyle{empty} \pagestyle{plain}

\section{Introduction}
In economic research, high dimensionality typically refers to the large number of unknown parameters relative to the sample size, under which traditional estimations are either infeasible or tend to yield noisy estimates. The issue of high dimensionality becomes more relevant as data availability grows and economic modeling involves more flexibility. Commonly, the problem of high dimensionality appears at least in three scenarios as follows:
\begin{itemize}
    \item The dimension of observable and potentially relevant variables can be large relative to the sample size. For example, in trade literature, preferential trade agreements (PTAs) usually involve a large number of provisions even though most policy analysis only focuses on the effect of a small subset of the provisions \footnote{For example, 282 PTAs were signed and notified to the WTO between 1958 and 2017, encompassing 937 provisions across 17 policy areas. See \cite{breinlich2022machine}.}. In demand analysis, even if the focus is on the own-price elasticity, the prices of relevant goods should also be included, unless strong assumptions for aggregation are made (see \citealp{chernozhukov2019demand}).
    
    \item With nonparametric or semiparametric modeling, the unknown functions are viewed as infinite dimensional parameters regardless of the dimension of observable characteristics. When the unknown function is approximated by a series of basis functions, the number of unknown parameters increases quickly. \footnote{For instance, the 2nd-order polynomial transformation of $k$-dimensional covariates generates $\frac{k^2}{2}+\frac{3}{2}k$ terms and the 3rd-order polynomial transformation generates $k+\frac{1}{2}k(k+1)+ \frac{1}{2} \sum_{l=1}^k l(l+1) = \frac{1}{6}k^3+k^2+\frac{11}{6}k$ terms.}

    \item Modeling of heterogeneity also raises the number of nuisance parameters. If the unobserved heterogeneity enters the model in a nonlinear way, either treating them as parameters or modeling them parametrically contributes to high dimensionality drastically. \footnote{This is particularly relevant in trade literature where the unobserved heterogeneity derived from the gravity model takes a pairwise form among the importers, exporters, and the time. As each of these three dimensions expands, the number of nuisance parameters explodes quickly.}.
\end{itemize}

Particularly, the modeling of unobserved heterogeneity in panel models makes high dimensionality more of a practical issue rather than just a theoretical concern. As a motivating example, let's consider a panel model where all three sources of high dimensionality are involved:
\begin{align}
         Y_{it} = D_{it}\theta_0+g_0(X_{it},c_i,d_t)+U_{it}, \label{example}
    \end{align}
where ${D}_{it}$ are low-dimensional treatments or policy variables. ${X}_{it}$ are high-dimensional controls. $D_{it}$ can also contain some higher-order effects and interactive effects with a subset of the controls to allow for nonlinear and heterogeneous effects in a parametric way. $g(.)$ is an unknown function, e.g., an infinite-dimensional parameter; $c_i$ and $d_t$ are unobserved heterogeneous effects. The interest is in the inference on the low-dimensional parameters ${\theta}_0$.

Without considering the features of panel data and the unobserved heterogeneity, it is a classic partial linear model that has been well-studied in the semiparametric literature. With high dimensionality, sparse approximation and regularization approaches have been widely employed to reduce the dimensionality. Essentially, regularization, also known as the machine learning approach, trades off bias for smaller variance to achieve desirable rates of convergence. However, due to the bias introduced by regularization and overfitting, inference can be challenging. Typically, some bias-correction procedures are involved to obtain estimators with better statistical properties and to conduct valid inference. 

In the case of panel data, three challenges remain with the existing high-dimensional approaches. First of all, the statistical properties of many regularized estimators remain unknown with panel data where dependence exists across space/unit and time. Secondly, some bias-correction procedures for inference, such as sample-splitting/cross-fitting, are particular about sampling assumptions, and existing approaches are not valid under two-way dependence in the panel. Thirdly, the unobserved individual and time effects may appear in a flexible way, which further complicates estimation and inference.

Although the standard LASSO that assumes the Gaussian error condition does not restrict the dependence structure of data as long as the (long-run) variance is finite, the validity of this approach largely depends on whether a theoretically required penalty level can be found in practice. The common approach for the standard LASSO to decide the penalty level is the K-fold cross-validation, and its validity has been established in \cite{chetverikov2021cross}, but proper modifications and validity for dependent data remain as open questions. In another strand of literature, \cite{gao2024robust} establish new Fuk-Nagaev type concentration inequality, by which they show two variants of LASSO continue to work with two-way dependence and non-Gaussianity, with a rate requirement on the penalty level. However, to determine the practical penalty level, it is unclear if the modified BIC approach continues to work under two-way dependence. Furthermore, the validity of this approach under cluster dependence where the correlation may be strong within the cluster remains unknown. \cite{chernozhukov2021lasso} deal with a system of (time series) equations, and they propose to run LASSO equation-by-equation while determining the penalty level jointly through a multiplier bootstrap which accounts for cross-sectional dependence. However, their performance bound for LASSO only works for each time series and is not compatible with a pooled panel. Also, its validity under cluster dependence is not unknown. I proposed a variant of LASSO that uses regressor-specific penalty weights robust to two-way cluster dependence and weak temporal dependence across clusters. Such a LASSO approach is labeled as the two-way cluster-LASSO, corresponding to the heteroskedasticity-robust LASSO in \cite{belloni2012sparse} and the cluster-LASSO in \cite{belloni2016inference}. This approach theoretically derives the common penalty level $\lambda$ up to a constant and a small-order sequence that does not vary across different data-generating processes. Therefore, data-driven tuning, such as information criterion and cross-validation, is not needed, which makes it more computationally efficient and avoids non-trivial theories that take data-driven tuning into account. In Table \ref{table_compare}, a comparison of aforementioned LASSO-type estimators are presented.

\begin{table}[htbp]
\centering \label{table_compare}
\footnotesize            
\setlength{\tabcolsep}{3pt} 
\renewcommand{\arraystretch}{1.15} 
\caption{Comparison of LASSO-type estimators for two-way dependent panel data}
\label{lasso_comparison}
\begin{tabular}{p{2.2cm}p{1.8cm}p{2.0cm}p{1.8cm}p{2.9cm}p{4.0cm}}
\hline
Method 
& Reference
& Penalty level
& Non-Gaussian error
& Two-way \ \ \ \ \ \ \ \ \ \ \ \ \ \ \ \ \ \ \ \ \ \ \ \ \ \ \  dependence
& Other limitations \\
\hline
Standard LASSO 
& \cite{hansen2022econometrics}
& Cross-validation 
& Not allowed
& Only if penalty is chosen correctly
& Validity of cross-validation under two-way dependence is unknown; computationally costly. \\
\hline
Adaptive or conservative LASSO 
& \cite{gao2024robust}
& Modified BIC
& Allowed
& Weak dependence allowed; unknown under cluster dependence
& Validity of the modified BIC under two-way dependence is unknown. \\
\hline
Equation-by-equation LASSO
& \cite{chernozhukov2021lasso} 
& Multiplier bootstrap
& Allowed 
& Weak dependence allowed; unknown under cluster dependence  
& Not compatible with pooled panels. \\
\hline
Two-way cluster-LASSO
& This paper
& Theoretically driven 
& Allowed
& Allowed in general
& More conservative; may not work well with weak signals. \\
\hline
\end{tabular}
\end{table}

According to the rate result, the proposed two-way cluster-LASSO is consistent as both $N,T$ diverge jointly, but the convergence rate is not as fast as the common rates for LASSO under independence or weak dependence\footnote{The intuition for the slow rate of convergence is illustrated in \hyperref[app_0]{Appendix}.}. When $N,T$ diverge at a comparable rate, the rate of convergence is slower than $(N\wedge T)^{-1/2}=(NT)^{-1/4}$ , which is a common rate requirement for inferential theory. This is where the second challenge arises: if a faster rate of convergence is not achievable due to the two-way cluster dependence, some bias-correction approaches are needed to relax the rate requirement for valid inference. There are many bias-correction methods for high-dimensional models.  The orthogonalization of moment functions in \cite{Chernozhukov2018} provides a general way for constructing estimators that features multiplicative error terms, by which the rate requirement on the nuisance parameter estimation can be relaxed. Combining with a cross-fitting procedure to further control the overfitting bias, they obtain valid inferential results for high-dimensional regression models. However, cross-fitting is sensitive to the sampling assumption. Building upon recent development of cross-fitting approaches for dependent data (\cite{Chiang2022,Semenova2023}), I propose a clustered-panel cross-fitting scheme and establish its validity. Effectively, this inferential procedure extends the double/debiased machine learning (DML, hereafter) approach by \cite{Chernozhukov2018} to panel data models, and it is labeled as the panel DML. Asymptotic normality for the panel DML estimator and the consistency for the variance estimator are established. It is shown that the crude requirement on the rate of convergence can be relaxed to $(N\wedge T)^{-1/4}$, which admits the first-step estimation through the two-way cluster-LASSO.

For the third challenge caused by the unobserved heterogeneity, existing approaches assume that $(c_i,d_t)$ are either additive (\citealp{belloni2016inference,kock2019uniform,clarke2025double}) or interactive (\citealp{vogt2022cce}).  To allow for flexible function forms while remaining tractable, I propose to model $(c_i,d_t)$ as correlated random effects through a generalized Mundlak device . In that way, a very rich form of heterogeneity is permitted. A closely related idea has been implemented in \cite{wooldridge2020inference} and \cite{clarke2025double}, what's different in this paper is that both unit and temporal unobserved effects are considered, and they are not separable from observable covariates. This is made possible by exploring the sparsity condition on $g_0$. Furthermore, a subtle issue of cross-fitting is discussed with the presence of  unobserved heterogeneity. Although valid inference remains challenging for high-dimensional models without cross-fitting in general, I show that inferential theory can be established in model (\ref{example}) using the full sample with a slightly stronger sparsity condition.  

The simulation results are quite revealing. When the model is mildly high-dimensional, non-robust methods, i.e. those designed for i.i.d. data or one-way cluster data, are comparable to the proposed two-way cluster-robust method, and the unregularized method remains valid.  This is true for both i.i.d and two-way clustered panel. However, when the nuisance parameters are truly high-dimensional, the pattern changes significantly: (1) with i.i.d data,  the unregularized method completely fails the task, and the proposed two-way cluster-robust methods has comparable performance to the non-robust methods in term of both estimation and inference; (2) with two-way clustered panel data, both unregularized methods and non-robust high-dimensional methods severely over-reject in the tests, and the proposed two-way cluster-robust method dominates in terms of estimation bias, SD, RMSE, and, particularly, inference coverage.

In the empirical application, I re-examine the effect of government spending on the output of an open economy following the framework of \cite{nakamura2014fiscal}, a well-cited empirical macro paper. The baseline model is not concerned with the high dimensionality: a linear panel model with a small number of covariates and additive unobserved heterogeneous effects; the identification is through the instrumental variable. However, even in a conventionally low-dimensional setting, high dimensionality could be hidden because the true model could be highly nonlinear in the covariates and the unobserved heterogeneity. To avoid the endogeneity caused by the potential misspecification in the function form, I consider extending the baseline model in a flexible way as in \ref{example}. The proposed dependence-robust estimation and inference for high-dimensional models can be leveraged, and the results can be used for a robustness check. It is shown that the estimates are consistent with the baseline results, which indicates that the nonlinear and interactive effects may not be very relevant in this model. However, existing approaches that are not robust to high dimensionality or two-way cluster dependence tend to over-fit, bringing noisy estimates and inaccurate inference results.

The rest of the paper is outlined as follows: Section \ref{tw_cluster_lasso} presents the two-way cluster-LASSO estimator and the investigation of its statistical properties under two-way cluster dependence. Section \ref{sub_sample} introduces the clustered-panel cross-fitting for inference. Section \ref{partial_linear_unob} studies the partial linear model with unobserved heterogeneity as a leading example. Simulation evidence is given in Section \ref{mc_simulation}. Section \ref{empirical_study} presents an empirical estimation of the government spending multiplier as an illustration of hidden high dimensionality and the application of the proposed toolkit. Section \ref{sec_conclusion} concludes the paper with empirical recommendations.

\subsection*{Notation.} \label{notation_sec} Here is a collection of frequently used notations in this paper. Some extra notations are defined along with the context. ${\rm E}$ and ${\rm P}$ are as generic expectation and probability operators. $\mathcal{P}_{NT}$ is an expanding collection of all data-generating processes ${P}$ that satisfy certain conditions. ${P}_{NT}$ is a sequence of probability laws such that ${P}_{NT} \in \mathcal{P}_{NT}$ for each $(N,T)$. The dependence on $(N,T)$ and ${P}_{NT}$ will be suppressed whenever clear in the context. $\Vert . \Vert$ is the Euclidean (Frobenius) norm for a matrix. Let $\mathbf{x}$ be a generic $k\times 1$ real vector, then the $l^q$ norm is denoted as  $\Vert \mathbf{x} \Vert_q := \left(\sum_{j=1}^k x_{j}^q \right)^{1/q}$ for $1\leq q < \infty$; $\Vert \mathbf{x} \Vert_{\infty} := \max_{1\leq j\leq k} |x_j|$. The $L^q({P})$ norm is denoted as  $\Vert f  \Vert_{{P},q} := \left(\int \Vert f(\omega)\Vert^q d {P}(\omega)\right)^{1/q} $ where $f$ is a random element with probability law ${P}$. I denote the empirical average of $f_{it}$ over $i=1,...,N$ and $t=1,...,T$ as $\mathbb{E}_{NT}[f_{it}] = \frac{1}{NT} \sum_{i=1}^N \sum_{t=1}^T f_{it}$ and the empirical $L^2$ norm as $ \left\Vert f_{it} \right\Vert_{NT,2} = \left(\frac{1}{NT} \sum_{i=1}^N\sum_{t=1}^T \Vert f_{it}\Vert ^2\right)^{1/2}$. Correspondingly, I denote the empirical average of $f_{it}$ over the sub-sample $i\in I_k$ and $t\in S_l$ as $\mathbb{E}_{kl}[f_{it}] = \frac{1}{N_k T_l} \sum_{i\in I_k, t\in S_l} f_{it}$ and the empirical $L^2$ norm over the subsample as $ \left\Vert f_{it} \right\Vert_{kl,2} = \left(\frac{1}{N_kT_l} \sum_{i\in N_l}\sum_{t\in T_l} \Vert f_{it}\Vert ^2\right)^{1/2}$, where $I_k,S_l$ are sub-sample index sets and $N_k,T_l$ are sub-sample sizes that will be introduced next section.

\section{Two-Way Cluster-LASSO} \label{tw_cluster_lasso}

In the existing literature, not much is known in terms of statistical properties for high-dimensional methods under cluster dependence in both cross-section and time. In this section, a variant of the $l1$-regularization methods, also known as the LASSO, is proposed and examined. To focus on the LASSO approach under two-way dependence, I consider a simple conditional expectation model of a scalar outcome given a potentially high-dimensional vector of covariates. Let $(Y_{it},X_{it})$ be a sample with $i=1,...,N$ and $t=1,...,T$. The conditional expectation model can be expressed as follows:
\begin{align}
    Y_{it} = f(X_{it}) + V_{it}, \ {\rm E}[V_{it}|X_{it}]=0 \label{model_lasso}
\end{align}
where $f(X_{it}):= {\rm E}[Y_{it}|X_{it}]$ is an unknown conditional expectation function of potentially high-dimensional covariates $X_{it}$; $V_{it}$ is the associated stochastic error. 

To characterize the two-way cluster dependence in the panel, I assume the random elements $W_{it}:=(Y_{it},X_{it},V_{it})$ are generated by the following process: 
\begin{assumption_ahk}[Aldous-Hoover-Kallenberg Component Structure Characterization]
\label{ahk} 
    \begin{align}
        W_{it}= \mu + f\left({\alpha }_i,{\gamma }_t,{\varepsilon }_{it}\right), \ \ \forall i\geq1, t\geq 1, \label{component}
    \end{align}
        where $\mu = {\rm E}[W_{it}]$, $f$ is some unknown measurable function; $(\alpha_{i})_{i\geq 1}$, $(\gamma_{t})_{t\geq 1}$, and $(\varepsilon_{it})_{i\geq 1, t\geq 1}$ are mutually independent
        sequences, $\alpha_{i}$ is i.i.d across $i$, $\varepsilon_{it}$ is i.i.d across $i$ and $t$. $\{\gamma_t\}_{t\geq 1}$ is strictly stationary and is beta-mixing at a geometric rate: 
    \begin{align}
        \beta_{\gamma}(m) = \sup_{s\leq T} \beta\left(\{\gamma_t\}_{t\leq s},\{\gamma_t\}_{t\geq s+m}\right) \leq c_{\kappa}exp(-\kappa m), \forall m \in \mathbb{Z^+} \label{beta},
    \end{align}
    for some constants $\kappa>0$ and $c_{\kappa}\geq 0$, where $\beta(X,Y) = \frac{1}{2}\Vert {P}_{X,Y}-{P}_{X}\times {P}_{Y} \Vert_{TV}$ and $\Vert . \Vert_{TV}$ denotes the total variation norm. 
\end{assumption_ahk}

Assumption \hyperref[ahk]{AHK} is motivated by a representation theorem for an exchangeable array, named after Aldous-Hoover-Kallenberg (AHK, hereafter), which states that if an array of random variables $(X_{ij})_{i\geq 1, j\geq 1}$ is separately or jointly exchangeable\footnote{An array $(X_{ij})_{i\geq 1, j\geq 1}$ is separately exchangeable if $\left(X_{\pi(i),\pi'(j)}\right)\overset{d}{=} \left(X_{ij}\right)$, and jointly exchangeable if the same condition holds with $\pi=\pi'$.}, then $X_{ij} = f(\xi_i,\zeta_{j},\iota_{ij})$ where $(\xi_i)_{i\geq 1},(\zeta_j)_{j\geq 1},(\iota_{ij})_{i\geq 1, j\geq 1}$ are mutually independent, uniformly distributed i.i.d. random variables. However, the exchangeability is not likely to hold for arrays with the presence of a temporal dimension since it is naturally ordered. In macroeconomics, for instance, we can interpret the time components $(\gamma_{t})_{t\geq 1}$ as unobserved common time shocks, which are naturally correlated over time, implying that the exchangeability is violated. Therefore, by allowing $\gamma_t$ to be correlated, it introduces temporal dependence across all clusters, making the characterization more sensible in the panel data setting. The beta-mixing condition restricts the temporal dependence of the common time effects to decay at an exponential rate, which is common in literature. The relaxation of the independence condition on $(\gamma_{t})_{t\geq 1}$ can be viewed as a generalization of the component structure representation, as argued by \cite{CHS_Restat}. It is important to note that the components in \ref{component} simply characterize the dependence in panel data. Differing from factor models or models with unobserved heterogeneity, they do not affect the identification of the regression model.

Due to the potential high dimensionality in $X$, traditional nonparametric methods are not appropriate for estimating the unknown function $f$ due to the curse of dimensionality. A common approach to reduce the dimensionality is by taking into account the sparsity information in models through regularization. Although the unknown function $f$ is an infinite-dimensional parameter, which is not exactly sparse,  we can view it from a sparse approximation perspective, following \cite{belloni2012sparse}:

\begin{assumption_asm}[Approximate Sparse Model]
\label{asm}
    The unknown function $f$ can be well-approximated by a dictionary of transformations $f_{it} = F(X_{it})$ where $f_{it}$ is a $p\times 1$ vector and $F$ is a measurable map, such that 
    \begin{align*}
        f(X_{it}) = f_{it}' \zeta_0 + r_{it} 
    \end{align*}
    where the coefficients $\zeta_0$ and the approximation error $r_{it}$ satisfy
    \begin{align*}
        \Vert \zeta_0 \Vert_0 \leq s = o(N\wedge T), \ \Vert r_{it}\Vert_{NT,2} =O_P\left(\sqrt{\frac{s}{N\wedge T}}\right).
    \end{align*}
\end{assumption_asm}
Assumption \hyperref[asm]{ASM} views the high-dimensional linear regression as an approximation. It requires a subset of the parameters $\zeta_0$ to be zero while controlling the size of the approximation error. Compared to the sparsity condition in previous literature, here it imposes a slower rate of growth restriction on the non-zero slope coefficients. For example, $s = o(NT)$ corresponds to the case of heteroskedasticity-robust LASSO under i.i.d data in \cite{belloni2012sparse}; $s=(Nl_T)$ corresponds to the cluster-robust LASSO under temporal dependence panel data in \cite{belloni2016inference} where $l_T\in [1,T]$ is an information index that equals T when there is no temporal dependence and equals 1 when there is cross-sectional independence and perfect temporal dependence. In other words, the underlying component structure restricts the growth of nonzero slope coefficients of the model in a way similar to the perfect temporal-dependence case. 

Under Assumption \hyperref[asm]{ASM}, we can rewrite the model \ref{model_lasso} as
\begin{align}
       Y_{it} = f_{it}'\zeta_0 + r_{it} + V_{it}, \ {\rm E}[V_{it}|X_{it}] =0. \label{model_lasso_approx}
\end{align}
We then apply $l1$ regularization in the least squared error problem under the model \ref{model_lasso_approx}. Let $\lambda$ be some non-negative common penalty level and $\omega$ be some non-negative $p\times p$ diagonal matrix of regressor-specific penalty weights. Consider the following generic weighted LASSO estimator:
\begin{align}
     \widehat{\zeta} = \arg\min_{\zeta} \frac{1}{NT}\sum_{i=1}^N \sum_{t=1}^T (Y_{it} - f_{it}'\zeta)^2 + \frac{\lambda}{NT} \Vert {\omega}\zeta \Vert_1. \label{lasso_main}
\end{align}

We note that the level of penalty term is determined by both $\lambda$ and $\omega$. A large penalty level can reduce avoid noisy estimation due to overfitting but meanwhile too large a penalty level can cause under selection which introduces missing variable bias. This is clearly a tradeoff between overfitting variance and regularization bias. Therefore, to obtain the desirable property of LASSO estimation, $\lambda$ and $\omega$ need to be determined jointly both in theory and in practice. A common choice of $\omega$ is $\mathbb{E}_{NT}[\dot{f}_{it}'\dot{f}_{it}]$ where $\dot{f}_{it}$ be the demeaned $f_{it}$ using the sample mean\footnote{The demeaning is done because of the inclusion of the intercept term which is not penalized.}. With this choice of $\omega$, regressors are standardized, so the model selection is not affected by the scale of the regressors. In theory, given $\omega$, $\lambda$ needs to be chosen in a way that the following event happens with high probability:
\begin{align}
  \max_{j=1,...,p}  \left|\frac{1}{NT}\sum_{i=1}^N\sum_{t=1}^T{\omega}_{j}^{-1}f_{it,j}V_{it}\right| \leq \frac{\lambda}{2C_\lambda NT}. \label{regevent}
\end{align}
where $C_\lambda>1$ is some unknown constant. Condition \ref{regevent} is referred to as the ``regularization event'' in the literature. If the error term $V_{it}$ is conditionally Gaussian or sub-Gaussian, then the Gaussian tail inequality implies an asymptotic order for $\lambda$. To enable applications in broader settings such as asymmetry and heavy tails of the error term, recent literature has investigated in new tools that ensures validity of LASSO with weaker conditions. Particularly, for panel data models, \cite{babii2023machine} derives a Fuk-Nagaev type concentration inequality for panel data which ensures the validity of LASSO under non-Gaussianity and weak temporal dependence. Under functional dependence measure, {\cite{gao2024robust} establish new Fuk-Nagaev type concentration inequality and \cite{chernozhukov2021lasso} extend Gaussian approximation results to ensure LASSO validity with the presence of temporal and cross-sectional dependence. In practice, however, all approaches mentioned above require further estimation of data-driven tuning for the penalty level using cross-validation, information criterion, or bootstrap. These tuning methods are usually computationally costly and hard to justify in theory, and they are further complicated by the temporal and cross-sectional dependence in panel \footnote{\cite{chernozhukov2021lasso} do provide theoretical justification for choosing $\lambda$ using the multiplier block bootstrap method.}. Furthermore, the restriction imposed for the functional dependence measure in \cite{chernozhukov2021lasso} and \cite{gao2024robust} excludes the two-way cluster dependence considered here as we allow the dependence to not decay over cross-sectional or time.

\cite{belloni2012sparse} view $\omega$ as a self-normalizer and leverage moderate deviation theorems for the self-normalized sums to choose the penalty level that ensures Condition \ref{regevent}. While this approach does not require extra data tuning once $\omega$ is chosen properly according to the theory, the validity of the approach is restricted: existing moderate deviation theorems only work for independent or weakly dependent random variables. For panel data with cross-sectional independence, one can cluster within each cross-sectional unit or construct temporal blocks that are approximately independent, but there is no existing moderate deviation theorem with self-normalizer that works for two-way dependence. Instead, I utilize the component structure characterization of the dependence and consider a Hoeffding-type decomposition of the high-dimensional mean-zero error term $f_{it}'V_{it}$: $$a_i = {\rm E}[f_{it}'V_{it}|\alpha_i], \ \  g_t = {\rm E}[f_{it}'V_{it}|\gamma_t], \ \ e_{it} = f_{it}'V_{it} - a_i - g_t.$$ 
The goal is to design penalty weights that account for the randomness from all three components and, importantly, being adaptive for both non-degenerate and degenerate cases.

To illustrate, let's first define a generic three-term regressor-specific penalty weight as follows:
\begin{align}
    \omega_{j}  = 
    \sqrt{\omega_{a,j}^2+\omega_{g,j}^2+\omega_{e,j}^2}  \label{penaltyweights}.
\end{align}
where $\omega_{a,j},\ \omega_{g,j}$, and $\omega_{e,j} $ are non-negative weights correspond to the three components. Let  $a_{i,j}$ be the entries of $a_i$ for $j=1,...,p$; $g_{t,j}$ and $e_{it,j}$ are defined similarly. To pin down an appropriate choice of penalty weights, we note that the following inequality holds,  for each $j=1,...,p$,
\begin{align}
 &P\left( \left|\frac{1}{NT}\sum_{i=1}^N\sum_{t=1}^T{\omega}_{j}^{-1}f_{it,j}V_{it}\right|>z\right)= P\left(\left|\frac{{\rm E}_{N}[a_{i,j}] +{\rm E}_{T}[g_{t,j}] + {\rm E}_{NT}[e_{it,j}] }{\sqrt{\omega_{a,j}^2+\omega_{g,j}^2+\omega_{e,j}^2} }\right| >z\right) \nonumber \\ 
\leq & P\left( \left|\frac{{\rm E}_{N}[a_{i,j}]}{\omega_{a,j}}\right| >\frac{z}{c_\omega}\right) + P\left( \left|\frac{{\rm E}_{T}[g_{t,j}]  }{\omega_{g,j} }\right| >\frac{z}{c_\omega}\right) +P\left( \left|\frac{ {\rm E}_{NT}[e_{it,j}] }{\omega_{e,j} }\right| >\frac{z}{c_\omega}\right) \label{three_ineq}
\end{align}
where $c_\omega:= \frac{\omega_{a,j}+\omega_{g,j}+\omega_{e,j}}{\sqrt{\omega_{a,j}^2+\omega_{g,j}^2+\omega_{e,j}^2}}\in \left[1,\sqrt{3}\right]$ is a scaling constant. To see why \ref{three_ineq} holds, we note that for any positive numbers $A,B,C$, $a,b,c$ and $d=\frac{a+b+c}{\sqrt{a^2+b^2+c^2}}$, we have
$$\left\{\frac{A+B+C}{\sqrt{a^2+b^2+c^2}}>z\right\} \subseteq \left\{\frac{A}{a}> \frac{z}{d} \right\} \cup \left\{\frac{B}{b}>\frac{z}{d}\right\} \cup \left\{\frac{C}{c}>\frac{z}{d}\right\}.$$ 
To prove it, suppose $\neg \left\{\left\{\frac{A}{a}>\frac{z}{d}\right\} \cup \left\{\frac{B}{b}>\frac{z}{d}\right\} \cup \left\{\frac{C}{c}>\frac{z}{d}\right\}\right\}$, which implies $\left\{\frac{A}{a}\leq \frac{z}{d}\right\} \cap \left\{\frac{B}{b}\leq \frac{z}{d}\right\} \cap \left\{\frac{C}{c}\leq \frac{z}{d}\right\}$. Then, we have $\frac{A+B+C}{\sqrt{a+b+c}} = \frac{A}{a} \frac{a}{\sqrt{a+b+c}} + \frac{B}{b} \frac{b}{\sqrt{a+b+c}} + \frac{C}{c} \frac{c}{\sqrt{a+b+c}} \leq \frac{z}{d} d =z $, which is $\neg \left\{\frac{A+B+C}{\sqrt{a+b+c}}>z\right\} $, so it is proved by contrapositive.

Inequality \ref{three_ineq} shows that it suffices to consider the tail probability of each component when choosing the penalty level. This greatly simplifies the problem since each of the components possesses much more tractable statistical properties compared to the original error term. It is shown in Appendix that, $a_{i,j}$ is independent over $i$, $\omega_{g,j}$ is weakly dependent over $t$, $e_{it,j}$ is independent conditional on $\{\gamma_t\}$, along with other useful properties. With the observation above, the infeasible regressor-specific penalty weights are proposed as follows:
\begin{align}
    \omega_{a,j}^2= \frac{1}{N}\sum_{i=1}^N a_{i,j}^2, \ \ \omega_{g,j}^2= \frac{N}{T^2} \sum_{b=1}^B\left(\sum_{t\in H_b} g_{t,j}\right)^2, \ \ \omega_{e,j}^2=\frac{1}{NT^2}\sum_{i=1}^N\left(\sum_{t=1}^T e_{it,j}\right)^2 \label{penalty_inf} 
\end{align}
where $B$ is the number of clusters/blocks, $h$ is the block length and $H_b$ is the associated index set. We note that $\omega_{a,j}$ can also be seen as a sample variance estimator (without centering around the sample mean). Furthermore, $\omega_{e,j}$ can be seen as a cluster variance estimator without estimation error, and $\omega_{g,j}$ can be seen as a correlated-cluster variance estimator (e.g., \citealp{bester2008inference}). It may seem natural to use their sample analogs as the feasible penalty weights, by estimating the components and then plugging in. However, it turns out that a direct sample analog of the infeasible penalty weights is not a good idea in certain settings. To see that, let's first focus on the estimation of the component without considering the estimation error in $V_{it}$, since the latter one is not the main source of the problem. For what follows, we consider the component estimators given $V_{it}$: let $v_{it,j}:= f_{it,j}V_{it}$ and define
\begin{align}
    \tilde{a}_{i,j} = \frac{1}{T}\sum_{t=1}^{T} v_{it,j}  , \ \ \ \tilde{g}_{t,j} =\frac{1}{N} \sum_{i=1}^{N} v_{it,j}  , \ \ \
    \tilde{e}_{it,j} = v_{it,j} -\tilde{a}_{i,j}-\tilde{g}_{t,j}. \label{compo_noerror}
\end{align}
Let $\tilde\omega_{a,j}^2$, $\tilde\omega_{g,j}^2$, and $\tilde\omega_{e,j}^2$ be defined as \ref{penalty_inf} with the components replaced by those in \ref{compo_noerror}, and let $\tilde{\omega}_j^2 = \tilde\omega_{a,j}^2 + \tilde\omega_{g,j}^2 + \tilde\omega_{e,j}^2$. We also define the variances $\sigma_{a,j}^2 = E[a_{i,j}^2]$, $\sigma_{g,j}^2 = E[g_{j,t}^2]$, $\sigma_{e,j}^2 = E[e_{it,j}^2]$ and the long-run variances $\Sigma_{g,j} = \sum_{l=-\infty}^\infty E[g_{t,j} g_{t+l,j}]$, $\Sigma_{e,j} = \sum_{l=-\infty}^\infty E[e_{it.j} e_{it+l,j}]$. It can be shown that\footnote{More details can be found in Lemma \ref{weight_valid_1} in Appendix}, when either $a_{i,j}$ or $g_{i,j}$ is non-degenerate, 
\begin{enumerate}
    \item  ${\omega}_{a,j}^2$ and $ \tilde{\omega}_{a,j}^2$ have the same probability limit $\sigma_{a,j}^2$; 
    \item ${\omega}_{g,j}^2$ and $\tilde{\omega}_{g,j}^2$ have the same probability limit $\Sigma_{g,j}$; 
    \item ${\omega}^2_{e,j}$ and $\tilde{\omega}^2_{e,j}$ have the same degenerate probability limit $0$.
\end{enumerate}
In this case, a direct sample analog of the infeasible one in \ref{penalty_inf} could be valid. However, when both $a_{i,j}$ and $g_{i,j}$ are degenerate, e.g. $W_{it}$ is i.i.d. over $i$ and $t$, then $T\tilde{\omega}_{a,j}$ and $T\tilde{\omega}_{g,j}$ are exactly zero, and it can be shown that 
\begin{align*}
    &T\omega_{e,j}^2 \overset{p}{\to}\Sigma_e = \sigma_e^2 \\
   & T(\tilde{\omega}_{a,j}^2 + \tilde{\omega}_{g,j}^2) \overset{p}{\to}  2 \sigma_e^2
\end{align*}
Here we rescale the penalty weights so that their asymptotic limits are not degenerate. This implies that the rescaled penalty weights would converge to the limit that is at least twice larger as the limit of the infeasible penalty weights, making $\tilde{\omega}_j$ too conservative under the degeneracy. Furthermore, we note that the sample analog of ${\omega}_{e,j}^2$ is mechanically downward biased:
\begin{align*}
    \tilde{\omega}_{e,j}^2:=&\frac{1}{NT^2}\sum_{i=1}^N\left(\sum_{t=1}^T \tilde{e}_{it,j}\right)^2 = \frac{1}{N}\sum_{i=1}^N\left(\frac{1}{T}\sum_{t=1}^T \left(v_{it,j} -\frac{1}{T}\sum_{t=1}^{T} v_{it,j} - \frac{1}{N} \sum_{i=1}^{N} v_{it,j} \right)\right)^2 \\
   =& \left(\frac{1}{NT} \sum_{i=1}^{N} \sum_{t=1}^Tv_{it,j} \right)^2  =o_P(1/T) \ \ \text{ under non-degeneracy}\\
   =&o_P(1/NT) \ \ \text{ under degeneracy}
\end{align*}
$ \tilde{\omega}_{e,j}^2$ is actually mechanically zero if we recenter the components when constructing the penalty weights. For similar reasons, the finite sample performance of  $\tilde{\omega}_{g,j}^2$ would not be satisfying when the size of each correlated cluster $H_b$ is chosen too large. 

Due to the aforementioned issues, I propose the following feasible penalty weights: Let $\hat{a}_{i,j}, \hat{g}_{t,j}$, and $\hat{e}_{it,j}$ be defined the same way as \ref{compo_noerror} with $v_{it,j}$ replaced by $\hat v_{it,j} = f_{it,j}\hat V_{it}$ and $\hat V_{it}$ is the residual from some initial estimate. Then, we define
\begin{align}
    \hat{\omega}_j^2 :=& \max(\hat{\omega}_{a,j}^2-\hat{\omega}_{e,j}^2,0) + \max(\hat{\omega}_{g,j}^2-\hat{\omega}_{e,j}^2,0) +  \hat{\omega}_{e,j}^2 \label{feasible_w} \\
    \hat{\omega}_{a,j}^2:= &  \frac{1}{N}\sum_{i=1}^{N} \hat{a}_{i,j}^2= \frac{1}{NT^2} \sum_{i=1}^N\left(\sum_{t=1}^{T} \hat{v}_{it,j}\right)^2 \label{feasible_a} \\
    \hat{\omega}_{g,j}^2:=& \frac{N}{T^2} \sum_{t=1}^T \sum_{s=1}^T k\left(\frac{|t-s|}{M}\right) \hat{g}_{t,j}\hat{g}_{s,j} =\frac{1}{NT^2} \sum_{t=1}^T \sum_{s=1}^T k\left(\frac{|t-s|}{M}\right) \left(\sum_{i=1}^N \hat{v}_{it,j}\right)\left(\sum_{i=1}^N \hat{v}_{is,j}\right) \label{feasible_g} \\
    \hat{\omega}_{e,j}^2:= &  \frac{1}{NT^2}\sum_{i=1}^N\sum_{t=1}^T \sum_{s=1}^T k\left(\frac{|t-s|}{M}\right) \hat{e}_{it,j}\hat{e}_{is,j} \label{feasible_e}
\end{align}
where $ k\left(\frac{|t-s|}{M}\right)$ is a Bartlett kernel and $M$ is a bandwidth with a requirement $M=o(T^{1/4})$. In \cite{menzel2021bootstrap} and \cite{hounyo2025}, similar adjustments are used to determine presence of the components in order for bootstrapping the components. Here these are used as self-normalizers adaptive to both non-degeneracy and degeneracy scenarios. We also note that $\hat{\omega}_{a,j}^2$ has been used as the feasible penalty weights in the cluster-LASSO method of \cite{belloni2016inference} but they don't need the extra subtraction adjustment because $\hat{\omega}_{a,j}^2$ itself is in some sense adaptive under one-way cluster dependence, as hinted by Lemma \ref{weight_valid_1} in Appendix. Additionally, we observe that $\hat{\omega}_{g,j}^2$ is effectively the Driscoll-Kraay variance estimator which accounts for cross-sectional cluster dependence and weak serial dependence, and $\hat{\omega}_{e,j}^2$ is the "average of HACs" estimator that accounts for the serial correlation remaining in $e$. 

Since the error term $v$ is unobserved, we need to first obtain the residual $\hat{v}$ from some initial estimation and then iterate the estimation. The detailed implementation is given as Algorithm 1 in \hyperref[app_0]{Appendix}. The validity of this proposal relies on whether the feasible penalty weights under the iterative estimation converges to the same limit of the infeasible penalty weights. We establish the validity formally in Appendix. For the main theorem below, we maintain a high-level assumption on the feasible penalty weights following \citealp{belloni2012sparse,belloni2016inference}: Let $\widehat{\omega}$ be the feasible diagonal weights. Suppose there exists $0<1/c_1<l\leq 1$ and $1\leq u<\infty$ such that $l\to 1$ and 
\begin{align}
    l \omega_{j} \leq \widehat{\omega}_{j} \leq u \omega_{j}, {\rm \ uniformly \ over \ j=1,...,p}, \label{omegacon}
\end{align}
where $\{\omega_j\}$ and $\{\widehat{\omega}_{j}\}$ are diagonal entries of $\omega$ and $\widehat{\omega}$, respectively.  

Given this choice of penalty weights above, the common penalty level $\lambda$ is pinned down as
\begin{align}
    \lambda & = 2C_\lambda\sqrt{N}T\Phi^{-1}\left(1-\frac{\gamma}{2p}\right), \label{penaltylevel} 
\end{align}
where $C_\lambda = c_1c_\omega>1$; $\gamma$ is a small order sequence. Other than $C_\lambda$ and $\gamma$, there are no unknown tuning parameters in this weighted LASSO estimation. In practice, $C_\lambda$ is chosen as a constant close to 1 and $\gamma = \alpha/\log( N \vee T)$ with $\alpha$ taken as a significance level\footnote{For example, $\alpha=0.1$ is more liberal than $\alpha=0.05$. The choices of $C_\lambda$ and $\gamma$ around the proposed practical choices do not matter much in finite sample.}. The order of $\gamma$ affects the convergence rate of the LASSO estimator: the theory only requires $\gamma=o(1)$ for LASSO to be consistent, but a faster rate of decay of $\gamma$ will result in a slower convergence rate of LASSO.

The key identification condition is given as follows. In the low-dimensional case, the identification in a linear regression is given by the non-singularity of ${\rm E}[f_{it}' f_{it}]$, which implies its empirical counterpart $\mathbb{E}_{NT} [f_{it}'f_{it}]$ is non-singular with high probability. However, as the dimension of $f_{it}$ grows larger than the sample size, $\mathbb{E}_{NT} [f_{it}'f_{it}]$ is singular almost surely. Fortunately, it turns out that under sparsity and $L1$ regularization, we only need certain sub-matrices to be well-behaved for identification. Define
  \begin{align*}
        \phi_{\rm min}(m)(M_f):= \min_{\delta\in \Delta(m)}\delta'M_f\delta \ {\rm and} \ \phi_{\rm max}(Cs)(M_f) := \max_{\delta\in \Delta(m)}\delta'M_f\delta,
  \end{align*}
where $\Delta(m)=\{\delta: \Vert\delta\Vert_0 = m, \Vert\delta\Vert_2=1\}$ and $ M_f = \mathbb{E}_{NT} [f_{it}'f_{it}]$.

\begin{assumption_se}[Sparse Eigenvalues]
    \label{sparse_eigen} For any $C>0$, there exists constants $0<\kappa_1<\kappa_2<\infty$ such that with probability approaching one, as $(N,T)\to\infty$ jointly, $\kappa_1\leq \phi_{\rm min}(Cs)(M_f)<\phi_{\rm max}(Cs)(M_f) \leq \kappa_2$. 
\end{assumption_se}
The sparse eigenvalue assumption follows from \cite{belloni2012sparse}. It implies a restricted eigenvalue condition, which represents a modulus of continuity between the prediction norm and the norm of $\delta$ within a restricted set. More primitive sufficient conditions are discussed in \cite{bickel2009simultaneous} and \cite{belloni2012sparse}.

\begin{assumption_reg}[Regularity Conditions]
    \label{regcon} 
    (i) $\log (p/\gamma) = o\left(T^{1/6}/(\log T)^2\right)$. (ii) For some $\mu>1,\delta>0$, $\max_{j\leq p}\emph{E}[\lvert f_{it,j} \rvert^{8(\mu+\delta)}]<\infty$, $\emph{E}[\lvert V_{it} \rvert^{8\mu+\delta)}]<\infty$. (iii) $\emph{E}\left[\left(\sum_{t=1}^Te_{it,j}\right)^2|\{\gamma_t\}_{t=1}^T\right] >0$ almost surely. (iv) For each $j$, either (1) $\emph{E}\left(a_{i,j}^2\right) + \emph{E}\left(g_{t,j}^2\right) >\epsilon$ for some $\epsilon>0$ or (2) $f_{it,j}V_{it}$ is i.i.d over $i,t$. (v) $1\le \max_{j\leq p} \omega_j / \min_{j\leq p} \omega_j  =O(1) $.
\end{assumption_reg}

Assumption \hyperref[regcon]{REG}(i) imposes a restriction on the dimension of $f_{it}$, $p$, while allowing it to be greater than the sample size. The moment conditions in Assumption \hyperref[regcon]{REG}(ii) are common in the literature. \hyperref[regcon]{REG}(iii) is allows for both non-degeneracy and i.i.d case (degeneracy)\footnote{Here the i.i.d is taken as an special case of degeneracy. 
By Hoeffding decomposition, $e_{it}$ can be further decomposed into a term featuring $E[f_{it}V_{it}|\alpha_i,\gamma_t]$ and a residual term. This is mainly relevant when the both $a_i$ and $g_t$ are degenerate and $E[f_{it}V_{it}|\alpha_i,\gamma_t]$ is not, referred to as the non-Gaussian degenerate case discussed in \cite{menzel2021bootstrap}. Also see recently \cite{davezies2025analytic} and \cite{hounyo2025}. This case still remains as a challenge and is a ongoing topic of research in the literature. We will focus on the Gaussian non-degenerate and degenerate cases.}. 

A common way to mitigate the shrinkage bias of LASSO is to apply least square estimation based on the selected model by LASSO, which is named Post-LASSO. We denote the index set of selected regressors as $\widehat{\Gamma} = \{j\in {1,...,p}: |\widehat{\zeta}_j|>0\}$ where $\widehat{\zeta}_j$ are two-way LASSO estimates. The next theorem gives convergence rate results for both two-way cluster-LASSO and its associated Post-LASSO.

\begin{theorem} \label{performance_bound}
Suppose Assumptions \hyperref[ahk]{AHK}, \hyperref[asm]{ASM}, \hyperref[regcon]{REG} hold for model \ref{model_lasso} as $N,T\to\infty$ jointly with $N/T\to c$. Then, by setting $\lambda$ as \ref{penaltylevel} and $\omega_j$ as \ref{penaltyweights} and \ref{penalty_inf}, we have (i) the event \ref{regevent} happens with probability approaching one. Additionally, suppose that Assumption \hyperref[sparse_eigen]{SE} holds and $\widehat{\omega}$ satisfies condition \ref{omegacon}. Let $\widehat{\zeta}$ be the two-way cluster-LASSO estimator or the post-LASSO estimator based on the two-way cluster-LASSO selection. Then, (ii) $ \Vert \widehat{\zeta}\Vert_0= O_P(s)$, and (iii) $ \frac{1}{NT} \sum_{i=1}^N\sum_{t=1}^T\left(f_{it}'\widehat{\zeta}- f_{it}'\zeta_0\right)^2 =  O_P\left( \frac{s\log(p/\gamma)}{l_{NT}}\right)$,  $ \left\Vert \widehat{\zeta} - \zeta_0 \right\Vert_1 =       O_P\left(s\sqrt{\frac{\log (p/\gamma)}{l_{NT}} }\right)$, and $\left\Vert \widehat{\zeta} - \zeta_0 \right\Vert_2 =  O_P\left(\sqrt{\frac{s\log (p/\gamma)}{l_{NT}} }\right)$, where $l_{NT} = N\wedge T$ for the case (1) and $l_{NT} = NT$ for the case (2)  of Assumption \hyperref[regcon]{REG}(iv).
\end{theorem}

Theorem \ref{performance_bound} establishes convergence rates in terms of the prediction, $l1$, and $l2$ norms for the (post) two-way cluster-LASSO estimator in an approximately sparse model. These results are the first that give convergence rates for a LASSO-based estimator allowing for two-way cluster dependence. It is shown that under the two-way cluster dependence, the two-way cluster-LASSO is consistent but, under the non-degenerate case, has a convergence rate slower than those of LASSO-based methods under the random sampling condition or weak dependence. Without loss of generality, let $N = N \wedge T$, then by choosing $\gamma$ according to $\log(1/\gamma) \simeq \log(p\vee N)$, we have $\left\Vert \widehat{\zeta} - \zeta_0 \right\Vert_2= O_P\left(\sqrt{\frac{s\log (p\vee N)}{N} }\right)$. As a comparison, the rate of convergence in terms of the $l2$ norm is $O_P\left(\sqrt{\frac{s\log p }{NT} }\right)$ under the random sampling and the homoskedasticity Gaussian error assumptions in \cite{bickel2009simultaneous} or the heteroskedasticity Gaussian error in Theorem 19.3 of \cite{hansen2022econometrics}, $O_P\left(\sqrt{\frac{s\log (p\vee N) }{NT} }\right)$ under random sampling in \cite{belloni2012sparse}, and $O_P\left(\sqrt{\frac{s\log (p\vee N) }{N I_T} }\right)$ under cross-sectional independence in \cite{belloni2016inference} where the information index $I_T = 1$ when there is perfect dependence within the cross-sectional cluster.

As illustrated in the Introduction, the slow rate of convergence is due to the underlying factor structure. It is unclear if valid inference is possible under the rate of convergence results in Theorem \ref{performance_bound} or if it is possible to relax the requirement through a cross-fitting procedure. These questions are addressed in the next section.

\section{Clustered-Panel Cross-Fitting and Inference}\label{sub_sample}
In this section, I will propose a cross-fitting scheme suitable for clustered panel data. The idea of sample splitting is to split the sample in a proper way and use the sub-samples separately to estimate nuisance parameters and main parameters of interest. If the sub-samples are independent of each other, then the first-step estimation of the nuisance parameters will be independent of the sample used for the second-step estimation. With this property, the error term that causes the bias can vanish with a less stringent rate requirement on the nuisance estimation. Intuitively, the dependence between the two steps is eliminated so that a potentially over-fitted nuisance estimate from the first step does not pollute the second step as much as it would otherwise do. 

Therefore, the goal of the cross-fitting scheme is to split the sample in a proper way so that the two resulting sub-samples are independent or, at least, \textquotedblleft approximately\textquotedblright \ independent. Under the AHK characterization in Assumption \hyperref[ahk]{AHK}, $W_{it}$ are cluster-dependent over both cross-section and time. Importantly, the cluster dependence does not vanish as the distance between observations (if there is any ordering) increases. If $\gamma_t$ is weakly dependent, which is the focus of this paper, then the dependence between observations that don't share the same cluster in either dimension dies out as the temporal distance grows. In that case, intuitively, one can split the sample so that the sub-samples do not share the same cluster and are far apart in temporal distance. This is exactly how this scheme works: 

\begin{definition}[Two-Way Clustered-Panel Cross-Fitting]          \ \ \
\label{panel_CV_def}
     \begin{enumerate}
         \item[(i)] Select some positive integers $(K,L)$. Randomly partition the cross-sectional index set $\{1,2,...,N\}$ into $K$ folds $\{I_1,I_2,...,I_K\}$ and partition the temporal index set $\{1,2,...,T\}$ into $L$ adjacent folds $\{S_1,S_2,...,S_L\}$ so that $\bigcup_{k=1}^K I_k = \{1,...,N\}$, $\bigcup_{l=1}^L S_l = \{1,...,T\}$\footnote{For 
simplicity, I assume $N$ and $T$ are divisible by $K$ and $L$, respectively. In practice, if $N$ is not divisible by $K$, the size for each cross-sectional block can be chosen differently with some length equal to ${\rm floor}(N/K)$ and others equal to ${\rm ceil}(N/K)$. and the same applies to the temporal dimension.}. 
        \item[(ii)] For each $k=1,...,K$ and $l=1,..,L$, construct the main sample $ W(k,l)=\{W_{it}: {i\in I_k, t\in{S_l}}\}$ and the auxiliary sample $ W(-k,-l) = \left\{W_{it}: {i\in \bigcup_{k'\ne k}I_{k'}, t\in \bigcup_{l'\ne l, l\pm 1}S_{l'}}\right\}$.
     \end{enumerate}
\end{definition}

Later on, we also use $I_{-k}$ and $S_{-l}$ to denote the index sets for the auxiliary sample $W(-k,-l)$. Similarly, we denote $N_{-k}$ and $T_{-l}$ as the cross-sectional and temporal sample sizes for the auxiliary sample $W(-k,-l)$. Figure 1 illustrates the cross-fitting with $K=4$ and $L=8$.

 \begin{center}
    \includegraphics[scale=0.25]{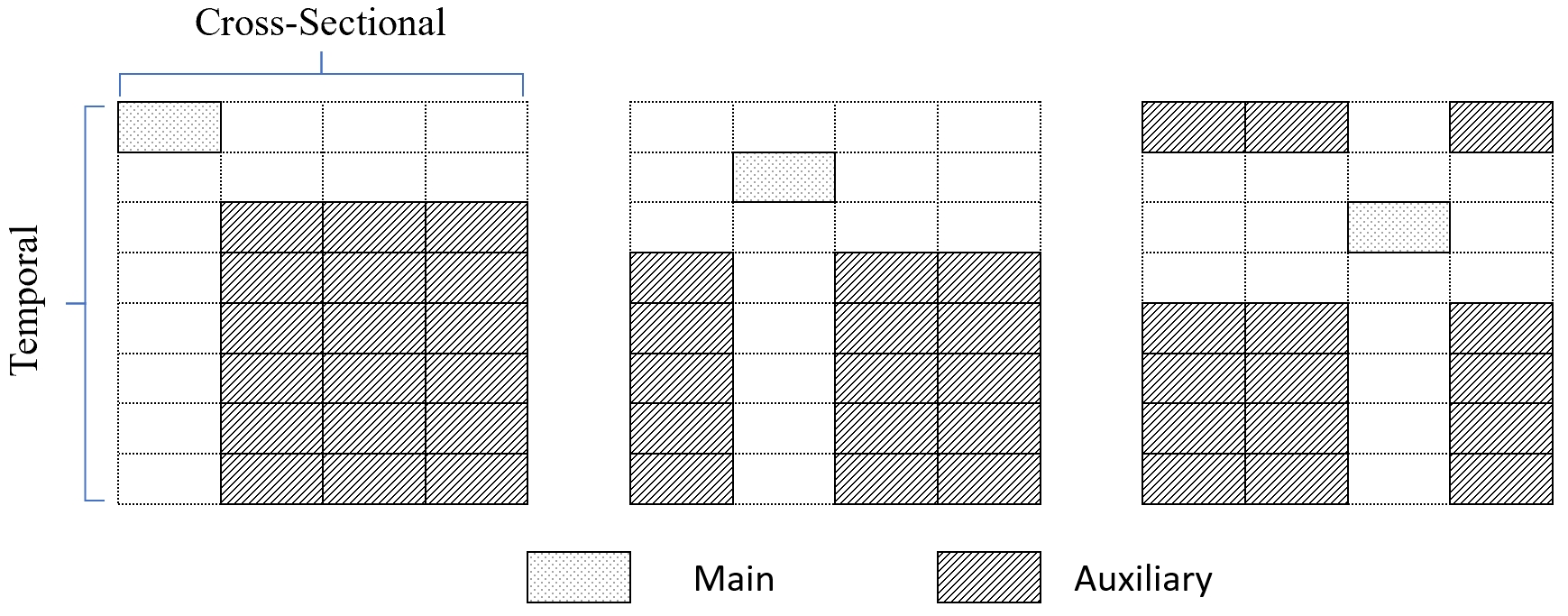}
\end{center}
\vspace{-4mm}
{\footnotesize \textbf{Figure 1}: Clustered-Panel cross-fitting with $K=4$ and $L=8$. Three graphs from left to right correspond to the main and auxiliary sample constructions with $(k,l)=(1,1)$, $(k,l)=(2,2)$, $(k,l)=(3,3)$. For a simple illustration, observations in the main sample are all adjacent in the cross-sectional dimension but it is not necessary in practice; the same applies to the auxiliary sample.} 

Since the sub-samples $W(k,l)$ and $W(-k,-l)$ do not share any cluster, they are free from cluster dependence, and what's left is the weak dependence over time. Unless imposing $m-$dependence, the sub-samples above will not be independent. However, under certain regularity conditions regarding the weak dependence, it can be shown through the coupling technique that as long as the temporal distance between the sub-samples diverges at a certain rate, there exist coupling sub-samples that are independent of each other while having the same marginal distributions as the constructed sub-samples with probability converging to 1. Such a result is provided in Lemma \ref{indep_couple}.

\begin{lemma}[Independent Coupling] \label{indep_couple}
Consider the sub-samples $W(k,l)$ and $W(-k,-l)$ for $k=1,...,K$ and $l=1,...,L$. Suppose Assumption \hyperref[ahk]{AHK} holds and $log(N)/T = o(1)$ as $T\to\infty$. Then, we can construct $\tilde{W}(k,l)$ and $\tilde{W}(-k,-l)$ such that: (i) they are independent of each other; (ii) have the same marginal distribution as $W(k,l)$ and $W(-k,-l)$, respectively; (iii) 
\begin{align*}
    {\rm P}\left\{\left(W(k,l),W(-k,-l)\right)\ne \left(\tilde{W}(k,l),\tilde{W}(-k,-l)\right), \ {\rm for \ some} \ (k,l) \right\} = o(1).
\end{align*}
\end{lemma}
The proof of Lemma \ref{indep_couple} is given in \hyperref[app_0]{Appendix}. Lemma \ref{indep_couple} shows that the main and auxiliary samples from the proposed clustered-panel cross-fitting scheme are approximately independent as $N,T$ diverge. Note that the hypothetical samples $\tilde{W}(k,l)$ and $\tilde{W}(-k,-l)$ only serve as technical tools and do not matter in practice. The proof of Lemma \ref{indep_couple} is based on independence coupling results introduced in \cite{Semenova2023}.

It turns out such properties of the cross-fitting scheme are very useful for inference in a high-dimensional panel model. Such procedure is used to remove the dependence between the first and the second steps of estimation so as to relax the rate requirement for the first-step estimation of the high-dimensional nuisance parameters. For what follows, I will define an inferential procedure for high-dimensional panel with a two-step estimator and cluster-robust variance estimator using the proposed cross-fitting scheme in the context of a semi-parametric moment restriction model, extending the DML approach in \cite{Chernozhukov2018} to a panel setting. 

Let $\varphi(W_{it};\theta,\eta)$ be some identifying moment functions where $\theta$ is a low-dimensional vector of parameters of interest and $\eta$ are nuisance functions. For example, $\eta = g_0$ in \ref{example}. Let $\psi(W_{it};\theta,\eta)$ be the orthogonalized moment function with the following properties:
\begin{align}
    {\rm E}[\psi(W_{it};\theta_0,\eta_0)] &= 0, \label{neyman1} \\
    \partial_{r}{\rm E}\left[\psi\left(W_{it};\theta_0,\eta_0+r(\eta-\eta_0)\right)\right]|_{r=0}&=0. \label{neyman2}
\end{align}
Essentially, $\psi(W_{it};\theta,\eta)$ is adjusted for the fact that $\eta_0$ needs to be estimated and, as a result, the nuisance functions have no first-order effect locally on the orthogonalized moment conditions, based on which the estimation of $\theta_0$ is therefore robust to the plug-in of noisy estimates of $\gamma_0$. In contrast, the original identifying moment conditions do not possess such a property. In model \ref{example}, $\varphi(W_{it};\theta,\eta) = D_{it} U_{it}$ and   $\psi(W_{it};\theta,\eta) = \left(D_{it} - {\rm E}[D_{it}|X_{it},c_i,d_t]\right) \left(Y_{it}-D_{it}\theta - g(X_{it},c_i,d_t) \right)$. In the treatment effect model with unconfoundedness conditional on covariates and unobserved heterogeneous effects, $\varphi(W_{it};\theta,\eta) =   {\rm E}[Y_{it}|D_{it}=1,X_{it},c_i,d_t]-{\rm E}[Y_{it}|D_{it}=0,X_{it},c_i,d_t] - \theta^{\rm ATE}$ and $\psi(W_{it};\theta,\eta)$ is the moment function corresponding to the well-known augmented inverse probability weighting estimator, which is doubly robust. 

Due to the presence of unit and time components, I consider a two-way cluster robust variance estimator similar to \cite{CHS_Restat} (CHS estimator) with adjustment due to cross-fitting. The variance estimator is motivated under arbitrary dependence in panel data and is shown to be robust to two-way clustering with correlated time effects in linear panel models. \cite{chen2024fixed} show that such variance estimator can be written as an affine combination of three well-known robust variance estimators: Liang-Zeger-Arellano estimator, Driscoll-Kraay estimator, and the "average of HACs" estimator, and they propose a two-term variant named as, DKA estimator, for better finite sample performance\footnote{\cite{chen2024fixed} also consider fixed-b asymptotic approximation and bias-correction as a function of the bandwidth $M$. It is not considered in this paper since the statistical property of such adjustment is not clear in this setting.}. We define the CHS and DKA variance estimator with cross-fitting adjustment as follows: 
\begin{alignat}{2}
     \widehat{V}_{\rm CHS} &= \widehat{\bar{A}}^{-1}\widehat{\Omega}_{\rm CHS} \widehat{\bar{A}}^{-1'}, \ \widehat{\Omega}_{\rm CHS} &&= \widehat{\Omega}_{\rm A} + \widehat{\Omega}_{\rm DK} - \widehat{\Omega}_{\rm NW}, \label{chs_cf} \\
    \widehat{V}_{\rm DKA} &= \widehat{\bar{A}}^{-1}\widehat{\Omega}_{\rm DKA} \widehat{\bar{A}}^{-1'}, \ \widehat{\Omega}_{\rm DKA} &&= \widehat{\Omega}_{\rm A} + \widehat{\Omega}_{\rm DK} \label{dka_cf}
\end{alignat}
where $ \widehat{\bar{A}}:= \frac{1}{KL} \sum_{k=1}^K\sum_{l=1}^L \frac{1}{N_k T_l}  \sum_{i\in I_k,s\in S_l}\psi^a({W_{it};\widehat{\eta}_{kl}})$ and, with $k\left(  \frac{m}{M_l}\right):= 1-\frac{m}{M_l}$ for $m=0,1,...,M_l-1$ and 0 otherwise, for some $M_l\in [1,T_l]$,
\begin{align*}
{\widehat{\Omega}}_{\rm A}:=  &  \frac{1}{KL}\sum\limits_{k=1}^{K}%
\sum\limits_{l=1}^{L} \frac{1}{N_k T_l^2}\sum\limits_{i\in I_k, t\in S_l, r\in S_l} \psi(W_{it};\widehat{\theta},\widehat{\eta}_{kl})%
\psi(W_{ir};\widehat{\theta},\widehat{\eta}_{kl})'  , \\
{\widehat{\Omega}}_{\rm DK}:=  & \frac{1}{KL}\sum\limits_{k=1}^{K}%
\sum\limits_{l=1}^{L} \frac{K/L}{N_k T_l^2} \sum\limits_{t\in S_l, r\in S_l}k\left(
\frac{\left\vert t-r\right\vert }{M}\right)   \sum\limits_{i\in I_k, j\in I_k} \psi(W_{it};\widehat{\theta},\widehat{\eta}_{kl})\psi_{jr}(\widehat{\theta},\widehat{\eta}_{kl})'  , \\
{\widehat{\Omega}}_{\rm NW}:= &  \frac{1}{KL}\sum\limits_{k=1}^{K}
\sum\limits_{l=1}^{L} \frac{K/L}{N_k T_l^2} \sum\limits_{i\in I_k, t\in S_l, r\in S_l}k\left(  \frac{\left\vert t-r\right\vert }{M}\right)
\psi(W_{it};\widehat{\theta},\widehat{\eta}_{kl})\psi(W_{ir};\widehat{\theta},\widehat{\eta}_{kl})'.
\end{align*}

The next definition summarizes the panel DML estimation and inference procedures for a semiparametric moment restriction model: 
\begin{definition}[Panel DML Algorithm]          \ \ \
\label{panel_DML_def}
     \begin{enumerate}
         \item[(i)] Given the identifying moment functions $\varphi(W;\theta,\eta)$ such that ${\rm E}_[\varphi(W;\theta_0,\eta_0)]=0$, find the orthogonalized moment function $\psi(W,\theta, \eta)$.
         \item[(ii)] Obtain cross-fitting sub-samples $W(k,l)$ and $W(-k,-l)$ as in Definition \ref{panel_CV_def}. 
         \item[(iii)] For each $k$ and $l$, use the sample $W(-k,-l)$ for the first-step estimation and obtain $\widehat{\eta}_{kl}$, then construct $\overline{\psi}_{kl}\left(\theta\right)=\mathbb{E}_{kl}[\psi(W_{it};\theta,\widehat{\eta}_{kl})]$ using $W(k,l)$ for each $(k,l)$, where $\mathbb{E}_{kl}[.]$ is defined in \hyperref[notation_sec]{Notation}. 
         \item[(iv)] Obtain the DML estimator $\widehat{\theta}$ as the solution to 
            \begin{align}
                \frac{1}{KL}\sum_{k=1}^K \sum_{l=1}^L \overline{\psi}_{kl}\left(\theta\right)= 0. \label{dml_exact}
            \end{align}
        \item[(v)] Conduct inference based on the asymptotic normality of $\hat{\theta}$ and the consistent variance estimators defined in \ref{chs_cf} or \ref{dka_cf}.
     \end{enumerate}
\end{definition}

\begin{remark}[The Choice of $K$ and $L$]
\label{remart1}
    Notice there is a trade-off in setting $(K,L)$ between the first step and second step accuracy: the bigger values of $(K,L)$, the bigger sample size of the auxiliary sample $W(-k,-l)$, which is beneficial for high-dimensional first-steps but at the cost of a noisier parametric second step. Due to leaving out the temporal neighborhood, it necessitates an $L\geq 4$ for feasible implementation (if $L=3$, for example, any main sample $W(k,l)$ with $l=2$ does not have a well-defined auxiliary sample). On the other hand, it is computationally costly to set the values of $(K,L)$ too large. In practice, $K=2 \ to \ 4$ and $L=4 \ to \ 8$ work well in simulations. 
\end{remark}

A formal study of the estimation and inference procedure given in Algorithm \ref{panel_DML_def} is given in the Online Supplementary Material. It is shown that the two-step panel-DML estimator of the fixed-dimensional parameter $\theta_0$ is asymptotically normal and the CHS- and DKA-type variance estimators are consistent for the asymptotic variance, under a non-degenerate condition of the underlying components. However, due to the cross-fitting adjustment, there is a subtle issue when the underlying components are degenerate: Since $(K,L)$ are assumed to be fixed, and the sample sizes in each block diverge, the asymptotic approximation of the variance estimators happens in the sub-sample level. If the components that drives the cluster dependence are degenerate, the rate of convergence would change, but the change of scaling for correct asymptotics of the variance estimator may not match with that of the DML estimator. It can be further shown that the proposed variance estimator would be conservative in the i.i.d case: In this case, the asymptotic variance is not degenerate when the DML estimator is scaled by $\sqrt{NT}$. The variance estimator in (\ref{chs_cf}) is designed for the asymptotic variance of the DML estimator scaled by $\sqrt{N}$. Therefore, to be adaptive to the i.i.d case, we need $T\hat{V}_{\rm CHS}$ to be consistent for the asymptotic variance. However, the correct asymptotics requires an adjustment less than $T$, which implies that current specification would be too conservative. This problem is not unique for the method in this paper but rather an intrinsic issue when cross-fitting is implemented in more than one-dimension. For example, \cite{Chiang2022} imposes a non-degenerate condition for their main results. 

Instead of analytical variance estimators, one could consider resampling method. However, there is one more subtle issue when unobserved heterogeneities are present in the panel model, as revealed in the next section. For the latter issue, it is dealt by a strong assumption on the unobserved heterogeneities or DML without cross-fitting in a partial linear model as a special case. As an implication of these two subtle issues, approaches and sufficient conditions that allow for DML estimation and inference in general settings  without cross-fitting would be favorable, and it is the focus of the ongoing research.

\section{Partial Linear Model with Unobserved Heterogeneity} 
\label{partial_linear_unob}
In this section, we study the partial linear model with unobserved heterogeneous effects as \ref{example} as a special case, and an excludable instrumental variable $Z_{it}$ is considered for identification: 
\begin{align}
    Y_{it} = & D_{it}\theta_0 + g(X_{it},c_i,d_t) + U_{it}, \ \ {\rm E}[U_{it}|{X}_{it},c_i,d_t]= 0,  \ E[Z_{it}U_{it}] = 0. \label{cremodel} 
\end{align}
For clearer presentation, $D_{it}$ is treated as a scalar variable. If the lags or leads of $D_{it}$ are exogenous, they can also be included in $X_{it}$. Doing so would not change the theory for estimation and inference, but could change the interpretation of $\theta_0$.

To deal with the unobserved heterogeneous effects that cause endogeneity, 
I take a correlated random-effects approach through the generalized Mundlak device:
\begin{assumption_gmd}[Generalized Mundlak Device] \label{crecon} For each $i=1,...,N$ and $t=1,...,T$, 
    \begin{align}
        c_i &= h_c(\bar{F}_i,\epsilon^c_i), \label{ci} \\
        d_t &= h_d(\bar{F}_t,\epsilon^d_t), \label{dt}
    \end{align}
where $\bar{F}_i = \frac{1}{T}\sum_{t=1}^T F_{it}$, $\bar{F}_{t} = \frac{1}{N}\sum_{i=1}^N F_{it}$, $F_{it}:= (D_{it},X_{it}')'$; $h_c$ and $h_d$ are some unknown measurable functions; $(\epsilon^c_i,\epsilon^d_t)$ are independent shocks; and $(c_i,d_t)$ are independent of $U_{it}$.
\end{assumption_gmd}

To justify this assumption, we shall recall the idea of the conventional Mundlak device. To explicitly model the correlation between the random effects and the covariates, \cite{mundlak1978pooling} proposes an auxiliary regression between the random effects and the cross-sectional sample average, and shows that if the random effects enter the model linearly then the resulting GLS estimator is equivalent to the common within-estimator. \cite{wooldridge2021two} further shows that the equivalence relations exist among the POLS estimators resulting from the Mundlak device, within-transformation, and the fixed-effects dummies. Therefore, if aforementioned approaches are sensible for dealing with unobserved heterogeneity, then allowing the Mundlak device to have a more flexible function form should also be reasonable and more robust. Compared to similar ideas implemented in \cite{wooldridge2020inference} and \cite{clarke2025double}, the approach here is more flexible, while, as a tradeoff, relying more on the sparsity condition on the nuisance functions.

Before utilizing the proposed panel DML approach for inference, we note that there is a subtle issue: the Mundlak device uses the full history of the covariates which potentially generates dependence across the cross-fitting sub-samples. Similar issues also appear in a simple linear panel model with additive unobserved effects where within-transformation also introduces sample-averages. 
Therefore, the cross-fitting may not be compatible with the presence of  unobserved heterogeneity. Without cross-fitting, it is challenging to establish an inferential theory with growing dimensionality in unknown parameters in general. Nevertheless, as is shown below, it is possible to establishing the asymptotic normality of the panel DML estimator using the full sample with a strengthened sparsity condition.

Under model \ref{cremodel}, $g(X_{it},c_i,d_t) = {\rm E}[Y_{it} - D_{it}\theta_0 | X_{it},c_i,d_t]$. We can rewrite \ref{cremodel} as follows:
\begin{align*}
    Y_{it} = \left(D_{it} - g_D(X_{it},c_i,d_t) \right)\theta_0 + g_Y(X_{it},c_i,d_t) + U_{it}. 
\end{align*}
where $g_D(X_{it},c_i,d_t):= {\rm E}[ D_{it} | X_{it},c_i,d_t] $ and $g_Y(X_{it},c_i,d_t):= {\rm E}[ Y_{it} | X_{it},c_i,d_t] $. Under Assumption \hyperref[crecon]{GMD}, $g_D(X_{it},c_i,d_t)$ and $g_Y(X_{it},c_i,d_t)$ can be rewritten as compound functions, which are assumed to be well-approximated by a linear combination of a $\tau-$th order polynomial transformation $L^{\tau}$ as follows:
\begin{align}
    g_D^{*}(X_{it},\bar{F}_i,\epsilon^c_i,\bar{F}_t,\epsilon^d_t)& := g_D(X_{it},h_c(\bar{F}_i,\epsilon^c_i),h_d(\bar{F}_t,\epsilon^d_t))=  L^{\tau} \left(X_{it},\bar{F}_{i}, \bar{F}_{t},\epsilon_i^c,\epsilon_t^d\right) \eta_D + r^D_{it} \label{sparse_D} \\
    g_Y^{*}(X_{it},\bar{F}_i,\epsilon^c_i,\bar{F}_t,\epsilon^d_t)& := g_Y(X_{it},h_c(\bar{F}_i,\epsilon^c_i),h_d(\bar{F}_t,\epsilon^d_t))=  L^{\tau} \left(X_{it},\bar{F}_{i}, \bar{F}_{t},\epsilon_i^c,\epsilon_t^d\right) \eta_Y + r^Y_{it} \label{sparse_Y}
\end{align}
where $(\eta_D,\eta_Y)$ are slope coefficients and $(r_{it}^D,r_{it}^Y)$ are the approximation errors. Furthermore, we can define a vector of transformed regressors as $L_{1,it} =L^{\tau}(X_{it},\bar{F}_{i}, \bar{F}_{t})$ and a vector of unobserved regressors as $L_{2,it} = L^{\tau}(X_{it},\bar{F}_{i}, \bar{F}_{t},\epsilon_i^c,\epsilon_t^d) \backslash L^{\tau}(X_{it},\bar{F}_{i}, \bar{F}_{t}) $. Let $(\eta_{D,1},\eta_{D,2})$ be such that
\begin{align*}
    L^{\tau}\left(X_{it},\bar{F}_{i},\bar{F}_{t},\epsilon_i^c,\epsilon_t^d\right) \eta_D = L_{1,it}\eta_{D,1} + L_{2,it}\eta_{D,2}.
\end{align*}
And $(\eta_{Y,1},\eta_{Y,2})$ are defined in the same way. Under the sparse approximation and Assumption \hyperref[crecon]{GMD}, we can rewrite model \ref{cremodel} as follows:
\begin{align*}
    Y_{it} =  \left(D_{it} - L_{1,it}\eta_{D,1} - L_{2,it}\eta_{D,2} - r^D_{it} \right)\theta_0 +L_{1,it}\eta_{Y,1} + L_{2,it}\eta_{Y,2} + r^Y_{it} + U_{it}.
\end{align*}
By defining a new error term $V_{it}^g := \left( L_{2,it} -{\rm E} [L_{2,it}]\right)\left(\eta_{Y,2}-\eta_{D,2}\theta_0\right)+ U_{it}$, a new approximation error $r_{it} = r_{it}^Y + r_{it}^D\theta_0$, the vector of observables $f_{it}:= \left(L_{1,it},1\right)$ with dimension denoted by $p$, and the nuisance vectors $\beta_0 := \left(\eta_{Y,1}, {\rm E} [L_{2,it}]\eta_{Y,2}\right)$, $\pi_0 := \left(\eta_{D,1}, {\rm E} [L_{2,it}]\eta_{D,2}\right)$, we can rewrite the model above as
\begin{align}
     Y_{it} = \left(D_{it} - f_{it}'\pi_0 \right)\theta_0 + f_{it}'\beta_0 + r_{it} + V_{it}^g. \label{cremodel_linear}
\end{align}
Noticeably, in this case, the parameters associated with the unobservables $L_{2,it}$ can be arbitrarily non-sparse. 

Given ${\rm E}[Z_{it}U_{it}]$ and the independence between $Z_{it}$ and $(\epsilon_i^c,\epsilon_t^d)$,  we have the identifying moment condition ${\rm E}[Z_{it}V_{it}^g] = 0$. Let $\zeta_0$ be the linear projection parameter of $Z_{it}$ onto $f_{it}$ and let $V^Z_{it}$ be the corresponding linear projection errors. By Eq. (2.18) of \cite{Chernozhukov2018}, the near-Neyman orthogonal moment function is given by:
\begin{align}
   \psi_{it}(\theta_0,\eta_0):= \left(Z_{it} - f_{it}'\zeta_0 \right)\left( Y_{it}-f_{it}'\beta_0 - \left(D_{it}-f_{it}'\pi_0\right)\theta_0   \right)  . \label{cre_neyman_feasible}
\end{align}
where we denote $\eta_0=(\zeta_0,\beta_0,\pi_0)$. Under the sparse approximation, we can also rewrite the conditional expectation models for $Y$ and $D$ as 
\begin{align*}
    Y_{it} &= {\rm E}[Y_{it}| X_{it},c_i,d_t] +U^Y_{it} = f_{it}'\beta_0+r^Y_{it}+V^Y_{it} \\
    D_{it} &= {\rm E}[Y_{it} |X_{it},c_i,d_t] +U^D_{it} = f_{it}'\pi_0+r^D_{it}+V^D_{it} .
\end{align*}
where ${V}^Y_{it} = \left( L_{2,it} -{\rm E} [L_{2,it}]\right)\eta_{Y,2}+ U_{it}^Y$ and ${V}^D_{it} = \left( L_{2,it} -{\rm E} [L_{2,it}]\right)\eta_{D,2}+ U_{it}^D$. For $l=Z,Y,D$, let ${\omega}_l$ be the infeasible penalty weights for the two-way cluster-LASSO estimation of $(\zeta_0,\beta_0,\pi_0)$, as defined in \ref{penaltyweights} with $V_{it}$ replaced by $V^l_{it}$. 
Correspondingly, let $\widehat{V}^l$ be the residuals, and $\widehat{\omega}_l$ be the feasible penalty weights. The two-step debiased estimator $\widehat{\theta}$ for $\theta_0$ using the full-sample is defined as the solution of $\mathbb{E}_{NT}[\psi_{it}(\theta_,\widehat{\eta})]=0$ where $\widehat{\eta}$ are the (post) two-way cluster-LASSO estimators for $\eta_0$ obtained in the first step using the full-sample. 

For statistical analysis, the following notations are used: $ {a}_i = {\rm E}[ V_{it}^ZV_{it}^g|\alpha_i]$,  ${g}_t = {\rm E}[V_{it}^ZV_{it}^g|\gamma_t]$, ${\Sigma}_a  = {\rm E}[{a}_i{a}_i']$, ${\Sigma}_g = \sum_{l=-\infty}^{\infty}{ E}[{g}_t{g}_{t+l}']$; 
      ${a}_{i,j,l} = {\rm E}[f_{it,j}V^l_{it}|\alpha_i]$, ${g}_{t,j,l} = {\rm E}[f_{it,j}V^l_{it}|\gamma_t]$, $e_{it,j,l} =f_{it,j}V^l_{it}-{a}_{i,j,l} -{g}_{t,j,l}$ for  $l =Z,Y,D$; ${A}_0 = { E}[V_{it}^Z{V}^D_{it}]$, ${\Omega}_0 = {\Sigma}_a+ c{\Sigma}_g$.
      
\begin{assumption_plm}[Regularity Conditions for the Partial Linear Model] \ \ \ \ \ \ \ \ \ \ \ 
    \label{regcon_app} 
    \begin{itemize}
         \item[(i)] ${A}_0$ is non-singular.
\vspace{-6pt}
         \item[(ii)] For any $\epsilon$, $h_c(F,\epsilon)$ and $h_d(F,\epsilon)$ are invertible in  $F$.
     \vspace{-6pt}
    
         \item[(iii)] For some $\mu>1,\delta>0$, $\max_{j\leq p}{\rm E}[\lvert {f}_{it,j}\rvert^{8(\mu+\delta)}]<\infty$ and ${\rm E}[\lvert {V}^l_{it} \rvert^{8(\mu+\delta)}]<\infty$ for $l=g,D,Y,Z$.
        \vspace{-6pt}
          
         \item[(iv)] For some $\epsilon>0$,  $\lambda_{min}[{\Sigma}_a] + \lambda_{min}[{\Sigma}_g]>\epsilon $, $\emph{E}\left(a_{i,j,l}^2\right) + \emph{E}\left(g_{t,j,l}^2\right) >\epsilon$, and $\min_{j\leq p} E\left[\left(\sum_{t=1}^Te_{it,j,l}\right)^2|\{\gamma_t\}_{t=1}^T\right] >\epsilon$ almost surely for $l=D,Y,Z$.
          \vspace{-6pt}
        
        \item[(v)] $\log (p/\gamma) = o\left(T^{1/6}/(\log T)^2\right)$.
    \vspace{-6pt}

        \item[(vi)] The feasible penalty weights $\widehat{\omega}_l$ satisfy condition \ref{omegacon} for $l=D,Y,Z$.
    \end{itemize}
\end{assumption_plm}

This set of regularity conditions follow from the assumptions for two-way cluster-LASSO. Here we focus on the non-degenerate case. The only extra condition is Assumption \hyperref[regcon_app]{REG-P}(ii) which is a smoothness condition that ensures the exogeneity properties of $\bar{F}_i$ and $\bar{F}_t$ inherited from $(c_i,\epsilon_i)$ and $(d_t,\epsilon_t)$.

\begin{theorem}
 \label{thm_asymp_norm}
    Suppose, for $P=P_{NT}$ for each $(N,T)$, the following conditions hold for model \ref{cremodel} and $W_{it} = \left(Y_{it},D_{it},X_{it},Z_{it},U_{it},c_i,d_t,\epsilon_i,\epsilon_t \right)$: (i) Assumptions \hyperref[ahk]{AHK}, \hyperref[sparse_eigen]{SE}, \hyperref[crecon]{GMD}, \hyperref[regcon_app]{REG-P}; (ii) sparse approximation in \ref{sparse_D} and \ref{sparse_Y} with $s=o\left( \frac{\sqrt{N\wedge T}}{\log(p/\gamma)} \right)$, $\Vert r_{it}^{\iota}\Vert_{NT,2} = o_P\left(\sqrt{\frac{1}{N\wedge T}}\right)$ for $l=Y,D$. Then, as $N,T\to\infty$ and $N/T\to c$ where $0<c<\infty$, $\sqrt{N}(\widehat{\theta} - \theta_0) \overset{d}{\to} \mathcal{N}(0,V), \ V := {A}_0^{-1} {\Omega}_0 {A}_0^{-1}.$
\end{theorem}

Theorem \ref{thm_asymp_norm} establishes the validity of the proposed inference procedure using the full sample. Note that the sparsity condition and the condition of the approximation errors are stronger than the ones needed for two-way LASSO estimation itself. To estimate the asymptotic variance, the following variance estimators are adapted from \cite{CHS_Restat} and \cite{chen2024fixed} using the full sample:
\begin{alignat}{2}
    \widetilde{V}_{\rm CHS} & = \widetilde{A}_{NT}^{-1}\widetilde{\Omega}_{\rm CHS}\widetilde{A}_{NT}^{-1'}, \ \ \ \ \widetilde{\Omega}_{\rm CHS} &&= \widetilde{\Omega}_{\rm A} + \widetilde{\Omega}_{\rm DK} - \widetilde{\Omega}_{\rm NW}, \label{CHS_app} \\
    \widetilde{V}_{\rm DKA} & = \widetilde{A}_{NT}^{-1}\widetilde{\Omega}_{\rm DKA}\widetilde{A}_{NT}^{-1'}, \ \ \ \ \widetilde{\Omega}_{\rm DKA} &&= \widetilde{\Omega}_{\rm A} + \widetilde{\Omega}_{\rm DK}, \label{DKA_app}
\end{alignat}
where $ \widetilde{A}_{NT}:= \frac{1}{NT}  \sum_{i=1}^N\sum_{t=1}^T (Z_{it}-f_{it}'\widetilde{\zeta})(D_{it}-f_{it}'\widetilde{\pi})$ and
\begin{align*}
{\widetilde{\Omega}}_{\rm A}:=  & \frac{1}{NT^2} \sum_{i=1}^N\sum_{t=1}^T\sum_{r=1}^T \psi_{it}(\widetilde{\theta},\widetilde{\eta}) \psi_{ir}(\widetilde{\theta},\widetilde{\eta})' , \\
{\widetilde{\Omega}}_{\rm DK}:=  & \frac{1}{NT^2} \sum_{t=1}^T\sum_{r=1}^T k\left(\frac{|t-r|}{M} \right) \sum_{i=1}^N\sum_{j=1}^N\psi_{it}(\widetilde{\theta},\widetilde{\eta}) \psi_{jr}(\widetilde{\theta},\widetilde{\eta})', \\
{\widetilde{\Omega}}_{\rm NW}:= & \frac{1}{NT^2} \sum_{i=1}^N\sum_{t=1}^T\sum_{r=1}^T k\left(\frac{|t-r|}{M} \right) \psi_{it}(\widetilde{\theta},\widetilde{\eta}) \psi_{ir}(\widetilde{\theta},\widetilde{\eta})' .
\end{align*}

For simplicity, we deliver the consistency results of variance estimators assuming the approximation is exact. Allowing for approximation errors does not change the main idea but only requires more regularity conditions on the approximation error and lengthier derivations.

\begin{theorem}
    \label{thm_consistency}
     Suppose assumptions for Theorem \ref{thm_asymp_norm} holds for $P=P_{NT}$ for each $(N,T)$ with $r_{it}^D = r_{it}^Y = 0 $ a.s., and $M = o\left(\sqrt{T}\right)$. Then, $(N,T)\to\infty$ and $N/T\to c$ where $0<c<\
    \infty$, 
    \begin{align*}
        \widetilde{V}_{\rm CHS} = & V + o_P(1), \\
        \widetilde{V}_{\rm DKA} = & \widetilde{V}_{\rm CHS} + o_P(1). 
    \end{align*}
\end{theorem}

Theorems \ref{thm_asymp_norm} and \ref{thm_consistency} together justify the panel DML inference without cross-fitting by exploiting the partial linear structure and the sparsity condition \footnote{For more general models, recent literature resort to extra restrictions to avoid cross-fitting(e.g., \citealp{chen2022debiased} and \citealp{cao2025neighborhood}).}.

\section{Monte Carlo Simulation}
\label{mc_simulation}
In this section, we examine the finite sample performance of proposed two-way cluster-LASSO estimation and debiased inference with or without cross-fitting. To focus on the main challenges caused by two-way clustering and high-dimensional nuisance estimation, the simulation study is anchored on linear approximately sparse models without unobserved heterogeneities. In the previous version of working paper, simulation results for linear model with exact sparsity and partial linear model with nuisance functions are also reported. For partial linear model, the performance is largely determined by whether the nuisance functions are well-approximated by basis functions. For the case of exact sparsity, the performance of LASSO-based method is sensitive to the size of the slope coefficients: a larger slope coefficient leads to better selection in the first stage across all LASSO methods. Since it is hard to justify proper choices of nuisance functions and sizes of slope coefficients in the DGP,  we will consider a linear model with the sequence of slope coefficients decay in a polynomial rate, and LASSO-type approaches decides between ``signal" and ``noise'' based on particular penalty choices\footnote{See \cite{Shen2025} for a critique on LASSO-type approaches for learning weak signals. They consider a scenario where signals are individually weak but jointly significant, which is implicitly excluded by the approximate sparsity condition considered in this paper. }. 

The linear model with high-dimensional covariates, approximate sparsity, and two-way clustering dependence is specified as follows:
\begin{align*}
Y_{it} &=  \ D_{it}\theta_0 + X_{it}\beta_0 + U_{it}, \ \ \ D_{it} =  \ X_{it}\pi_0 + V_{it}, \\
U_{it} &=  \ w_1\alpha_{i}^u+  w_2\gamma_{t}^u +  w_3\varepsilon_{it}^u, \ \ \ V_{it} =  w_1\alpha_{i}^v+  w_2\gamma_{t}^v +  w_3\varepsilon_{it}^v,\\
X_{it,j} &=  \ w_1\alpha_{i,j}+  w_2\gamma_{t,j} +  w_3\varepsilon_{it,j} \ \text{for each j=1,...,p.}
\end{align*}
where $\theta_0=1$ is the true parameter of interest. $\beta_0 =\pi_0= \left(1, \frac{1}{2^2},..., \frac{1}{p^2}\right)'$ are $p$-dimensional nuisance parameters. $\alpha_{i}^u, \alpha_{i}^v, \varepsilon_{it}^u, \varepsilon_{it}^v$ are each random draws from $N(0,1)$; $(\alpha_{i,1},...,\alpha_{i,p})'$ and $(\varepsilon_{it,1},...,\varepsilon_{it,p})'$ are each random draws from a joint normal distribution with mean zero and variance-covariance matrix equal to $0.5^{|j-k|}$ in the $(j,k)$'s entry; $\gamma_t^u,\gamma_t^v, \gamma_{t,1},...,\gamma_{t,p}$ each follows an AR(1) process with the coefficient equal to $0.5$ and the initial values randomly drawn from $N(0,0.75)$. The weights $(w_1,w_2,w_3) = (1/3, 1/3, 1/3)$ control the relative importance of the components.

The simulation study examines the Monte Carlo bias (Bias), standard deviation (SD), root mean square error (RMSE), and coverage probability of estimators for $\theta_0$. All estimations are based on the orthogonal moment condition given by \ref{cre_neyman_feasible} with $Z_{it} = D_{it}$ and $f_{it} = X_{it}$. The comparison will be among procedures with and without cross-fitting. The first-step estimations will be based on the POLS estimator (if feasible), the post heteroskedasticity-robust LASSO from \cite{belloni2012sparse}, the post cluster-robust LASSO from \cite{belloni2016inference}, and the post two-way cluster-LASSO. The CHS-type and DKA-type variance estimators, given in Section \ref{sub_sample} with cross-fitting and Section \ref{partial_linear_unob} without cross-fitting, are used for inference. \footnote{In some unreported simulations, I also compare CHS/DKA type variance estimators with Eicker-Huber-White type estimators in \cite{Chernozhukov2018} for random sampling data and Cameron-Galbach-Miller type estimator from \cite{Chiang2022} for multiway clustered data. Since it is well-known that inference based on variance estimators that do not sufficiently account for the dependence would cause over-rejection, it is omitted here}.

The simulation results are based on 1000 Monte Carlo replications\footnote{For replication, the code is included in the supplementary material and also publicly available at \nolinkurl{http://kaichengchen.github.io/twlasso_paneldml_replication.zip}.}. It is a relatively small number of replications but it is necessitated by the high computational cost of multiple high-dimensional estimation and inference procedures, particularly with cross-fitting. For variance estimation, bandwidth parameters $M$ of the Bartlett kernel are required. I use the min-MSE rule from \cite{Andrews1991} for both purposes. For a generic scalar score $v_{it}$, the formula is given as $\hat{M}=1.8171\left(  \frac{\widehat{\rho}^{2}}{\left(  1-\widehat{\rho
}^{2}\right)  ^{2}}\right)  ^{1/3}T^{1/3} + 1$ where $\widehat{\rho}$ is the OLS estimator from the regression $\bar
{\hat{v}}_{t}=\rho\bar{\hat{v}}_{t-1}+\eta_{t}$ where
$\bar{\hat{v}}_{t}=\frac{1}{N}\sum_{i=1}^{N}\hat{v}_{it}$ and $\hat{v}_{it} = \hat{U}_{it}\hat{V}_{it}$.

Table \hyperref[table1]{5.1} presents a set of baseline results with a mild dimension of covariates, $p=200$. The number of covariates is much larger than either cross-sectional or temporal dimensions. In the first step, model selections are done using different LASSO approaches reported in the second column. Comparing the results obtained without using cross-fitting, it is shown that when the number of regressors is not extremely large relative to the sample size, the POLS estimator dominates the sparse methods through different LASSOs in terms of Monte Carlo bias, standard deviation, and coverage probability obtained using DKA standard error, even though the true model is approximately sparse. Among the sparse methods, cluster-robust methods including the proposed two-way cluster-LASSO exhibit the smallest biases and best coverages, with slightly larger standard deviation. In terms of selection, the proposed method selects the number of regressors closest to the true number of relevant regressors while other sparse methods over-select to different extents. 

\begin{table}\centering
\begin{threeparttable}\label{table1}
\begin{tabular}
[c]{c|c|ccc|cc}%
\multicolumn{7}{c}{Table 5.1: $N=T=25$, $p=200$} \\ \hline\hline
{Cross}&{First-Step}&\multicolumn{3}{c|}{Second-Step}&\multicolumn{2}{c}{Coverage (\%)} \\
{Fitting}& {Estimator}&Bias & SD & RMSE  & CHS & DKA  \\ \hline 
\multirow{4}*{No}& POLS&0.001&0.056&0.056&77.5&94.2  \\ 
 & H LASSO       & 0.108&0.077&0.133&45.6&63.8 \\ 
 & C LASSO      & 0.038&0.113&0.119&80.8&86.8 \\
 & TW LASSO    & 0.041&0.109&0.116&80.8&87.8 \\ \hline 
 \multirow{3}*{Yes}& H LASSOS &0.08&0.145&0.167&92.9&97.0 \\
 & C LASSO     &0.043&0.138&0.145&96.0&97.6 \\
 & TW LASSO  &0.009&0.134&0.135&96.8&98.6 \\ \hline 
\end{tabular}
\begin{tablenotes}
\footnotesize
\item Note: Simulation results are based on 1000 replications. Tuning parameters: $(K,L) = (4,8)$, $C_\lambda = 2$, and $\gamma= 0.1/ \log(N \vee T)$. At most 10 iterations are used in calculating the penalty weights. H: heteroskedastic-LASSO; C: cluster-LASSO; TW: two-way cluster-LASSO. Post-LASSO POLS is performed in all first steps. Nominal coverage probability: 0.95.
\end{tablenotes}
\end{threeparttable}
\end{table}

When cross-fitting is employed, all methods have witnessed a significant improvement in terms of sample coverage. This is particularly true for LASSO-based method not robust for dependent data. This is not too surprising because those non-robust sparse methods tend to over-select, and the cross-fitting is designed to remove the overfitting bias and to restore asymptotic normality. As a cost of cross-fitting, the Monte Carlo standard deviation increased, indicating the efficiency loss due to the exclusion of sub-samples in the first-step estimation. It is also worth emphasizing that the CHS- and DKA-type variance estimators designed for cross-fitting approaches play an important role in the desirable sample coverage. In some unreported simulations, it is shown that inference based on the cross-fitting variance estimators proposed in \cite{Chernozhukov2018} and \cite{Chiang2022} suffers from severe under-coverage. It implies that while two-way dependence potentially affects both estimation and inference, its negative impact on the inference is more salient. 

\begin{table}\centering
\begin{threeparttable}\label{table2}
\begin{tabular}
[c]{c|c|ccc|cc}%
\multicolumn{7}{c}{Table 5.2: $N=T=25$, $p=600$} \\ \hline\hline 
{Cross}&{First-Step}&\multicolumn{3}{c|}{Second-Step}&\multicolumn{2}{c}{Coverage (\%)} \\
{Fitting}& {Estimator}&Bias & SD & RMSE  & CHS & DKA  \\ \hline 
\multirow{4}*{No}& POLS& 0.006&0.223&0.223&25.2&37.9 \\ 
& H LASSO       &0.124&0.065&0.140&27.5&50.5 \\ 
& C LASSO     &0.057&0.118&0.131&75.0&82.2 \\
& TW LASSO    &0.054&0.109&0.121&79.0&87.1 \\ \hline 
  \multirow{3}*{Yes}& H LASSO& 0.105&0.140&0.175&92.4&96.0\\ 
 & C LASSO    &0.068&0.136&0.152&94.6&96.5 \\
 & TW LASSO   &0.022&0.139&0.141&95.9&97.3  \\ \hline 
\end{tabular}
\begin{tablenotes}
\footnotesize
\item Note: Simulation results are based on 1000 replications. Tuning parameters: $(K,L) = (4,8)$, $C_\lambda = 2$, and $\gamma= 0.1/ \log(N \vee T)$. At most 10 iterations are used in calculating the penalty weights. H: heteroskedastic-LASSO; C: cluster-LASSO; TW: two-way cluster-LASSO. Post-LASSO POLS is performed in all first steps. Nominal coverage probability: 0.95.
\end{tablenotes}
\end{threeparttable}
\end{table}

As the dimension of the covariates significantly increases and becomes as large as the overall sample size, a different pattern is revealed. Table \hyperref[table2]{5.2} also reports simulation results under the same DGP except that the dimension $p$ now increases to 600, slightly smaller than the overall sample size 625. First, we compare the results obtained without cross-fitting. The simulation results demonstrate that the methods based on the POLS with no selection and those based on the existing LASSO approaches with over-selection all suffer from severe under-coverage. The proposed method, in contrast, remain among the best in terms of bias, RMSE, and coverage. When cross-fitting is performed, there is again a significant improvement across all approaches in terms of the sample coverage but it is also in a cost of a slight efficiency loss measured by the increase in SD. 

Finally, we examine the performance of the two-way cluster-robust method as well as the one-way cluster-robust and non-robust methods when the true DGP features i.i.d data. The theoretical results tell that the high-dimensional methods should remain valid. The pattern in Table \ref{table3} matches with the theories, and the non-robust high-dimensional methods performs slightly better than the cluster-robust ones. The conservativeness of the DKA variance estimator is well-known in the cluster-robust inference literature: two-way cluster-robust variance without the double-counting adjustment can ensure the semi-positive definiteness, but when the components are degenerate, it has a probability limit twice of the true asymptotic variance (see, for example, \citealp{mackinnon2021wild,chen2024fixed}). 

\begin{table}\centering
\begin{threeparttable}\label{table3}
\begin{tabular}
[c]{c|c|ccc|cc}%
\multicolumn{7}{c}{Table 5.3: $N=T=25$, $p=600$, i.i.d data.} \\ \hline\hline 
{Cross}&{First-Step}&\multicolumn{3}{c|}{Second-Step}&\multicolumn{2}{c}{Coverage (\%)} \\
{Fitting}& {Estimator}&Bias & SD & RMSE  & CHS & DKA  \\ \hline 
\multirow{4}*{No}& POLS& 0.014&0.210&0.210&25.8&39.1  \\ 
& H LASSO     &0.001&0.041&0.041&87.0&98.9 \\ 
& C LASSO     &0.002&0.041&0.041&86.9&99.2 \\
& TW LASSO    &0.012&0.043&0.045&84.9&97.6 \\ \hline 
\end{tabular}
\begin{tablenotes}
\footnotesize
\item Note: Simulation results are based on 1000 replications. Tuning parameters: $(K,L) = (4,8)$, $C_\lambda = 2$, and $\gamma= 0.1/ \log(N \vee T)$. At most 10 iterations are used in calculating the penalty weights. H: heteroskedastic-LASSO; C: cluster-LASSO; TW: two-way cluster-LASSO. Post-LASSO POLS is performed in all first steps. Nominal coverage probability: 0.95.
\end{tablenotes}
\end{threeparttable}
\end{table}

\section{Empirical Application} \label{empirical_study}
In this section, I re-examine the effects of government spending on the output of an open economy following the framework of \cite{nakamura2014fiscal}.  It is one of the most cited empirical macro papers in the American Economic Review, and it investigates one classic quantity of interest in economics: the government spending multiplier. The question here is, can we improve the estimation and inference through more robust and flexible methods? As I will show, it is made possible by the proposed toolkit in this paper.

This framework utilizes the regional variation in military spending in the US to estimate the percentage increase in output that results from the increase of government spending by 1 percent of GDP, i.e. government spending multiplier. It is referred to as the "open economy relative multiplier" because this framework takes advantage of uniform monetary and tax policies across the regions in the US to difference-out their effects on government spending and output. The parameter of interest is a scalar, and the baseline model is identified without considering control variables, so why is the high dimensionality relevant here? As discussed below, the high dimensionality from heterogeneity and flexible modeling could be hidden. 

Due to the endogeneity in the variation of the regional military procurement, \cite{nakamura2014fiscal} achieves identification through an instrumental variable (IV) approach. As argued by the authors, the national military spending is largely determined by geopolitical events so it is likely exogenous to the unobserved factors of regional military spending and it affects the regional military spending disproportionally. In other words, the identifying assumption is that the buildups and drawdowns in national military spending are not due to unbalanced military development across regions. Based on this observation, a share-shift type IV is considered and the share is estimated by regressing the regional military spending on the national military spending allowing for region-specific constant slope coefficients.\footnote{All quantities, unless specifically defined, are in terms of two-year growth rate of the real per capita values. Per capita is in terms of total population. \cite{nakamura2014fiscal} also presents results when per capita is calculated using the working age population as a robustness check.} To focus on the main idea, the shares are taken as given and the resulting instrument variable is treated as observable instead of generated regressors.


In this paper, to avoid the endogeneity caused by the misspecification of the function form, I extend the linear model with additive unobserved heterogeneous effects to a partial linear model with non-additive unobserved heterogeneous effects. Let $D_{it}$ be the percentage change in per capita regional military spending in state $i$ and time $t$ and $Z_{it}$ be the IV. Specifically, the baseline model from the original study and the one from this paper differ as follows: 
\begin{alignat*}{2}
    & {\rm  \textbf{Baseline \ model}}:   \ \ && {\rm \textbf{Partial \ linear \ model}}: \\
    & Y_{it} = \theta_0 D_{it} +\pi_i W_t + c_i + d_t + U_{it}, \ \ \ \ \ \ &&  Y_{it} = \theta_0 D_{it} + g(X_{it},W_t,c_i,d_t)+U_{it}
\end{alignat*}
where $\theta_0$ is the parameter of interest, i.e. the true multiplier; $X_{it}$ and $W_t$ are exogenous control variables with the latter being only time-varying; $\pi_i$ are non-random unit specific slope coefficients of $W_t$;  $(c_i,d_t)$ are unobserved heterogeneous effects. In the original study, the linear model is estimated by the two-stage least square (2SLS) with two-way fixed effects. We apply the approach in Section \ref{partial_linear_unob} for the extended model.

In the baseline specification of \cite{nakamura2014fiscal}, $W_t$ are not included. In their alternative specifications, $W_t$ is chosen as the real interest rate or the change in national oil price. These two variables are never included together in the original study. Note that allowing the unit-specific slope coefficients for controls generates many nuisance parameters: with 51 state groups\footnote{The regions in this analysis are defined by the states. \cite{nakamura2014fiscal} also presents results on regions as clusters of states.}, one control would increase 51 parameters and two controls would generate 102 parameters, without considering interactions or higher order terms. With a sample size of less than 2000, the high dimensionality in nuisance parameters could result in a noisy estimate of $\theta_0$. In this paper, I consider additional controls. As is shown in Table 3 of \cite{nakamura2014fiscal}, the change in state population is likely not affected by the treatment (the regional military spending), so it is not subject to the "bad control" issue; But it could affect the treatment and the outcome, so it is included in $X_{it}$. By considering more flexible function forms and additional exogenous control variables, the excludability condition of the instruments is more plausible. On the other hand, the high-dimensionality arose from the flexible function form and the unobserved heterogeneity necessitates the use of high-dimensional methods. Moreover, state-level yearly variables of those macroeconomic characteristics are often considered to be cluster-dependent in both cross-sectional and temporal groups due to common time shocks and state unobserved effects. These concerns justify the use of the proposed robust methods in this paper.

\begin{table}[th]
\label{table1_emp}
\centering{ \begin{threeparttable}
\begin{tabular}
[c]{cccccccc}%
\multicolumn{8}{c}{Table 6.1:  Multiplier estimates from the original model} \\\hline\hline
(1)&(2)&(3)&(4)&(5)&(6)&(7)&(8) \\
Unobs.&Oil& Real&&First&IV 1&CHS&DKA \\ 
Heterog.&Price&Int.&Pop.&Stage&$\widehat{\theta}$&s.e.&s.e. \\\hline 
\multirow{5}*{Fixed Effects}&No&No&No&POLS&1.43&0.68&0.81 \\
&Yes&No&No&POLS&1.30&0.56&0.72 \\
&No&Yes&No&POLS&1.40&0.57&0.70 \\
&Yes&Yes&No&POLS&1.27&0.45&0.71  \\
&Yes&Yes&Yes&POLS&1.36&0.43&0.56 \\\hline
\end{tabular}
\begin{tablenotes}
\small
\item Note: Standard errors are calculated with the truncation parameter $M$ chosen by the min-MSE rule given in Section \ref{mc_simulation}. 
\end{tablenotes}
\end{threeparttable}}\end{table}

The data is available through \cite{nakamura2014fiscal}. It is a balanced (after trimming) state-level yearly panel data with 51 states from 1971-2005 years. The military spending data is collected from the electronic database of DD-350 military procurement forms of the US Department of Defense. The state output is measured by state DGP collected from the US Bureau of Economics Analysis (BEA). The state population data is from the Census Bureau. Data on oil prices is from West Texas Intermediate. The Federal Funds rate is from the FRED database of the St. Louis Federal Reserve. The state inflation measures are constructed from several sources. For more details on data construction, readers are referred to \cite{nakamura2014fiscal}. \footnote{ For replication, the data and code used for this empirical study are included in the supplementary material and also publicly available at \nolinkurl{http://kaichengchen.github.io/twlasso_paneldml_replication.zip}.}


Table \href{table1_emp}{6.1} provides benchmark results for the original model with different choices of control variables. All estimates (columns 6) of are given by 2SLS with two-way fixed effects and the standard errors (s.e.) are calculated using CHS and DKA formulas given in Section \ref{partial_linear_unob}. The estimates of the multiplier replicate those given in \cite{nakamura2014fiscal} with significant differences in the standard errors. It is because the variance estimates here account for the potential two-way dependence while the variance estimator used in \cite{nakamura2014fiscal} assumes cross-sectional independence. 

\begin{table}[th]
\label{table2_emp}
\centering{ \begin{threeparttable}
\begin{tabular}
[c]{cccccccc}%
\multicolumn{8}{c}{Table 6.2:  Open economy relative multiplier estimates from the extended model.} \\\hline\hline
(1)&(2)&(3)&(4)&(5)&(6)&(7)&(8) \\
Cross-&Unobs.&Poly.&Param.&First & &CHS& DKA \\ 
Fitting&Heterog.&Trans.&Gen.&Stage&$\widehat{\theta}$&s.e.&s.e.  \\\hline 
\multirow{4}*{No}&\multirow{4}*{Mundlak}&\multirow{4}*{None}&\multirow{4}*{7}&POLS&1.51&0.66&0.82 \\
&&&&H LASSO& 1.43&0.66&0.81 \\
&&&&C LASSO& 1.43&0.66&0.81 \\
&&&&TW LASSO&1.42&0.69&0.84 \\\hline 
\multirow{4}*{No}&\multirow{4}*{Mundlak}&\multirow{4}*{2nd}&\multirow{4}*{35}&POLS&1.73&0.99&1.15  \\
&&&&H LASSO&1.73&1.01&1.17 \\
&&&&C LASSO&1.75&1.02&1.19 \\
&&&&TW LASSO&1.73&1.00&1.16 \\\hline 
\multirow{4}*{No}&\multirow{4}*{Mundlak}&\multirow{4}*{3rd}&\multirow{4}*{119}&POLS &2.20&1.19&1.37 \\
&&&&H LASSO&1.97&1.16&1.38 \\
&&&&C LASSO&0.98&0.66&0.82 \\
&&&&TW LASSO&1.47&0.59&0.75 \\\hline 
\end{tabular}
\begin{tablenotes}
\small
\item Note: Tuning parameters are chosen as $C_\lambda = 2$, and $\gamma= 0.1/ \log(N \vee T)$. The 7 original regressors are used for initial estimation, and at most 10 iterations are used in calculating the penalty weights. H: heteroskedastic-LASSO; C: cluster-LASSO; TW: two-way cluster-LASSO. The number of predictors generated by the polynomial transformation is reported in column (4). Standard errors are calculated with the truncation parameter $M$ chosen by the min-MSE rule given in Section \ref{mc_simulation}. 
\end{tablenotes}
\end{threeparttable}}\end{table}

The main comparisons are done in Tables \href{table2_emp}{6.2} and \href{table3_emp}{6.3}. In Table \href{table2_emp}{6.2}, no cross-fitting is performed in the first stage. The number of parameters associated with regressors generated by the polynomials transformations are reported in column (4). Overall, with more controls and the polynomial transformation of the observables, the standard errors are generally larger than those in \href{table1_emp}{6.1}. With no transformations of the original regressors, the estimates obtained by four different methods are similar and are consistent with the baseline results. As the flexibility increases with the higher-order polynomial transformations, the number of selected regressors increases across all methods. While the standard errors of most approaches become larger and the estimates deviate from the baseline results, the proposed approach remains less noisy. This indicates that many higher-order polynomials included in the extended model for robustness in the function form may not matter much. While the existing approaches tend to over-select those terms under potential two-way dependence, the proposed LASSO method is robust against over-selection and the panel DML estimator remains accurate. Table \href{table3_emp}{6.3} demonstrates the comparison between various sparse methods with the clustered-panel cross-fitting \footnote{Due to a smaller sample used in the first-step estimation and multicollinearity among the polynomial terms, methods based on the POLS first-step is too noisy and so they are omitted for comparison here.}. It reveals a similar pattern as in Table \href{table2_emp}{6.2}: The variability of different methods increases as the model approximated by higher-order polynomial series, except for the two-way cluster-robust approach. 

\begin{table}[th]
\label{table3_emp}
\centering{ \begin{threeparttable}
\begin{tabular}
[c]{cccccccc}%
\multicolumn{8}{c}{Table 6.3: Open economy relative multiplier estimates from the extended model.} \\\hline\hline
(1)&(2)&(3)&(4)&(5)&(6)&(7)&(8) \\
Cross-&Unobs.&Poly.&Param.&First& &CHS& DKA \\ 
Fitting&Heterog.&Trans.&Gen.&Stage&$\widehat{\theta}$&s.e.&s.e.  \\\hline 
\multirow{3}*{Yes}&\multirow{3}*{Mundlak}&\multirow{3}*{None}&\multirow{3}*{7}&H LASSO& 1.28&1.73&2.00  \\ 
&&&&C LASSO &1.32&1.75&2.03\\ 
&&&&TW LASSO&1.38&1.27&1.54 \\\hline 
\multirow{3}*{Yes}&\multirow{3}*{Mundlak}&\multirow{3}*{2nd}&\multirow{3}*{35}&H LASSO& 1.12&2.18&2.52  \\ 
&&&&C LASSO&1.46&1.95&2.24 \\ 
&&&&TW LASSO&0.01&1.29&1.56 \\\hline 
\multirow{3}*{Yes}&\multirow{3}*{Mundlak}&\multirow{3}*{3rd}&\multirow{3}*{119}&H LASSO&1.81&3.17&3.47 \\ 
&&&&C LASSO&1.25&1.59&1.91 \\
&&&&TW LASSO&1.34 & 1.37&1.64 \\\hline 
\end{tabular}
\begin{tablenotes}
\small
\item Note: The tuning parameters are chosen as $(K,L)=(4,8)$, $C_\lambda = 2$, and $\gamma= 0.1/ \log(p\vee N \vee T)$. The 7 original regressors are used for initial estimation, and at most 10 iterations are used in calculating the penalty weights. H: heteroskedastic-LASSO; C: cluster-LASSO; TW: two-way cluster-LASSO. The number of predictors generated by the polynomial transformation is reported in column (4). Standard errors are calculated with the truncation parameter $M$ chosen by the min-MSE rule given in Section \ref{mc_simulation}. 
\end{tablenotes}
\end{threeparttable}}\end{table}

To conclude, the empirical study of the government spending multiplier using a flexible model and sparse methods illustrates the issue of hidden dimensionality. In the current example, the estimates obtained through the high-dimensional methods do not deviate much from the baseline results, so it implies the nonlinear effects omitted from the original model may not be very relevant. While the two-way cluster-LASSO and the inference procedure remain relatively accurate and provide results as a robustness check, other sparse methods tend to over-select and the DML estimators become too noisy to interpret.

\section{Conclusion and Discussion} \label{sec_conclusion}
The inferential theory for high-dimensional models is particularly relevant in panel data settings where the modeling of unobserved heterogeneity commonly leads to high-dimensional nuisance parameters. This paper enriches the toolbox of researchers in dealing with high-dimensional panel models. Particularly, I propose a package of tools that deal with the estimation and inference in high-dimensional panel models that feature two-way cluster dependence and unobserved heterogeneity. I first develop a weighted LASSO approach for two-way clustered panels. As is shown in the statistical analysis of the two-way cluster-LASSO, the convergence rates are slow due to the cluster dependence, making it challenging for inference purposes. However, by utilizing a cross-fitting method designed for a two-way clustered panel, the rate requirement for the first step can be substantially relaxed, making the proposed two-way cluster-LASSO a feasible first-step estimator for the panel-DML inference procedure in a high-dimensional semiparametric model. Individually, both the two-way cluster-LASSO and the clustered-panel cross-fitting can be of independent interest; Together, they extend the DML approach to panel data settings. Two concerns are also revealed regarding cross-fitting for DML in panel data setting, and alternative approaches are discussed and left for future research.  In the partial linear panel model with high dimensionality as a special case, I further consider the unobserved heterogeneity and inferential theory using the full sample. The validity of the full-sample estimation and inference is established under a slightly stronger sparsity condition compared to the general case with cross-fitting.

The estimation and inferential theory are empirically relevant. I illustrate the proposed approaches in an empirical example and exemplify that high-dimensionality could be hidden in questions not traditionally considered high-dimensional. In practice, when the question is naturally high-dimensional and answered by panel data, then the proposed approaches are natural solutions. When the questions are originally not high-dimensional, it is reasonable to start with a simple model as a baseline and then extend it to a more general and flexible model for a robustness check.

\section*{Data Availability Statement}
The author confirms that the data supporting the findings of this study are available within the article and its supplementary materials.

\section*{Acknowledgment}
For helpful comments and discussions, I thank Tim Vogelsang, Hugo Freeman, Jeff Wooldridge, Antonio Galvao, Kyoo il Kim, Shlomo Levental, Chiang Harold, Rahul Singh, Le Wang, Louise Laage, and Saera Oh. I wish to thank Whitney Newey for his enlightening lectures in 2023 Asian Summer School in Econometrics and Statistics and helpful advice. I also appreciate all helpful comments from seminar participants at Michigan State and conference participants at AMES2024, CES2024, ESIF-AIML2024, MEG2024, CES2025, ESWC2025.

\bibliographystyle{elsarticle-harv}
\bibliography{dml.bib}

@article{cao2025neighborhood,
  title={Neighborhood Stability in Double/Debiased Machine Learning with Dependent Data},
  author={Cao, Jianfei and Leung, Michael P},
  journal={arXiv preprint arXiv:2511.10995},
  year={2025}
}

@article{chen2022debiased,
  title={Debiased machine learning without sample-splitting for stable estimators},
  author={Chen, Qizhao and Syrgkanis, Vasilis and Austern, Morgane},
  journal={Advances in Neural Information Processing Systems},
  volume={35},
  pages={3096--3109},
  year={2022}
}

@article{Shen2025,
journal={NBER Working Papers},
author={Zhouyu Shen and Dacheng Xiu},
title={Can Machines Learn Weak Signals?},
year={2025},
month={Jan},
number={33421},
}

@article{newey1987simple,
  title={A Simple, Positive Semi-Definite, Heteroskedasticity and Autocorrelation},
  author={Newey, Whitney K and West, Kenneth D},
  journal={Econometrica},
  volume={55},
  number={3},
  pages={703--708},
  year={1987}
}

@misc{hounyo2025,
author = {Hounyo, Ulrich and Lin, Jiahao},
year = {2025},
month = {01},
pages = {},
title = {Projection-Based Wild Bootstrap Under General Two-Way Cluster Dependence with Serial Dependence},
}

@article{davezies2025analytic,
  title={Analytic inference with two-way clustering},
  author={Davezies, Laurent and D'Haultf{\oe}uille, Xavier and Guyonvarch, Yannick},
  journal={arXiv preprint arXiv:2506.20749},
  year={2025}
}

@article{babii2023machine,
  title={Machine learning panel data regressions with heavy-tailed dependent data: Theory and application},
  author={Babii, Andrii and Ball, Ryan T and Ghysels, Eric and Striaukas, Jonas},
  journal={Journal of Econometrics},
  volume={237},
  number={2},
  pages={105315},
  year={2023},
  publisher={Elsevier}
}

@article{kock2019uniform,
  title={Uniform inference in high-dimensional dynamic panel data models with approximately sparse fixed effects},
  author={Kock, Anders Bredahl and Tang, Haihan},
  journal={Econometric Theory},
  volume={35},
  number={2},
  pages={295--359},
  year={2019},
  publisher={Cambridge University Press}
}

@article{chetverikov2021cross,
  title={On cross-validated lasso in high dimensions},
  author={Chetverikov, Denis and Liao, Zhipeng and Chernozhukov, Victor},
  journal={Ann. Statist.},
  volume={49},
  number={3},
  pages={1300--1317},
  year={2021},
  publisher={Institute of Mathematical Statistics}
}

@article{gonccalves2011moving,
  title={The moving blocks bootstrap for panel linear regression models with individual fixed effects},
  author={Gon{\c{c}}alves, Silvia},
  journal={Econometric Theory},
  volume={27},
  number={5},
  pages={1048--1082},
  year={2011},
  publisher={Cambridge University Press}
}

@article{bester2008inference,
  title={Inference with Dependent Data Using Cluster Covariance Estimators},
  author={Bester, C Alan and Conley, Timothy G and Hansen, Christian B},
  journal={Working Paper},
  year={2008}
}

@article{vogt2022cce,
  title={CCE estimation of high-dimensional panel data models with interactive fixed effects},
  author={Vogt, Michael and Walsh, Christopher and Linton, Oliver},
  journal={arXiv preprint arXiv:2206.12152},
  year={2022}
}

@article{gao2024robust,
  title={Robust Inference for High-Dimensional Panel Data Models},
  author={Gao, Jiti and Peng, Bin and Yan, Yayi},
  journal={Available at SSRN 4825772},
  year={2024}
}

@article{chernozhukov2021lasso,
  title={Lasso-driven inference in time and space},
  author={Chernozhukov, Victor and Karl H{\"a}rdle, Wolfgang and Huang, Chen and Wang, Weining},
  journal={Ann. Statist.},
  volume={49},
  number={3},
  pages={1702--1735},
  year={2021},
  publisher={Institute of Mathematical Statistics}
}

@article{mackinnon2021wild,
  title={Wild bootstrap and asymptotic inference with multiway clustering},
  author={MacKinnon, James G and Nielsen, Morten {\O}rregaard and Webb, Matthew D},
  journal={Journal of Business \& Economic Statistics},
  volume={39},
  number={2},
  pages={505--519},
  year={2021},
}

@article{gao2022refined,
  title={Refined Cram{\'e}r-type moderate deviation theorems for general self-normalized sums with applications to dependent random variables and winsorized mean},
  author={Gao, Lan and Shao, Qi-Man and Shi, Jiasheng},
  journal={Ann. Statist.},
  volume={50},
  number={2},
  pages={673--697},
  year={2022},
  publisher={Institute of Mathematical Statistics}
}

@book{pena2009self,
  title={Self-normalized processes: Limit theory and Statistical Applications},
  author={Pe{\~n}a, Victor H and Lai, Tze Leung and Shao, Qi-Man},
  year={2009},
  publisher={Springer}
}

@article{bickel2009simultaneous,
  title={Simultaneous analysis of Lasso and Dantzig selector},
  author={Bickel, Peter J and Ritov, Ya’acov and Tsybakov, Alexandre B},
  journal={Ann. Statist.},
  volume={37},
  number={4},
  pages={1705-1732},
  year={2009},
  publisher={Institute of Mathematical Statistics}
}

@article{belloni2012sparse,
  title={Sparse models and methods for optimal instruments with an application to eminent domain},
  author={Belloni, Alexandre and Chen, Daniel and Chernozhukov, Victor and Hansen, Christian},
  journal={Econometrica},
  volume={80},
  number={6},
  pages={2369--2429},
  year={2012},
  publisher={Wiley Online Library}
}

@article{mundlak1978pooling,
  title={On the Pooling of Cross-section and Time-series Data},
  author={Mundlak, Yair},
  journal={Econometrica},
  volume={46(1)},
  pages={69-85},
  year={1978}
}

@article{wooldridge2021two,
  title={Two-way fixed effects, the two-way mundlak regression, and difference-in-differences estimators},
  author={Wooldridge, Jeffrey M},
  journal={Empirical Economics},
    volume={69(5)},
  pages={2545-2587},
  year={2025}
}

@article{belloni2016inference,
  title={Inference in high-dimensional panel models with an application to gun control},
  author={Belloni, Alexandre and Chernozhukov, Victor and Hansen, Christian and Kozbur, Damian},
  journal={Journal of Business \& Economic Statistics},
  volume={34},
  number={4},
  pages={590--605},
  year={2016},
  publisher={Taylor \& Francis}
}

@article{CHS_Restat,
  title={Standard errors for two-way clustering with serially correlated time effects},
  author={Chiang, Harold D and Hansen, Bruce E and Sasaki, Yuya},
  journal={Review of Economics and Statistics},
  pages={1--40},
  year={2024},
  publisher={MIT Press 255 Main Street, 9th Floor, Cambridge, Massachusetts 02142, USA~…}
}

@article{menzel2021bootstrap,
  title={Bootstrap with cluster-dependence in two or more dimensions},
  author={Menzel, Konrad},
  journal={Econometrica},
  volume={89},
  number={5},
  pages={2143--2188},
  year={2021},
  publisher={Wiley Online Library}
}

@article{wooldridge2020inference,
  title={Inference in approximately sparse correlated random effects probit models with panel data},
  author={Wooldridge, Jeffrey M and Zhu, Ying},
  journal={Journal of Business \& Economic Statistics},
  volume={38},
  number={1},
  pages={1--18},
  year={2020},
  publisher={Taylor \& Francis}
}

@article{nakamura2014fiscal,
  title={Fiscal stimulus in a monetary union: Evidence from US regions},
  author={Nakamura, Emi and Steinsson, J{\'o}n},
  journal={American Economic Review},
  volume={104},
  number={3},
  pages={753--792},
  year={2014},
  publisher={American Economic Association 2014 Broadway, Suite 305, Nashville, TN 37203}
}

@article{Andrews1991,
   author = {Donald W. K. Andrews},
   issn = {00129682},
   issue = {3},
   journal = {Econometrica},
   pages = {817},
   title = {Heteroskedasticity and Autocorrelation Consistent Covariance Matrix Estimation},
   volume = {59},
   year = {1991},
}

@article{dehling2010central,
  title={Central limit theorem and the bootstrap for U-statistics of strongly mixing data},
  author={Dehling, Herold and Wendler, Martin},
  journal={Journal of Multivariate Analysis},
  volume={101},
  number={1},
  pages={126--137},
  year={2010},
  publisher={Elsevier}
}

@article{hansen1992consistent,
  title={Consistent covariance matrix estimation for dependent heterogeneous processes},
  author={Hansen, Bruce E},
  journal={Econometrica},
  pages={967--972},
  year={1992},
  publisher={JSTOR}
}

@book{hansen2022econometrics,
  title={Econometrics},
  author={Hansen, Bruce},
  year={2022},
  publisher={Princeton University Press}
}

@book{davidson1994stochastic,
  title={Stochastic limit theory: An introduction for econometricians},
  author={Davidson, James},
  year={1994},
  publisher={OUP Oxford}
}

@article{chen2024fixed,
  title={Fixed-b asymptotics for panel models with two-way clustering},
  author={Chen, Kaicheng and Vogelsang, Timothy J},
  journal={Journal of Econometrics},
  volume={244},
  number={1},
  pages={105831},
  year={2024},
  publisher={Elsevier}
}

@article{breinlich2022machine,
  title={Machine learning in international trade research-evaluating the impact of trade agreements},
  author={Breinlich, Holger and Corradi, Valentina and Rocha, Nadia and Ruta, Michele and Santos Silva, JMC and Zylkin, Thomas},
  year={2022},
  journal={CEPR Discussion Paper}
}

@techreport{chernozhukov2019demand,
  title={Demand analysis with many prices},
  author={Chernozhukov, Victor and Hausman, Jerry A and Newey, Whitney K},
  year={2019},
  institution={National Bureau of Economic Research}
}

@article{Chernozhukov2018,
author = {Chernozhukov, Victor and Chetverikov, Denis and Demirer, Mert and Duflo, Esther and Hansen, Christian and Newey, Whitney and Robins, James},
issn = {1368423X},
journal = {Econometrics Journal},
number = {1},
pages = {C1--C68},
title = {{Double/debiased machine learning for treatment and structural parameters}},
volume = {21},
year = {2018}
}

@article{Semenova2023,
author = {Semenova, Vira and Goldman, Matt and Chernozhukov, Victor and Taddy, Matt},
issn = {1759-7323},
journal = {Quantitative Economics},
number = {2},
pages = {471--510},
publisher = {The Econometric Society},
title = {{Inference on heterogeneous treatment effects in high‐dimensional dynamic panels under weak dependence}},
volume = {14},
year = {2023}
}

@article{Chiang2022,
archivePrefix = {arXiv},
arxivId = {1909.03489},
author = {Chiang, Harold D. and Kato, Kengo and Ma, Yukun and Sasaki, Yuya},
issn = {15372707},
journal = {Journal of Business and Economic Statistics},
number = {3},
pages = {1046--1056},
title = {Multiway cluster robust double/debiased machine learning},
volume = {40},
year = {2022}
}

\section{Appendix}
\label{app_0}
This appendix provides extra results and proofs for the two-way cluster-LASSO and cross-fitting. Proofs for the panel DML and the partial linear model are provided in the Online Supplementary Material.\\

\noindent \textit{\textbf{The Slow Rate of Convergence in an Oracle Case} } \\
To illustrate the slow rate of convergence under the component structure representation, we consider the simplest multivariate mean model $Y_{it} = \theta_0 + f(\alpha_i,\gamma_t,\epsilon_{it})$ where $Y_{it}$ is a high-dimensional vector with dimension $s = o(NT)$ and $\theta_0 = {\rm E}[Y_{it}]$. To estimate the high-dimensional vector $\theta_0$, we consider the sample mean estimator $\widehat{\theta} = \frac{1}{NT}\sum_{i=1}^N\sum_{t=1}^T Y_{it}$. Consider the following decomposition:
\begin{align}
    \widehat{\theta} - \theta_0 = \frac{1}{NT}\sum_{i=1}^N\sum_{t=1}^T (a_i+g_t+e_{it}) = \frac{1}{N}\sum_{i=1}^N a_i + \frac{1}{T}\sum_{t=1}^T g_t + \frac{1}{NT}\sum_{i=1}^N\sum_{t=1}^T e_{it}, \label{mean_estimator}
\end{align}
where $a_i := {\rm E}[Y_{it}-\theta_0| \alpha_i], g_t := {\rm E}[Y_{it} - \gamma_t]$, and $e_{it} := Y_{it}-\theta_0 - a_i - g_t$. For illustration purposes, suppose those components are i.i.d sequences and independent of each other. Then it can be shown that, under some regularity conditions, $\Vert\widehat{\theta} - \theta_0\Vert_2 = O_P\left(\sqrt{\frac{s}{N\wedge T}}\right)$, which is slower than the common rate requirement for inferential theory. \\

\noindent\textit{\textbf{Algorithm 1: Implementation of the Two-Way Cluster-LASSO}}
\begin{enumerate}
    \item[i] Let $\hat{\omega}_{j,0}^2$ be the product of the sample variance of $f_{it,j}$ and the sample variance of $Y$, for each j. Set ${\lambda}_0 = 2C_\lambda \sqrt{NT} \Phi^{-1}\left(1-\frac{\gamma}{2p}\right)$.
    Obtain the initial residuals $\tilde{V}_0$ using post-LASSO with the penalty level and weights set as above. 
    
    \item[ii] Given residual $\tilde{V}_l$ for $l=0,1,2,...$, calculate $\hat{\omega}_{j,l}$ according to \ref{feasible_w} for each $j$. Run post-LASSO with $\hat{\omega}_{j,l}$ and $\lambda = 2C_\lambda \sqrt{N}T \Phi^{-1}\left(1-\frac{\gamma}{2p}\right)$and update the residual as $\tilde{V}_{l+1}$.

    \item[iv] Repeat Step ii until it converges with the number of iteration greater or equal to 2. Obtains the (post) LASSO estimates from the last iteration. 
\end{enumerate}

\noindent\textit{\textbf{Some Lemmas for Two-Way Clustering}} \\
We will first introduce two lemmas regarding the law of large number (LLN) and the central limit theorem (CLT) for two-way clustered arrays with correlated time effects. They are restated and generalized from Theorems 1 and 2 in \cite{CHS_Restat}. The following notations will also be used frequently throughout the appendices:  Let $\{W_{it}: i=1,..,N; t=1,...,T\}$ be an array of random vectors taking values in $\mathbb{R}^{p}$. Let $F: \mathbb{R}^{p}\rightarrow \mathbb{R}^{k}$ be a measurable function where $k$ is a constant. We define the components of Hajek-type decomposition as $a_i = {\rm E}[F(W_{it}) - {\rm E}[F(W_{it})]|\alpha_i]$, $g_t = {\rm E}[F(W_{it})- {\rm E}[F(W_{it})]|\gamma_t]$, and $e_{it} = W_{it}- {\rm E}[F(W_{it})] - a_i -g_t$ and their corresponding (long-run) variance-covariance matrices: 
\begin{align*}
    \Sigma_a = {\rm E}[a_ia_i'], \ \ \ \Sigma_g = \sum_{l=-\infty}^{\infty} {\rm E}[g_tg_{t+l}'], \ \ \ \Sigma_e = \sum_{l=-\infty}^{\infty} {\rm E}[e_{it}e_{i,t+l}'].
\end{align*}
We can rewrite $F(W_{it}) = a_i+g_t+e_{it}$. Suppose that $W_{it}$ satisfy Assumption \hyperref[ahk]{AHK}, then the decomposition has the following properties:
\begin{enumerate}
    \item[(i)] $\{a_i\}_{i\geq 1}$ is a sequence of i.i.d random vectors,  $\{g_t\}_{t\geq 1}$ are strictly stationary and $\beta$-mixing with the the mixing coefficient $\beta_g(m)\leq \beta_{\gamma}(m)$ for all $m\geq 1$; for each $i$, $\{e_{it}\}_{t\geq 1}$ is also strictly stationary; and $a_i$ is independent of $g_t$. 
    \item[(ii)] $a_i$, $b_t$, $e_{it}$ are mean zero. 
    \item[(iii)] Conditional on $(\gamma_t,\gamma_r)$, $e_{it}$ and $e_{jr}$ are independent for $j\ne i$. 
    \item[(iv)] The sequences $\{a_i\}$, $\{g_t\}$, $\{e_{it}\}$ are mutually uncorrelated. 
\end{enumerate}

Properties (i) and (ii) are straightforward. Property (iii) is due to the assumption that $\{\alpha_i\}$ and $\{\varepsilon_{it}\}$ are each i.i.d sequence and independent of each other. Property (iv) is less obvious. One can show ${\rm E}[e_{it}|\gamma_r] = 0 $ and ${\rm E}[e_{it}|\alpha_j]$ for any $i,t,j,r$. It is less obvious to see ${\rm E}[e_{it}|\gamma_r] = 0 $ for some $r\ne t$:
\begin{align*}
    {\rm E}[e_{it}|\gamma_r] & = {\rm E}[\psi\left(W_{it};\theta_0,\eta_0\right)|\gamma_r] - {\rm E}[a_i|\gamma_r]-{\rm E}[g_s|\gamma_t]\\
    & = {\rm E}[{\rm E}[\psi\left(f(\alpha_i,\gamma_t,\varepsilon_{it});\theta_0,\eta_0\right)|\gamma_t,\gamma_r]|\gamma_r] - {\rm E}[a_i]-{\rm E}[g_t|\gamma_r]\\
    & = {\rm E}[{\rm E}[\psi\left(f(\alpha_i,\gamma_t,\varepsilon_{it});\theta_0,\eta_0\right)|\gamma_t]|\gamma_r] - {\rm E}[a_i]-{\rm E}[g_t|\gamma_r]\\
    & = {\rm E}[g_t|\gamma_r] - {\rm E}[g_t|\gamma_r]=0
\end{align*}
where the second equality follows from the iterated expectation and the independence of $\alpha_i$ and $\gamma_r$ and the third equality follows from that given $\gamma_t$, $\gamma_r$ is independent of $(\alpha_i,\gamma_t,\varepsilon_{it})$.  

Using the properties above, one can derive the LLN and CLT for two-way clustered panel data. The following two lemmas are regarding the LLN and CLT for two-way clustered data with correlated time effects.

\begin{lemma}
\label{lemma_a1}
    Suppose that $W_{it}$ satisfy Assumption \hyperref[ahk]{AHK} and ${\rm E}\left[\Vert F(W_{it})\Vert^{4(r+\delta)}\right]<\infty$. Then, 
\begin{itemize}
    \item[i] $\Vert \Sigma_a \Vert <\infty$,  $\Vert \Sigma_g \Vert <\infty$, and  $\Vert\Sigma_e \Vert <\infty$ where 
    \item[ii] $\emph{Var}\left(\emph{E}_{NT}[F(W_{it})]\right)  = \frac{1}{N}\Sigma_a + \frac{1}{T}\Sigma_g (1+o(1)) + \frac{1}{NT}\Sigma_e(1+o(1))$ as $N,T\to\infty$.
    \item[iii]  $\emph{E}_{NT}[F(W_{it})]\overset{p}{\to} {\rm E}[F(W_{it})]$ as $N,T\to\infty$.
\end{itemize}
\end{lemma}

\begin{lemma}
    \label{lemma_a2}
    With the same setting as in Lemma \ref{lemma_a1}, further assume that either (1) $\lambda_{min} \left[\Sigma_a\right] >0$ or $\lambda_{min} \left[\Sigma_g\right] >0$, or (2) $W_{it}$ is i.i.d. over $i$ and $t$. Then, as $N,T\to \infty$ and $N/T\to c$, $\sqrt{N} \left(\emph{E}_{NT}[F(W_{it})] - {\rm E}[F(W_{it})]\right) \overset{d}{\to} \mathcal{N}(0, \Sigma_a+c\Sigma_g)$
\end{lemma}

Lemmas \ref{lemma_a1} and \ref{lemma_a2} follow from Theorems 1 and 2 in \cite{CHS_Restat}. \\

\begin{proof}[\textbf{Proof of Theorem \ref{performance_bound}}]
We will first show the regularization event in terms of the infeasible penalty weights $\omega$ as defined in \ref{penaltyweights}. Due to the AHK representation as in Assumption \hyperref[ahk]{AHK}, we can decompose $f_{it,j}V_{it}$ as $ f_{it,j}V_{it} = a_{i,j} + g_{t,j}+e_{it,j}$ where $a_{i,j}:= {\rm E}[f_{it,j}V_{it}|\alpha_i]$, $g_{t,j} = {\rm E}[f_{it,j}V_{it}|\gamma_t]$, and $e_{it,j} = f_{it,j}V_{it} - a_{i,j} - g_{t,j} $, for $j=1,..,p$.

We first consider the non-degeneracy case: $\min_{j\leq p}{\rm E}(a_{i,j}^2)>0$ or $\min_{j\leq p}{\rm E}(g_{t,j}^2)>0$. Due to the inequality \ref{three_ineq}, we have, for each $j=1,...,p$: 
\begin{align*}
        & {\rm P}\left(\frac{1}{NT}\left|\sum_{i=1}^N\sum_{t=1}^T\omega_{j}^{-1}f_{it,j}V_{it}\right| > \frac{\lambda}{2c_1NT}\right) \leq   {\rm P}\left( \left|\frac{1}{\sqrt{N}}\sum_{i=1}^N\omega_{a,j}^{-1} a_{i,j} \right|  > \frac{ \lambda\sqrt{N}}{2C_\lambda NT}\right)\\
        &+ {\rm P}\left( \left |\frac{1}{\sqrt{T}}\sum_{t=1}^T \left(\sqrt{T/N})\omega_{g,j}\right)^{-1} g_{t,j} \right|  > \frac{ \lambda\sqrt{N}}{2C_\lambda NT}\right)  + {\rm P}\left( \left|\frac{1}{\sqrt{N}T} \sum_{i=1}^N\sum_{t=1}^T\omega_{e,j}^{-1} e_{it,j} \right|  > \frac{ \lambda\sqrt{N}}{2C_\lambda NT}\right) \\
         :=& p_{1,j}(\lambda) + p_{2,j}(\lambda) + p_{3,j}(\lambda) 
\end{align*}
where $C_\lambda = c_1c_\omega$. $\omega_{a,j}^2= \frac{1}{N}\sum_{i=1}^N a_{i,j}^2, \ \ \omega_{g,j}^2 = \frac{N}{T^2} \sum_{b=1}^B\left(\sum_{t\in H_b} g_{t,j}\right)^2$, $\omega_{e,j}^2=\frac{1}{NT^2}\sum_{i=1}^N\left(\sum_{t=1}^T e_{it,j}\right)^2$. Then, the union-bound inequality implies 
  \begin{align*}
        & {\rm P}\left(\max_{j=1,...,p}  \left| \sum_{i=1}^N\sum_{t=1}^T\omega_{j}^{-1}f_{it,j}V_{it}\right| > \frac{\lambda}{2c_1}\right)  \leq \sum_{j=1}^p \left[p_{1,j}(\lambda) + p_{2,j}(\lambda) + p_{3,j}(\lambda)\right] 
    \end{align*}

To bound $p_{1,j}(\lambda)$, we consider a moderate deviation theorem for self-normalized sums of independent random variables. For $j=1,...,p$, define $ \Xi_{a,j} =  \frac{\left[{\rm E}(a_{i,j})^2\right]^{1/2}}{\left[{\rm E}(a_{i,j})^{3}\right]^{1/3}}$. Under Assumption  \hyperref[regcon]{REG}(ii), $\max_{j\leq p}E|a_{i,j}|^3 <\infty$ by Hölder's inequality and Jensen's inequality. Under case (3), $\min_{j\leq p} E|a_{i,j}|^2 >0$. Therefore, $\min_j \Xi_{a,j} > 0$. By Theorem 7.4 of \cite{pena2009self} with $\delta=1$, we have for any $x\in [0,N^{1/6}\Xi_{a,j}]$ that 
\begin{align*}
     {\rm P}\left( \left|\frac{1}{\sqrt{N}}\sum_{i=1}^N\omega_{a,j}^{-1} a_{i,j} \right|  > x \right) \leq 2\left( 1- \Phi(x) \right) \left[ 1 + O(1)\left(\frac{1+x}{N^{1/6}\Xi_{a,j}}\right)^{3}\right]
\end{align*}
Let $l_{a,N}$ be some positive increasing sequence. If $N^{1/6}\Xi_{a,j}/l_{a,N} - 1>0$ and $x\in [0,N^{1/6}\Xi_{a,j}/l_{a,N} - 1]$, then
\begin{align*}
     {\rm P}\left( \left|\frac{1}{\sqrt{N}}\sum_{i=1}^N\omega_{a,j}^{-1} a_{i,j} \right|  > x \right) \leq 2\left( 1- \Phi(x) \right) \left[ 1 + O(1)\left(\frac{1}{l_{a,N}}\right)^{3}\right]
\end{align*}
Then, setting $\lambda = 2C_\lambda \sqrt{N}T\Phi^{-1}\left(1-\frac{\gamma}{2p}\right)$ gives
\begin{align*}
    &\sum_{j=1}^p p_{1,j}(\lambda) = {\rm P}\left( \left| \frac{1}{\sqrt{N}}\sum_{i=1}^N\omega_{a,j}^{-1} a_{i,j} \right|  > \frac{  \lambda\sqrt{N}}{2C_\lambda NT}\right)\\
    \leq & 2p(1-\Phi(\Phi^{-1}(1-\gamma/2p))) \leq  {\gamma}[1+O(1)(1/l_{a,N})^{3}]
\end{align*}
given that $\Phi^{-1}\left(1-\frac{\gamma}{2p}\right) \in [0,N^{1/6}\min_j\Xi_{a,j}/l_{a,N} - 1]$ and $N^{1/6}\min_j\Xi_{a,j}/l_{a,N} - 1>0$. Note that $\Phi^{-1}\left(1-\frac{\gamma}{2p}\right) \lesssim \sqrt{\log(p/\gamma)}=o\left(N^{1/12} / \log N\right)$  under Assumption \hyperref[regcon]{REG}(i) and $N/T\to c$ as $N,T\to\infty$. Therefore, it suffices to take $l_{a,N} = O(\log N)$, and it follows that $\sum_{j=1}^p p_{1,kl}(\lambda) \to 0$ as $\gamma\to 0$ and $(N,T)\to \infty$.

To bound $p_{2,j}(\lambda)$, we utilize a moderate deviation theorem for self-normalized sums of weakly dependent random variables. Observe that $g_{t,j} = {\rm E}[f_{it,j}V_{it}|\gamma_t]$ is beta-mixing with coefficient $\beta_g(q)$ satisfying 
\begin{align*}
    \beta_g(q) \leq \beta_\gamma(q) \leq c_\kappa exp(-\kappa q) \ \forall q \in \mathbb{Z^+} 
\end{align*}
Furthermore, by the strict stationarity and the non-degeneracy condition of case (3), we can verify that for some $\nu>0$, $ {\rm E}\left[\sum_{t=r}^{r+m}{g_{t,j}}\right]^2 \geq \nu^2 m $ for all $t\geq 1, r\geq 0, m\geq 1$. By Assumption \hyperref[regcon]{REG}(ii) and Hölder's inequality, we have ${\rm E}|f_{it,j}V_{it}|^{4(\mu+\delta)}<\infty$  for some $\mu>1, \ \delta>0$. Then, by Theorem 3.2 of \cite{gao2022refined} with $\tau = 1$ and $\alpha= \frac{1}{1+2\tau}$, we have
\begin{align*}
   \sum_{j=1}^p {\rm P}\left( \left | \frac{1}{\sqrt{T}}\sum_{t=1}^T\left(\sqrt{T/N})\omega_{g,j}\right)^{-1} g_{t,j} \right|  > x \right)\leq 2 p\left(1-\Phi(x)\right) \left[ 1+ O(1)\left( \frac{1}{l_{g,T}} \right)^2 \right]
\end{align*}
uniformly for $x\in \left(0,d_0(\log T)^{-1/2}T^{1/12}/l_{g,T} \right)$ where $d_0$ is some positive constant and $l_{g,T}$ is some positive increasing sequence. Then, setting $\lambda =2C_\lambda\sqrt{N}T  \Phi^{-1}(1-\frac{\gamma}{2p})$ gives, for all $j=1,...,p$,
\begin{align*}
 \sum_{j=1}^p  p_{2,j}(\lambda) = {\rm P}\left( \left | \frac{1}{\sqrt{T}}\sum_{t=1}^T\left(\sqrt{T/N})\omega_{g,j}\right)^{-1} g_{t,j} \right|  > \frac{\lambda\sqrt{N}}{2C_\lambda NT}\right) \leq \gamma \left[ 1+ O(1)\left( \frac{1}{l_{g,T}} \right)^2 \right]
\end{align*}
given that $\Phi^{-1}(1-\frac{\gamma}{2p}) \in  \left(0,d_0(\log T)^{-1/2}T^{1/12}/l_{g,T}  \right)$. Under Assumption \hyperref[regcon]{REG}(i), we have $\log (p/\gamma) = o\left(T^{1/6}/ (\log T)^2\right))$ and so $ \Phi^{-1}\left(1-\frac{\gamma}{2p}\right)  \lesssim \sqrt{\log(p/\gamma)} = o\left(T^{1/12}/(\log T\right)$. Therefore, by taking $l_{g,T} = O\left( (\log T)^{1/2}\right)$, it follows that $\sum_{j=1}^p p_{2,j}(\lambda) \to 0$ as $\gamma\to 0$ and $(N,T)\to \infty$.

Consider $p_{3,j}(\lambda)$. We observe that ${\rm E}[{e}_{i,j}] = 0$ by iterated expectation, and ${e}_{it,j}$ are independent over $i$ conditional on $\{\gamma_t\}_{t=1}^T$. We have shown previously that $\max_{j\leq p}{\rm E}|f_{it,j}V_{it}|^{4(\mu+\delta)}<\infty$  for some $\mu>1, \ \delta>0$. Given that $e_{it,j} = f_{it,j}V_{it} - a_{i,j}-g_{t,j}$ and $\max_{j\leq p}{\rm E}|a_{i,j}|^{4(\mu+\delta)}<\infty$, $\max_{j\leq p}{\rm E}|g_{t,j}|^{4(\mu+\delta)}<\infty$ due to Jansen's inequality and iterated expectation, we have $\max_{j\leq p}{\rm E}|e_{it,j}|^{4(\mu+\delta)}<\infty$ and so $\max_{j\leq p}{\rm E}|\bar{e}_{i,j}|^{4(\mu+\delta)}<\infty$ due to Minkowski's inequality. 
Note that
\begin{align*}
    {\rm Var}\left(\bar{e}_{i,j}\right)  =\frac{1}
{T}\sum_{l=-(T-1)}^{T-1}\left(  1-\frac{|l|}{T}\right)
{\rm E}(e_{it,j}e_{i,t+l,j})=\frac{1}{T}\Sigma_{e,j}(1+o(1)).
\end{align*}
By Lemma \ref{lemma_a1}, $|\Sigma_{e,j}|<\infty$. Furthermore, as is shown below, $\omega_j$ is bounded from below by some constant $a>0$.

Define $\Xi_{e,j}^\gamma:= \frac{(E[|\bar{e}_{i,j}|^2\mid  \{\gamma_t\}_{t=1}^T])^{1/2} }{(E[|\bar{e}_{i,j}|^3\mid  \{\gamma_t\}_{t=1}^T])^{1/3}}$. Since $\bar{e}_{i,j}$ is conditionally independent, we can apply the conditional version of self-normalization theorem for independent variables from \cite{pena2009self} as follows:
\begin{align*}
     {\rm P}\left( \left| \frac{1}{\sqrt{N}} \sum_{i=1}^N\omega_{e,j}^{-1} \bar{e}_{i,j} \right|  > x \mid \{\gamma_t\}_{t=1}^T\right) \leq 2\left( 1- \Phi(x) \right) \left[ 1 + O(1)\left(\frac{1+x}{N^{1/6}\Xi_{e,j}^\gamma}\right)^{3}\right]
\end{align*}
for $x\in [0,N^{1/6}\Xi_{e,j}^\gamma]$. Under Assumption \hyperref[regcon]{REG}(iii), $1/\Xi_{e,j}^{\gamma} = O_P((E[|\bar{e}_{i,j}|^3\mid  \{\gamma_t\}_{t=1}^T])^{1/3})$. Under Assumption  \hyperref[regcon]{REG}(ii), $E[(E[|\bar{e}_{i,j}|^3\mid \{\gamma_t\}_{t=1}^T])^{1/3} ] \leq (E[|\bar{e}_{i,j}|^3])^{1/3}<\infty$, uniformly over $j\leq p$ by Jensen's inequality, iterated expectation, and Hölder's inequality. Therefore, by integrating out $\{\gamma_t\}_{t=1}^T$, we obtain that 
\begin{align*}
     {\rm P}\left( \left| \frac{1}{\sqrt{N}} \sum_{i=1}^N\omega_{e,j}^{-1} \bar{e}_{i,j} \right|  > x \right) \leq 2\left( 1- \Phi(x) \right) \left[ 1 + O(1)\left(\frac{1+x}{N^{1/6}}\right)^{3} \right]
\end{align*}
for $x\in [0,AN^{1/6}]$ for some constant $A<\infty$. Let $l_{e,N}$ be some increasing sequence. If $AN^{1/6}/l_{e,N} - 1>0$ and $x\in [0,AN^{1/6}/l_{e,N} - 1]$, then
\begin{align*}
    = {\rm P}\left( \left|\frac{1}{\sqrt{N}} \sum_{i=1}^N\omega_{e,j}^{-1} \bar{e}_{i,j} \right|  > x \right) \leq 2\left( 1- \Phi(x) \right) \left[ 1 + O(1)\left(\frac{1}{l_{e,N}}\right)^{3}\right]
\end{align*}
Then, setting $\lambda = 2C_\lambda\sqrt{N}T\Phi^{-1}\left(1-\frac{\gamma}{2p}\right)$ gives
\begin{align*}
    &\sum_{j=1}^p p_{3,j}(\lambda) ={\rm P}\left( \left| \frac{1}{\sqrt{N}T}\sum_{i=1}^N\sum_{t=1}^T\omega_{e,j}^{-1} e_{it,j} \right|  > \frac{\lambda\sqrt{N}}{2C_\lambda NT}\right) \\
    \leq & 2p(1-\Phi(\Phi^{-1}(1-\gamma/2p))) \leq  {\gamma}[1+O(1)(1/l_{e,N})^{3}]
\end{align*}
given that $\Phi^{-1}\left(1-\frac{\gamma}{2p}\right) \in [0,AN^{1/6}/l_{e,N} - 1]$ and $AN^{1/6}/l_{e,N} - 1>0$. Note that $\Phi^{-1}\left(1-\frac{\gamma}{2p}\right) \lesssim \sqrt{\log(p/\gamma)}=o\left(N^{1/12} / \log N\right)$  under Assumption \hyperref[regcon]{REG}(i) and $N/T\to c$ as $N,T\to\infty$. Therefore, it suffices to take $l_{e,N} = O(\log N)$, and it follows that $\sum_{j=1}^p p_{3,kl}(\lambda) \to 0$ as $\gamma\to 0$ and $(N,T)\to \infty$.

Put together, we have shown 
\begin{align}
        & {\rm P}\left(\max_{j=1,...,p}  \left|\frac{1}{NT} \sum_{i=1}^N\sum_{t=1}^T\omega_{j}^{-1}f_{it,j}V_{it}\right| \leq \frac{\lambda}{2c_1NT}\right) \to 1 \label{step1} .
\end{align}

Now we consider the degenerate case: $W_{it}$ is i.i.d over $i$ and $t$. In this case, we have $a_{i,j} = {\rm E}[f_{it,j}V_{it}|\alpha_i] = {\rm E}[f_{it,j}V_{it}] = 0$, $g_{t,j} = {\rm E}[f_{it,j}V_{it}|\gamma_t] = {\rm E}[f_{it,j}V_{it}] = 0$, and $e_{it,j} = f_{it,j}V_{it}$. Then, we observe that $p_{1,j}(\lambda) =p_{2,j}(\lambda)  $. It follows that
\begin{align*}
        & {\rm P}\left(\frac{1}{NT}\left|\sum_{i=1}^N\sum_{t=1}^T\omega_{j}^{-1}f_{it,j}V_{it}\right| > \frac{\lambda}{2c_1NT}\right) = {\rm P}\left(\frac{1}{\sqrt{NT}}\left|\sum_{i=1}^N\sum_{t=1}^T ( \sqrt{T}\omega_{e,j})^{-1}e_{it,j}\right| > \frac{\lambda}{2c_1\sqrt{NT}}\right)
\end{align*} 
Then, following the same steps as for $p_{3,j}(\lambda)$, we have that, with $\lambda = 2C_\lambda\sqrt{NT}\Phi^{-1}\left(1-\frac{\gamma}{2p}\right)$, the probability above converges to zero and so Condition \ref{step1} is obtained for the degenerate case, too. Note that here we rescale the penalty weights and the common penalty level so that the penalty weights is not degenerate under the iid case. See Lemma \ref{weight_valid_1} for the rate of convergence for the penalty weights under the non-degenerate and degenerate cases. 

Now, we can apply Lemma 6 of \cite{belloni2012sparse} to obtain the finite sample bounds on $\left\Vert f_{it}'\left( \widehat{\zeta} - \zeta_0\right) \right\Vert_{NT,2}$ and $\left\vert \omega\left( \widehat{\zeta} - \zeta_0 \right) \right\vert_1 $.
Let $\delta$ be some generic vector of nuisance parameters and let $J^1_p$ be a subset of an index set $J_p={1,...,p}$ and $J^0_p = J_p \backslash J_p^1$. Let $\delta^1$ be a copy of $\delta$ with its $j$-th element replaced by 0 for all $j\in J_p^0$ and similarly let $\delta^0$ be a copy of $\delta$ with its $j$-th element replaced by 0 for all $j\in J_p^1$. Define the restricted eigenvalues and Gram matrix as follows:
\begin{align*}
     K_C(M_f) = \min_{\delta: \Vert\delta^0\Vert_1 \leq C \Vert \delta^1\Vert_1, \ \Vert\delta\Vert\ne 0, \ |J^1_p|\leq s} \frac{\sqrt{s \delta 'M_f \delta}}{\left\Vert \delta^1 \right\Vert_1}, \ \ \ M_f  ={\rm E}_{NT} [f_{it}'f_{it}].
\end{align*}
Define the weighted restricted eigenvalues as follows:
\begin{align*}
     K_C^{\omega}(M_f)  = \min_{\delta: \Vert\omega\delta^0\Vert_1 \leq C \Vert \omega\delta^1\Vert_1, \ \Vert\delta\Vert\ne 0, \ |J^1_p|\leq s} \frac{\sqrt{s \delta 'M_f \delta}}{\left\Vert \omega \delta^1 \right\Vert_1}.
\end{align*}
Assumption \hyperref[asm]{ASM}, Conditions \ref{omegacon}, and \ref{step1}, Lemma 6 of \cite{belloni2012sparse} implies that, for the non-degenerate case, 
\begin{align*}
    \left\Vert f_{it}'\left( \widehat{\zeta} - \zeta_0\right) \right\Vert_{NT,2} \leq & \left(u+\frac{1}{c_1}\right)\frac{\sqrt{s}\lambda}{NTK_{c_0}^{\omega}(M_f)} + 2   \left\Vert r \right\Vert_{NT,2}=  O_P\left( \frac{1}{K_{c_0}^{\omega}(M_f)}\sqrt{\frac{s\log(p/\gamma)}{N\wedge T}} +  \sqrt{\frac{s}{N\wedge T}}\right), \\
     \left\Vert \omega\left( \widehat{\zeta} - \zeta_0 \right) \right\Vert_1 \leq & \frac{3c_0\sqrt{s}}{K_{2c_0}^{\omega}(M_f)}\left[\left(u+\frac{1}{c_1}\right) \frac{\sqrt{s}\lambda}{NTK_{c_0}^{\omega}(M_f)} + 2\left\Vert r \right\Vert_{NT,2}\right] + 3c_0 \frac{NT}{\lambda} \left\Vert r \right\Vert_{NT,2}^2, \\
     =& O_P\left(\frac{s}{K_{2c_0}^{\omega}(M_f)K_{c_0}^{\omega}(M_f)}\sqrt{\frac{\log (p/\gamma)}{N\wedge T}} + \sqrt{\frac{s}{N\wedge T}} + \frac{s/\sqrt{N \wedge T}}{\log (p/\gamma)}\right)
     \end{align*}
where $c_0 := \frac{uc_1+1}{lc_1-1}>1$. For the degenerate case, the common penalty level $\lambda$ is rescaled so that we replace $N\wedge T$ with $NT$ in the above rates. This rate difference remains the same for the rest of the proof, and we will focus on the rate of the non-degenerate case. 

Let $a := \min_{j=1,...,p} \omega_{j}^{1/2}, \ b:= \max_{j=1,...,p} \omega_{j}^{1/2}$. As is shown in \cite{belloni2016inference}, 
\begin{align}
K_C^{\omega}(M_f) \geq \frac{1}{b} K_{bC/a}(M_f). \label{inequality_restricted} 
\end{align}
Under Assumption \hyperref[regcon]{REG}(v), $b<\infty$ and $1<b/a<\infty$. By \ref{inequality_restricted}, we have $1/K_{c_0}^{\omega}(M_f) \leq b/K_{\bar{C}}(M_f)$ where $\bar{C}: = bc_0/a$. By arguments given in \cite{bickel2009simultaneous}, Assumption \hyperref[sparse_eigen]{SE} implies that $1/K_{C}(M_f)=O_P(1)$ for any $C>0$. Therefore, 
\begin{align*}
    \left\Vert f_{it}'\left( \widehat{\zeta} - \zeta_0\right) \right\Vert_{NT,2} =  O_P\left( \sqrt{\frac{s\log(p/\gamma)}{N\wedge T}} \right), \ \ \left\vert \omega\left( \widehat{\zeta} - \zeta_0 \right) \right\vert_1 = O_P\left(s\sqrt{\frac{\log (p/\gamma)}{N \wedge T}}\right).
     \end{align*}
By Hölder's inequality and that $\min_j\omega_j \geq a>0$, 
\begin{align*}
    \Vert\widehat{\zeta}-\zeta_0\Vert_1 \leq \Vert \omega^{-1} \Vert_\infty \left\vert \omega\left( \widehat{\zeta} - \zeta_0 \right) \right\vert_1 = O_P\left(s\sqrt{\frac{\log (p/\gamma)}{N\wedge T}}\right) = O_P\left(s\sqrt{\frac{\log (p\vee NT)}{N\wedge T}}\right)
\end{align*}
where the first inequality follows from Hölder's inequality.

The $l2$ rate of convergence will be derived after the sparsity bounds. We now switch the focus to the Post-LASSO. By the finite sample bounds of Lemma \ref{lemma_a3}, we have
\begin{align}
    \Vert f(X_{it}) - f_{it}'\widehat{\zeta}_{PL} \Vert_{NT,2}= & \left(\sqrt{\frac{s}{\phi_{min}(s)(M_f)}} + \sqrt{\frac{\widehat{m}}{\phi_{min}(\widehat{m})(M_f)}}\right) O_P\left(\frac{\lambda}{NT}\right) \nonumber \\
    & + O_P\left(\Vert f(X_{it}) - \left(\mathcal{P}_{\widehat{\Gamma}}f\right)_{it}\Vert_{NT,2}\right), \label{post1}
\end{align}
By the finite sample bounds of Lemma 7 from \cite{belloni2012sparse}, we have
\begin{align}
    \Vert f_{it}'(\widehat{\zeta}_{PL}-\zeta_0)\Vert_{NT,2} \leq &  \Vert f_{it}(X_{it})-f_{it}'\widehat{\zeta}_{PL}\Vert_{NT,2} + \Vert r_{it} \Vert_{NT,2}, \label{post2} \\
    \Vert \omega(\widehat{\zeta}_{PL}-\zeta_0)\Vert_{1}\leq & \frac{b \sqrt{\widehat{m}+s}}{\sqrt{\phi_{min}(\widehat{m}+s)(M_f)}} \times \Vert f_{it}'(\widehat{\zeta}_{PL}-\zeta_0)\Vert_{NT,2} \label{post3}    \\
    \Vert f(X_{it}) - \mathcal{P}_{\widehat{\Gamma}} f(X_{it}) \Vert_{NT,2} \leq &  \left(u+\frac{1}{c_1}\right)\frac{\lambda\sqrt{s}}{NT K_{c_0}^{\omega}(M_f)} + 3 \Vert r_{it}\Vert_{NT,2} \label{post4} .
\end{align}
The finite sample bound of Lemma 8 from \cite{belloni2012sparse} gives
\begin{align*}
    \widehat{m}\leq \phi_{max}(\widehat{m})(M_f) a^{-2}  \left(\frac{2c_0\sqrt{s}}{K_{c_0}^{\omega}(M_f)} + \frac{6c_0NT \Vert r_{it}\Vert_{NT,2} }{\lambda}\right)^2.
\end{align*}
where $a>0$ has been shown previously. 

Let $\mathcal{M} = \left\{m\in \mathbb{N}: m> 2\phi_{max}(m)(M_f) a^{-2} \left(\frac{2c_0\sqrt{s}}{K_{c_0}^{\omega}(M_f)} + \frac{6c_0NT \Vert r_{it}\Vert_{NT,2} }{\lambda}\right)^2 \right\}$. Lemma 10 of \cite{belloni2012sparse} gives
\begin{align}
    \widehat{m} \leq \min_{m\in\mathcal{M}}\phi_{max}(m\wedge NT)(M_f)a^{-2} \left(\frac{2c_0\sqrt{s}}{K_{c_0}^{\omega}(M_f)} + \frac{6c_0NT \Vert r_{it}\Vert_{NT,2} }{\lambda}\right)^2. \label{lemma10}
\end{align}
Note that $\frac{6c_0NT \Vert r_{it}\Vert_{NT,2} }{\lambda\sqrt{s}} = O_P(1/\log(p\wedge NT)) \overset{p}{\rightarrow} 0 $. Recall that $1/K_{c_0}^{\omega}(M_f) \leq b/K_{\bar{C}}(M_f)<\infty$. Let $\mu := \min_{m}\left\{\sqrt{\phi_{max}(m)(M_f)/\phi_{min}(m)(M_f)}: m> 18\bar{C}^2 s \phi_{max}(m)(M_f)/K^2_{\bar{C}}(M_f)\right\}$, and let $\bar{m}$ be the integer associated with $\mu$. By the definition of $\mathcal{M}$, it implies that $\bar{m}\in \mathcal{M}$ with probability approaching one, which implies $\bar{m}>\widehat{m}$ due to \ref{lemma10}. By Lemma 9 (the sub-linearity of sparse eigenvalues) from \cite{belloni2012sparse} and \ref{lemma10}, we have
\begin{align*}
    \widehat{m}\lesssim_P s\mu^2\phi_{min}(\bar{m}+s)/K_{\bar{C}}^2\lesssim s\mu^2 \phi_{min}(\hat{m}+s)/K_{\bar{C}}^2. 
\end{align*}
Combining the results above with \ref{post1} and \ref{post4} to gives
\begin{align*}
    \Vert f(X_{it}) - f_{it}'\widehat{\zeta}_{PL})\Vert_{NT,2} = O_P\left( \sqrt{\frac{s\mu^2\log(p/\gamma)}{(N\wedge T)K_{\bar{C}}^2}} + \Vert r_{it}\Vert_{NT,2} + \frac{\lambda\sqrt{s}}{NT K_{c_0}^\omega(M_f)}\right).
\end{align*}
Recall that $b<\infty$ and Condition \hyperref[sparse_eigen]{SE} imply $1/K_{c_0}^\omega(M_f)\leq 1/K_{\bar{C}}(M_f) <\infty$. Then, Condition \hyperref[sparse_eigen]{SE}, Condition \hyperref[asm]{ASM} and the choice of $\lambda$ together imply
\begin{align*}
    \Vert f(X_{it}) - f_{it}'\widehat{\zeta}_{PL})\Vert_{NT,2} = O_P\left(\sqrt{\frac{s\log(p/\gamma)}{N\wedge T}}\right).
\end{align*}
For the $l1$ convergence rate, note that $\Vert \widehat{\zeta}_{PL} - \zeta_0\Vert_0 \leq \widehat{m} + s$. Then, applying Cauchy-Schwarz inequality to $ \Vert \widehat{\zeta}_{PL} - \zeta_0\Vert_1 = \sum_{j=1}^p |\widehat{\zeta}_{PL} - \zeta_0| = \sum_{j\in \{\widehat{\Gamma} \bigcup \Gamma_0 \}} |\widehat{\zeta}_{PL} - \zeta_0|$ gives
\begin{align*}
   \Vert \widehat{\zeta}_{PL} - \zeta_0\Vert_1\leq \sqrt{\widehat{m} + s} \Vert \widehat{\zeta}_{PL} - \zeta_0\Vert_2 
\end{align*}

To derive the convergence rates in $l2$-norm of the Post-LASSO estimator (the $l2$ rate for the LASSO estimator is obtained similarly), we will utilize the sparse eigenvalue condition and the prediction norm. If $\widehat{\zeta}_{PL}-\zeta_0 = 0$, then the conclusion holds trivially. Otherwise, define $ b = \left(\widehat{\zeta}_{PL}-\zeta_0\right) / \Vert\widehat{\zeta}_{PL}-\zeta_0\Vert^{-1}_2$. Then, we have $ \Vert b \Vert_2 = 1 $ and so $b\in \Delta(\widehat{m} + s)=\{\delta: \Vert\delta\Vert_0 = \widehat{m} + s, \Vert\delta\Vert_2=1\} $. By Assumption \hyperref[sparse_eigen]{SE}, we have
\begin{align*}
    0<\kappa_1\leq \phi_{min}(\widehat{m} + s)(M_f) \leq \frac{(b'M_fb)^{1/2}}{\Vert b \Vert_2} = \frac{\left\Vert f_{it}'\left( \widehat{\zeta}_{PL} - \zeta_0\right) \right\Vert_{NT,2}}{\Vert\widehat{\zeta}_{PL} - \zeta_0\Vert_2},
\end{align*}
Therefore, using the bound on the prediction norm above, we conclude that
\begin{align*}
    \Vert\widehat{\zeta}_{PL} - \zeta_0\Vert_2\leq \frac{\left\Vert f_{it}'\left( \widehat{\zeta}_{PL} - \zeta_0\right) \right\Vert_{NT,2}}{\kappa_1 } = O_P\left( \sqrt{\frac{s\log(p/\gamma)}{N\wedge T}} \right).
\end{align*}
It implies that $\Vert \widehat{\zeta}_{PL} - \zeta_0\Vert_1 = \sqrt{\widehat{m} + s} O_P\left( \sqrt{\frac{s\log(p/\gamma)}{N\wedge T}} \right) =O_P\left( \sqrt{\frac{s^2\log(p/\gamma)}{N\wedge T}} \right) $. 

\end{proof}

\noindent\textit{\textbf{Feasible Penalty Weights Validity}}\\
Let $\check\omega_j$ be the feasible penalty weights without estimation error, defined the same way as \ref{feasible_w} with $\hat v$ replaced by $v$. The validity of the feasible penalty weights is given by the following Lemmas \ref{weight_valid_1} and \ref{weight_valid_2}.

\begin{lemma}
    \label{weight_valid_1}
    Under Assumptions \hyperref[ahk]{AHK} and \hyperref[regcon]{REG}, and additionally, (i) $M=o(T^{1/4})$, (ii) $E[U_{it}|X_{i't'}]=0$ for any $i,t,i',t'$, (iii) $\sup_{i,t}\vert v_{it,j}\vert$ bounded by some absolute constant, then for each $j=1,...,p$,  
    \begin{enumerate}
        \item[(i)] If  $\emph{E}\left(a_{i,j}^2\right) + \emph{E}\left(g_{t,j}^2\right) >\epsilon$ for some $\epsilon>0$, then
        \begin{align*}
             \omega_{a,j}^2 - \left(\check\omega_{a,j}^2-\check\omega_{e,j}^2\right) &= o_P(1) \\
             \omega_{g,j}^2 - \left(\check\omega_{g,j}^2-\check\omega_{e,j}^2\right) &= o_P(1) \\
             \omega_{e,j}^2-  \check\omega_{e,j}^2 &=o_P(1)
        \end{align*}
        \item[(ii)] If $f_{it,j}V_{it}$ is i.i.d, then
        \begin{align*}
             T\left[\omega_{a,j}^2 - \left(\check\omega_{a,j}^2-\check\omega_{e,j}^2\right)\right] &= o_P(1) \\
             T\left[\omega_{g,j}^2 - \left(\check\omega_{g,j}^2-\check\omega_{e,j}^2\right)\right] &= o_P(1) \\
             T\left[\omega_{e,j}^2-  \check\omega_{e,j}^2 \right]&=o_P(1)
        \end{align*}
    \end{enumerate}
\end{lemma}

\begin{lemma}\label{weight_valid_2} 
   Under the same setup as Theorem \ref{performance_bound}, and, additionally, (i) $\sup_{i,t} |f_{it,j}|\leq \bar F$, (ii) the initial penalty weights satisfy Condition \ref{omegacon}, (iii) $M=o\left(\sqrt{\frac{N\wedge T}{s\log(p/\gamma)}}\right)$, then the refined penalty weight components have the following properties: 
   \begin{align*}
       \hat\omega_{a,j}^2-\check\omega_{a,j}^2 = o_P(1),\\
       \hat\omega_{g,j}^2-\check\omega_{g,j}^2 = o_P(1),\\
       \hat\omega_{e,j}^2-\check\omega_{e,j}^2 = o_P(1)
   \end{align*}
\end{lemma}

The proofs of Lemmas \ref{weight_valid_1} and \ref{weight_valid_2} are given in  the Online Supplementary Material. \\

\noindent\textit{\textbf{Cross-Fitting Validity}}\\ 
Let $(X,Y)$ be random elements taking values in Euclidean space $S=(S_1\times S_2)$ with probability laws ${P}_X$ and ${P}_Y$, respectively. Let $\Vert\nu\Vert_{TV}$ denote the total variation norm of a signed measure $\nu$ on a measurable space $(S,\Sigma)$ where $\Sigma$ is a $\sigma$-algebra on $S$:
\begin{align*}
    \Vert \nu \Vert_{TV} = \sup_{A\in \Sigma} \nu(A) - \nu(A^c).
\end{align*}
The following lemma is quoted from \cite{Semenova2023}(Lemma A.3): 

\begin{lemma}
\label{lemma_a4}    
Let $(X,Y)$ be random element taking values in Polish space $S=(S_1\times S_2)$ with laws ${P}_X$ and ${P}_Y$, respectively. Then, we can construct $(\tilde{X},\tilde{Y})$ taking values in $(S_1, S_2)$ such that (i) they are independent of each other; (ii) their laws $\mathcal{L}(\tilde{X}) = {P}_X$ and $\mathcal{L}(\tilde{Y}) = {P}_Y$; (iii)
\begin{align*}
    {\rm P}\left\{ (X,Y)\ne (\tilde{X},\tilde{Y}) \right\} = \frac{1}{2}\Vert {P}_{X,Y} - {P}_X\times {P}_Y  \Vert_{TV} 
\end{align*}
\end{lemma}
The proof is provided in the supplementary material of \cite{Semenova2023}. To apply the independence coupling result for cross-fitting in the panel data, we need to introduce another lemma:

\begin{lemma}
    \label{lemma_a5}
Let $X_1,...,X_q$ and $Y$ be random elements taking values in Polish space $S=(S_1\times...\times S_m\times S_y)$. 
\begin{align*}
    \beta((X_1,...,X_m), Y)\leq \sum_{i=1}^q \beta(X_i,Y).
\end{align*}
\end{lemma}

\begin{proof}[\textbf{Proof of Lemma \ref{lemma_a5}}]
By Lemma \ref{lemma_a4}, we have
\begin{align*}
    \beta((X_1,...,X_m), Y) & = \frac{1}{2} \left\Vert P_{(X_1,...,X_q), Y} - P_{(X_1,...,X_m)}\times P_{Y}
 \right\Vert_{TV} \\
    & = {\rm P}((X_1,...,X_m,Y)\ne (\tilde{X}_1,...,\tilde{X}_m,\tilde{Y})) \leq  \sum_{i=1}^m {\rm P}((X_i,Y)\ne (\tilde{X}_i,\tilde{Y})) = \sum_{i=1}^m \beta(X_i,Y),
\end{align*}
where the inequality follows from the union bound.
\end{proof}

Now we can prove Lemma \ref{indep_couple} from the main body of the paper:
\begin{proof}[\textbf{Proof of Lemma \ref{indep_couple}}]

By Lemma \ref{lemma_a4}, for each $(k,l)$ we have
\begin{align*}
  &  {\rm P}\{(W(k,l),W(-k,-l))\ne(\tilde{W}(k,l),\tilde{W}(-k,-l))\} \\
= &\beta\left(W(k,l),W(-k,-l)\right)=\beta\left(\{W_{it}\}_{i\in I_k, t\in S_l},\bigcup_{k'\ne k, l'\ne l, l\pm 1}\{W_{it}\}_{i\in I_{k'}, t\in S_{l'}}\right) \\
\leq & \sum_{i\in I_k} \beta\left(\{W_{it}\}_{t\in S_l},\bigcup_{k'\ne k, l'\ne l, l\pm 1}\{W_{it}\}_{i\in I_{k'}, t\in S_{l'}}\right) \leq  \sum_{k'\ne k, l'\ne l, l\pm 1}\sum_{j\in I_{k'}} \sum_{i\in I_k} \beta\left(\{W_{it}\}_{t\in S_l},\{W_{jt}\}_{t\in S_{l'}}\right)
\end{align*}
where the last two inequalities follow from Lemma \ref{lemma_a5}. Note that for $s,m\geq 1$, we have
\begin{align*}
    &\beta(\{W_{it}\}_{t\leq s},\{W_{jt}\}_{t\geq s+m}) \\
    = &\left\Vert P_{\{W_{it}\}_{t\leq s},\{W_{jt}\}_{t\geq s+m}} - P_{\{W_{it}\}_{t\leq s}}\times P_{\{W_{jt}\}_{t\geq s+m}} \right\Vert_{TV} 
     \leq \sup_{A\in \sigma(\{W_{jt}\}_{t\geq s+m})} {\rm E} \left\vert {\rm P}(A|\sigma(\{W_{it}\}_{t\leq s})) - {\rm P}(A)\right\vert \\
     =&\sup_{A\in \sigma(\{W_{jt}\}_{t\geq s+m})} {\rm E} \left\vert {\rm P}({\rm P}(A|\sigma(\alpha_i,\{\gamma_t\}_{t\leq s}, \{\varepsilon_{it}\}_{t\leq s}))|\sigma(\{W_{it}\}_{t\leq s})) - {\rm P}(A)\right\vert\\
     = &\sup_{A\in \sigma(\{W_{jt}\}_{t\geq s+m})} {\rm E} \left\vert {\rm P}(A|\sigma(\{\gamma_t\}_{t\leq s}) - {\rm P}(A)\right\vert  = \sup_{A\in \sigma(\{\gamma_{t}\}_{t\geq s+m})} {\rm E} \left\vert {\rm P}(A|\sigma(\{\gamma_t\}_{t\leq s}) - {\rm P}(A)\right\vert \leq c_{\kappa}exp(-\kappa m), 
\end{align*}
where the last inequality follows from Assumption \ref{ahk}. Therefore, 
\begin{align*}
    {\rm P}\{(W(k,l),W(-k,-l))\ne(\tilde{W}(k,l),\tilde{W}(-k,-l))\} \leq KLN^2c_{\kappa}exp(-\kappa  T_l),
\end{align*}
which in turn gives
\begin{align*}
  {\rm P}\{(W(k,l),W(-k,-l))\ne(\tilde{W}(k,l),\tilde{W}(-k,-l)),\ for \ some \ (k,l) \} \leq K^2L^2N^2c_{\kappa}exp(-\kappa  T_l),
\end{align*}
where $ T_l = T/L$. Given that $\log(N)/T=o(1)$ and $(K,L)$ are finite, it follows that
\begin{align*}
  {\rm P}\{(W(k,l),W(-k,-l))\ne(\tilde{W}(k,l),\tilde{W}(-k,-l)),\ for \ some \ (k,l) \} = o(1)
\end{align*}
\end{proof}

\pagebreak
\setcounter{page}{1}

\begin{center}
    \LARGE Online Supplementary Material \\
    \Large Inference in High-Dimensional Panel Models: \\
Two-Way Dependence and Unobserved Heterogeneity
\end{center}

This Online Supplementary Material is structured as follows. 
Section \ref{app_tw} collects extra lemmas and proofs the two-way cluster-LASSO. Section \ref{app_1} provides a formal discussion and statistical analysis of the panel DML method. Section \ref{app_2} collects proofs of the theorems of in Section \ref{partial_linear_unob}.

\subsection{Extra proofs for two-way cluster-LASSO}
\label{app_tw}

\begin{lemma}
    \label{lemma_a3}
    Under Assumption \hyperref[asm]{ASM}, if $S_{\max}:=\max_{1\leq j\leq p} \vert\mathbb{E}_{NT} [\omega^{-1}f_{it,j}V_{it}]\vert \leq \frac{\lambda}{2c_1NT}$, $0<a = \min_{j} \omega \leq \max_{j}\omega= b<\infty$, and $u\geq 1 \geq l \geq 1/c_1$, then
\begin{align*}
    &\Vert f(X_{it}) - f_{it}'\widehat{\zeta}_{PL} \Vert_{NT,2} \\
    = & \left(\sqrt{\frac{s}{\phi_{min}(s)(M_f)}} + \sqrt{\frac{\widehat{m}}{\phi_{min}(\widehat{m})(M_f)}}\right) O_P\left(\frac{\lambda}{NT}\right) \nonumber + O_P\left(\Vert f(X_{it}) - \left(\mathcal{P}_{\widehat{\Gamma}}f\right)_{it}\Vert_{NT,2}\right).
\end{align*}
\end{lemma}

\begin{proof}[\textbf{Proof of Lemma \ref{lemma_a3}}]
We can decompose $f(X_{it}) - f_{it}'\widehat{\zeta}_{PL}$ as follows:
\begin{align*}
    & f(X_{it}) - f_{it}'\widehat{\zeta}_{PL} = f(X_{it}) - (\mathcal{P}_{\widehat{\Gamma}}Y)_{it} = \left( (I_{NT} - \mathcal{P}_{\widehat{\Gamma}}) f(X) - \mathcal{P}_{\widehat{\Gamma}} V\right)_{it} \\  
    = & \left( (I_{NT} - \mathcal{P}_{\widehat{\Gamma}}) f - (\mathcal{P}_{\widehat{\Gamma}\backslash \Gamma_0}+\mathcal{P}_{\Gamma_0}) V\right)_{it} \leq \left\Vert \left(I_{NT} - \mathcal{P}_{\widehat{\Gamma}}) f\right)_{it} \right\Vert_{NT,2}   +  \left\Vert \left(\mathcal{P}_{\Gamma_0}V\right)_{it} \right\Vert_{NT,2} +\left\Vert \left( \mathcal{P}_{\widehat{\Gamma}\backslash \Gamma_0}V\right)_{it} \right\Vert_{NT,2}.
\end{align*}
where the last equality follows from the property of the linear projection and the inequality follows from Minkowski's inequality. By H\"{o}lder's inequality and the property of spectral norm, we have
\begin{align*}
    &\left\Vert \left( \mathcal{P}_{\widehat{\Gamma}\backslash \Gamma_0}V\right)_{it} \right\Vert_{NT,2} = \frac{1}{\sqrt{NT}} \left\Vert  \mathcal{P}_{\widehat{\Gamma}\backslash \Gamma_0}V \right\Vert_{2} \leq  \frac{1}{\sqrt{NT}} \left\Vert f_{\widehat{\Gamma}\backslash \Gamma_0} (f'_{\widehat{\Gamma}\backslash \Gamma_0}f_{\widehat{\Gamma}\backslash \Gamma_0})^{-1} \right\Vert_{\infty} \left\Vert f_{\widehat{\Gamma}\backslash \Gamma_0}' V \right\Vert_{2}  \\
    \leq & \frac{1}{\sqrt{NT}} \sqrt{\frac{1}{NT\phi_{min}(\widehat{m})(M_f)}} \left( \sum_{j\in \widehat{\Gamma}\backslash \Gamma_0} \left(\sum_{i=1}^N\sum_{t=1}^T f_{it,j}V_{it}\right)^{2}\right)^{1/2}\leq \sqrt{\frac{\widehat{m}}{\phi_{min}(\widehat{m})(M_f)}} S_{\max} \\
    = & \sqrt{\frac{\widehat{m}}{\phi_{min}(\widehat{m})(M_f)}} O_P\left(\frac{\lambda}{NT}\right)
\end{align*}
where the last line follows from $ \min_j \omega_j = a>0$ and $S_{\max} \leq \frac{\lambda}{2c_1NT}$. By similar arguments, we have
\begin{align*}
    &\left\Vert \left(\mathcal{P}_{\Gamma_0}V\right)_{it} \right\Vert_{NT,2} = \frac{1}{\sqrt{NT}} \left\Vert  \mathcal{P}_{\Gamma_0}V \right\Vert_{2} \leq  \frac{1}{\sqrt{NT}} \left\Vert f_{ \Gamma_0} (f'_{\Gamma_0}f_{\Gamma_0})^{-1} \right\Vert_{\infty} \left\Vert f_{ \Gamma_0}' V \right\Vert_{2}  \\
    \leq & \frac{1}{\sqrt{NT}} \sqrt{\frac{1}{NT\phi_{min}(s)(M_f)}} \left( \sum_{j\in \Gamma_0} \left(\sum_{i=1}^N\sum_{t=1}^T f_{it,j}V_{it}\right)^{2}\right)^{1/2}\leq  \sqrt{\frac{s}{\phi_{min}(s)(M_f)}}  O_P\left(\frac{\lambda}{NT}\right).
\end{align*}\end{proof}

\begin{proof}[\textbf{Proof of Lemma \ref{weight_valid_1}}]
   We first consider the case (ii). In this case, since $a_{i,j} = g_{t,j} = 0$ and $e_{it,j} = v_{it,j}$, it is trivially true that $\Sigma_{e,j} = \sigma_{e,j}^2 = E[v_{it,j}^2]$ and
   \begin{align*}
       T\omega_{a,j}^2 = 0, \ \ T\omega_{g,j}^2 =0, \ \  T\omega_{e,j}^2 = \frac{1}{NT}\sum_{i=1}^N\sum_{t=1}^T\sum_{s=1}^T k\left(\frac{|t-s|}{M}\right)v_{it,j}v_{is,j}.
   \end{align*}
    Due to the independence over $i,j$, we have
    \begin{align*}
        E[ T\omega_{e,j}^2] =E[v_{it,j}^2]. 
    \end{align*}
And 
\begin{align*}
    Var(T\omega_{e,j}^2 )\lesssim E[v_{it,j}^4]/N +  E[v_{it,j}^4]/(NT)
\end{align*}
Then, by Chebyshev inequality, we have $T\omega_{e,j}^2 = E[v_{it,j}^2]+o_P(1)$. Then, we consider the feasible penalty weights without estimation errors:
\begin{align*}
     E\left[  T\check\omega_{a,j}^2\right] =&E\left[ \frac{1}{NT} \sum_{i=1}^N\left(\sum_{t=1}^T v_{it,j}\right)^2\right] = E[v_{it,j}^2] \\
       E\left[  T\check\omega_{g,j}^2\right] =&E\left[ \frac{1}{NT} \sum_{t=1}^T \sum_{s=1}^T k\left(\frac{|t-s|}{M}\right) \left(\sum_{i=1}^N {v}_{it,j}\right)\left(\sum_{i=1}^N {v}_{is,j}\right)\right] = E[v_{it,j}^2]\\
       Var(T\check\omega_{a,j}^2) \lesssim & E[v_{it,j}^4]/N + E[v_{it,j}^4]/(NT)  \\
          Var(T\check\omega_{a,j}^2) \lesssim & E[v_{it,j}^4]/N +E[v_{it,j}^4]/T + E[v_{it,j}^4]/(NT)
\end{align*}
Then, by Chebyshev inequality, we have $T\check\omega_{a,j}^2 = E[v_{it,j}^2] +o_P(1)$ and $T\check\omega_{a,j}^2 = E[v_{it,j}^2] +o_P(1)$
Furthermore, 
\begin{align*}
         E\left[  T\check\omega_{e,j}^2\right] =&E\left[\frac{1}{NT}\sum_{i=1}^N\sum_{t=1}^T \sum_{s=1}^T k\left(\frac{|t-s|}{M}\right) \tilde{e}_{it,j}\tilde{e}_{is,j}\right]  =  E[T\omega_{e,j}^2 ] +O(1/N) + O(1/T) + o(1)\\
         Var(T\check\omega_{e,j}^2) \lesssim &E[v_{it,j}^4]/N
\end{align*}
By Chebyshev inequality, we have $ T\check\omega_{e,j}^2 = E[v_{it,j}^2]$. Then, the second statement of the lemma is proved.

Now we consider the first case. By Markov inequality and the fact that $a_{i,j}$ is i.i.d over $i$, we have
    \begin{align*}
        P\left( \left|\frac{1}{N}\sum_{i=1}^{N} {a}_{i,j}^2 - \sigma_a^2 \right| >\epsilon\right) \leq Var\left(\frac{1}{N}\sum_{i=1}^{N} {a}_{i,j}^2\right)/\epsilon^2 =  E\left[a_{i,j}^4\right] /(N\epsilon)\leq E\left[v_{it,j}^4\right] /(N\epsilon)
    \end{align*}
    So, we have $\omega_{a,j}^2 =\sigma_{a,j}^2 + O_P(1/\sqrt{N})$. Next, consider $\check\omega_{a,j}^2 =\frac{1}{N}\sum_{i=1}^{N} \tilde{a}_{i,j}^2 $. We note that
    \begin{align*}
\left(a_i+\frac{1}{T}\sum_{t=1}^{T}e_{it}\right)-\tilde{a}_i
&= \frac{1}{T}\sum_{t=1}^{T}g_t =: R_{g,NT},
\end{align*}
which does not depend on $i$ and $R_{g,NT}=O_P\left(1/\sqrt{T}\right)$. Then
\begin{align*}
\frac{1}{N}\sum_{i=1}^{N}\tilde{a}_i^2
&= \frac{1}{N}\sum_{i=1}^{N}\left[\left(a_i+\frac{1}{T}\sum_{t=1}^{T}e_{it}\right)-R_{g,NT}\right]^2 \\
&= \frac{1}{N}\sum_{i=1}^{N}\left(a_i+\frac{1}{T}\sum_{t=1}^{T}e_{it}\right)^2
   -\frac{2R_{g,NT}}{N}\sum_{i=1}^{N}\left(a_i+\frac{1}{T}\sum_{t=1}^{T}e_{it}\right)
   +R_{g,NT}^2 .
\end{align*}
Note that $a_i+\frac{1}{T}\sum_{t=1}^{T}e_{it}$ is independent over $i$
conditional on $\{\gamma_t\}$, and that
\begin{align*}
E\left(a_i+\frac{1}{T}\sum_{t=1}^{T}e_{it}\right)^2
&= \sigma_a^2+\Sigma_e (1+ o(1/T))/T = \sigma_a^2+o(1),\\
\operatorname{Var}\left(a_i+\frac{1}{T}\sum_{t=1}^{T}e_{it}\right)^2
&\leq E\left(a_i+\frac{1}{T}\sum_{t=1}^{T}e_{it}\right)^4
   \leq Ea_i^4+Ee_{it}^4<\infty .
\end{align*}
Moreover,
\begin{align*}
\operatorname{Var}\left(\frac{1}{N}\sum_{i=1}^{N}
\left(a_i+\frac{1}{T}\sum_{t=1}^{T}e_{it}\right)^2
\,\middle|\,\{\gamma_t\}\right)
&= \frac{1}{N}\operatorname{Var}\left(\left(a_i+\frac{1}{T}\sum_{t=1}^{T}e_{it}\right)^2
\,\middle|\,\{\gamma_t\}\right) \\
&\le \frac{1}{N}E\left(\left(a_i+\frac{1}{T}\sum_{t=1}^{T}e_{it}\right)^4
\,\middle|\,\{\gamma_t\}\right) \\
&\le \frac{1}{N}
\left(\left(Ea_i^4\right)^{1/4}
     +\left(E\left[e_{it}^4\mid\gamma_t\right]\right)^{1/4}\right)^4 .
\end{align*}
By Jensen’s inequality,
\[
E\left[\left(E\left[e_{it}^4\mid\gamma_t\right]\right)^{1/4}\right]
\leq \left(E\left[E\left[e_{it}^4\mid\gamma_t\right]\right]\right)^{1/4}
= E\left[e_{it}^4\right]^{1/4} \leq E[v_{it,j}^4]^{1/4}.
\]
Then, by the conditional Markov inequality and Fubini’s theorem, we have
\[
\frac{1}{N}\sum_{i=1}^{N}\left(a_i+\frac{1}{T}\sum_{t=1}^{T}e_{it}\right)^2
=\left(\sigma_a^2+\frac{1}{T}\Sigma_e\right)+O_P\left(\frac{1}{\sqrt{N}}\right).
\]
And so 
\begin{align*}
    \check\omega_{a,j}^2 = \left(\sigma_a^2+\frac{1}{T}\Sigma_e\right)+O_P\left(\frac{1}{\sqrt{N}}\right) + O_P\left(\frac{1}{\sqrt{NT}} \right)+ O_P\left(\frac{1}{T}\right) = \sigma_a^2+o_P(1).
\end{align*}

For $\omega_{g,j}^2$, we can 
apply Proposition 2 of \cite{bester2008inference} by verifying Assumption 7 from the same paper. Since we can take the block size $h = round(T^{1/5}) + 1$, it diverges with the time sample size and $h/T\to 0$ as $T\to \infty$. So, Assumption 7(i) follows. Note that the $\beta$-mixing property of $g_{t,j}$ implies that it is also $\alpha$-mixing with the mixing coefficient $\alpha_{g}(q) \leq \beta_{g}(q) \leq \beta_{\gamma}(q) = c_\kappa \text{exp}(-\kappa q) $ for all $q\geq 1$. Let  $\zeta$ be some positive constant, then we have
\begin{align*}
    \sum_{q=1}^\infty q^2 \alpha_g(q)^{\zeta / (4+\zeta)} \leq c_\kappa^{\zeta/(4+\zeta)} \sum_{q=1}^\infty q^2  \text{exp}(- \kappa\zeta q /(4+\zeta) ) = c_\kappa^{\zeta/(4+\zeta)} \sum_{q=1}^\infty q^2  \text{exp}(- \frac{\kappa\zeta}{4+\zeta} q )
\end{align*}
We can use the ratio test to examine the convergence of sum: 
\begin{align*}
    \text{lim}_{q\to\infty} \frac{(q+1)^2\text{exp}(-\frac{\kappa\zeta}{4+\zeta}(q+1))}{q^2\text{exp}(-\frac{\kappa\zeta}{4+\zeta}q)} =\text{lim}_{q\to\infty} \left(\frac{q+1}{q}\right) \text{exp}(-\frac{\kappa\zeta}{4+\zeta}) = \text{exp}(-\frac{\kappa\zeta}{4+\zeta})
\end{align*}
Since $\kappa>0$ and $\zeta>0$, we have $\frac{\kappa\zeta}{4+\zeta}>0$ and so $\text{exp}(-\frac{\kappa\zeta}{4+\zeta})<1$. Thus we conclude the infinite sum does not diverge to infinity. The third condition is ensured by our assumptions directly. Thus, by Proposition 2 of \cite{bester2008inference}, we have
\begin{align*}
   \omega_{g,j}^2 =  \frac{N}{T}\frac{1}{T} \sum_{b=1}^B\left(\sum_{t\in H_b} g_{t,j}\right)^2 = c\Sigma_{g,j} + o_P(1).
\end{align*}
Therefore, we conclude that $\omega_{j} =  \tilde{\omega}_j^{\rm CHS} + o_{P}(1)$. 

Next, we consider $\check\omega_{g,j}^2$. We note that $g_t + \frac{1}{N}\sum_{i=1}^Ne_{it,j} - \tilde g_t = \frac{1}{N}\sum_{i=1}^Na_i=: R_{a,NT} = O_P(1/\sqrt{N})$, which does not depend on $i$ or $t$. Then, we have 
\begin{align*}
 T/N  \check\omega_{g,j}^2= & \frac{1}{T} \sum_{t=1}^T \sum_{s=1}^T k\left(\frac{|t-s|}{M}\right) \left(g_t + \frac{1}{N}\sum_{i=1}^Ne_{it,j} - R_{a,NT}\right) \left(g_s+ \frac{1}{N}\sum_{i=1}^Ne_{is,j}- R_{a,NT} \right)\\
   =& \frac{1}{T} \sum_{t=1}^T\sum_{s=1}^Tk\left(\frac{|t-s|}{M}\right) \left(g_t+\frac{1}{N}\sum_{i=1}^Ne_{it,j}\right)\left(g_s +\frac{1}{N}\sum_{i=1}^Ne_{is,j}\right)+ o_P(1) 
\end{align*}
where the last $o_P(1)$ term follows from $M=o(T^{1/4})=o(\sqrt{T})$ and $R_{a,NT} = O_P(1/\sqrt{N})$. Due to the mutual uncorrelatedness and that $e_{it,j}$ is uncorrelated over $i$, the long-run variance of $g_t+\frac{1}{N}\sum_{i=1}^Ne_{it,j}$ can be written as 
\begin{align*}
    \frac{1}{T}\sum_{t=1}^T \sum_{s=1}^T E(g_t+\frac{1}{N}\sum_{i=1}^Ne_{it,j}) (g_s+\frac{1}{N}\sum_{i=1}^Ne_{is,j}) = \Sigma_{g,j}(1+o(1) ) + \Sigma_{e,j}(1+o(1) )/N =\Sigma_{g,t}+o(1)
\end{align*}
Note that conditional on $\{\alpha_i\}$, $g_t+\frac{1}{N}\sum_{i=1}^Ne_{it,j}$ is beta-mixing at a geometric rate under Assumption \hyperref[ahk]{AHK}. We also note that $\frac{1}{T} \sum_{t=1}^T\sum_{s=1}^Tk\left(\frac{|t-s|}{M}\right) g_t\frac{1}{N}\sum_{i=1}^Ne_{is,j} $ is mean zero. Let $\bar e_{s,j} =\frac{1}{N}\sum_{i=1}^Ne_{it,j} $, then we can write
\begin{align*}
   \frac{1}{T} \sum_{t=1}^T\sum_{s=1}^Tk\left(\frac{|t-s|}{M}\right) g_t\bar e_{s,j} = \frac{1}{T}\sum_{l=-M}^{M} k(|l|/M) \sum_{t=1}^{T-|l|} g_{t,j}\bar e_{t+l,j}
\end{align*}
Consider $z_{t,j}(l):=  g_{t,j}\bar e_{t+l,j}$, which is also beta-mixing sequence with the geometric rate.  For each $l$, we can apply Theorem 14.13 of \cite{hansen2022econometrics}
\begin{align*}
    Var\left( \frac{1}{T}  \sum_{t=1}^{T-|l|} z_{t,j}(l)|\{\alpha_i\}\right) = (E[z_{t,j}^4(l)| \{\alpha_i\} ])^{1/2} O(1/T)
\end{align*}
By Jensen's inequality and iterated expectation, we have $E[(E[z_{t,j}^4(l)| \{\alpha_i\} ])^{1/2}]\leq E[z_{t,j}^4(l)]^{1/2} \leq\infty$. It follows that, $    Var\left( \frac{1}{T}  \sum_{t=1}^{T-|l|} z_{t,j}(l)\right)=O(1/T)$. Let $\bar z_{j}(l)= \frac{1}{T}  \sum_{t=1}^{T-|l|} z_{t,j}(l)$ By Holder's inequality, we have
\begin{align*}
 &   Var\left(\frac{1}{T}\sum_{l=-M}^{M} k(|l|/M) \bar z_{j}(l)\right)  =E\left[ \sum_{l=-M}^M k\left(\frac{|l|}{M}\right)\bar z_{j}(l)\right]^2\\
 =&\sum_{l=-M}^M k\left(\frac{|l|}{M}\right)\sum_{l'=-M}^M k\left(\frac{|l'|}{M}\right)E\left[ \bar z_{j}(l)\bar z_{j}(l')\right] \leq \sum_{l=-M}^M k\left(\frac{|l|}{M}\right)\sum_{l'=-M}^M k\left(\frac{|l'|}{M}\right) \sqrt{Var(\bar z_{j}(l'))Var(\bar z_{j}(l))} \\
&=O(M^2/T)=o(1).
 \end{align*}
 Therefore, we have $   \frac{1}{T} \sum_{t=1}^T\sum_{s=1}^Tk\left(\frac{|t-s|}{M}\right) g_t\bar e_{s,j}  = o_P(1)$. Similarly, we can show $\frac{1}{T} \sum_{t=1}^T\sum_{s=1}^Tk\left(\frac{|t-s|}{M}\right) \bar e_{t,j} \bar e_{s,j} = \Sigma_e/N$. With the extra requirement on $M=o(T^{1/4})$, we can verify the conditions of Theorem 2 from \cite{newey1987simple}, we have  $\frac{1}{T} \sum_{t=1}^T\sum_{s=1}^Tk\left(\frac{|t-s|}{M}\right)  g_{t,j}  g_{s,j} = \Sigma_g + o_P(1).$ It follows that $\check \omega_{g,j}^2 - \omega_{g,j}^2= o_P(1)$. 

 Lastly, consider $ \omega_{e,j}^2 -  \check \omega_{e,j}^2$. Under the extra strict heterogeneity and uniform bound on the scores, we can apply the results of Theorem 5 from \cite{CHS_Restat}: $\max_{i}|a_{i,j}-\frac{1}{T}\sum_{t=1}^Tv_{it,j} |=o_P(1) $ and $\max_{i}|g_{t,j}-\frac{1}{N}\sum_{i=1}^Nv_{it,j} |=o_P(1) $. Therefore, 
 \begin{align*}
 \check \omega_{e,j}^2 = &\frac{1}{NT^2}\sum_{i=1}^N \left(\sum_{t=1}^T  e_{it} + a_{i,j}-\frac{1}{T}\sum_{t=1}^Tv_{it,j} +  g_{t,j}-\frac{1}{N}\sum_{i=1}^Nv_{it,j} \right)^2 \\
 =&  \frac{1}{NT^2}\sum_{i=1}^N \left(\sum_{t=1}^T  e_{it}  \right)^2 + o_P(1) = \omega_{e,j}^2 + o_P(1)
 \end{align*}
 This finishes the proof of the first case and therefore completes the proof.

    \end{proof}

    \begin{proof}[\textbf{Proof of Lemma \ref{weight_valid_2}}]
    We first consider $\hat\omega_{g,j}^2-\check\omega_{g,j}^2$. We can rewrite $\hat\omega_{g,j}^2$ (and similar for $\check\omega_{g,j}^2$) as
    \begin{align*}
        \hat\omega_{g,j}^2 = \frac{1}{NT^2} \sum_{l=-M}^M\sum_{t=1}^{T-|l|}k\left(\frac{|l|}{M}\right)\sum_{i=1}^N\sum_{i'=1}^N \hat v_{it,j}\hat v_{i's,j}  
    \end{align*}
     By product decomposition, triangle inequality, and Cauchy-Schwarz inequality, we can decompose the difference as follows:
     \begin{align*}
        |\sum_{t=1}^{T-|l|}k\left(\frac{|l|}{M}\right)\sum_{i=1}^N\sum_{i'=1}^N \hat v_{it,j}\hat v_{i's,j} - \sum_{t=1}^{T-|l|}k\left(\frac{|l|}{M}\right)\sum_{i=1}^N\sum_{i'=1}^N  v_{it,j} v_{i's,j}| \lesssim R_{NT}\left( ||v_{it,j}||_{NT,2} + R_{NT}\right) 
     \end{align*}
    where, with $\hat\varepsilon_{it} := f_{it}'(\hat\beta-\beta)-r_{it}$, 
    \begin{align*}
        R_{NT}:= \Vert \hat v_{it,j} - v_{it,j} \Vert_{NT,2} = \frac{1}{NT} \left(\sum_{i=1}^N\sum_{t=1}^T (f_{it,j} \hat\varepsilon_{it})^2\right)^{1/2}\leq \bar F \Vert \hat\varepsilon_{it}\Vert_{NT,2} \leq F (\Vert f_{it}'(\hat\beta-\beta)\Vert_{NT,2}+ \Vert r_{it}\Vert_{NT,2}) 
    \end{align*}
    Under Assumption \ref{asm}, we have $\Vert r_{it}\Vert_{NT,2} = O_P\left(\sqrt{\frac{s}{N\wedge T}}\right)$. Since the initial weights satisfies Condition \ref{omegacon}, we have from Theorem \ref{performance_bound} that
    $\Vert f_{it}'(\hat\beta-\beta)\Vert_{NT,2} = O_P\left(\sqrt{\frac{s\log(p/\gamma)}{N\wedge T}}\right)$. Then, we have, under $M=o\left(\sqrt{\frac{N\wedge T}{s\log(p/\gamma)}}\right)$,
    \begin{align*}
        |\hat\omega_{g,j}^2-\check\omega_{g,j}^2|\lesssim O_P\left(\sqrt{\frac{sM^2\log(p/\gamma)}{N\wedge T}}\right) = o_P(1)
    \end{align*}
    We can show similarly for the other two difference, with less stringent conditions. 
\end{proof}

\subsection{Results for the Panel DML Method}
\label{app_1}
In this section, a formal study of the panel DML method as an extension of the prototypical DML approach from \cite{Chernozhukov2018} is provided. Differing from the existing literature, the approach in this paper is made robust to two-way cluster dependence characterized by Assumption \hyperref[ahk]{AHK}. To restrict the focus,\footnote{Although the panel DML procedure also works well for the i.i.d data, a formal result is not given in this paper, and the main theorem excludes such cases. The DML inference procedure for the i.i.d case has been studied in \cite{Chernozhukov2018} with the only difference in the cross-fitting scheme and the variance estimation. The proposed cross-fitting scheme is trivially valid for the i.i.d case, and the two-way cluster robust variance estimators considered in this paper, which are adapted from \cite{CHS_Restat,chen2024fixed}, can also be shown valid under i.i.d assumption following similar arguments in \cite{CHS_Restat,chen2024fixed}.} I will assume a non-degeneracy condition. First, I define the decomposition components and their corresponding (long-run) variance-covariance matrices as follows:
\begin{alignat*}{2}
a_i &:= {\rm E}\left[\psi\left(W_{it};\theta_0,\eta_0\right)|\alpha_i\right], \ \ & & \Sigma_a  := {\rm E}[a_ia_i'], \\
g_t &:= {\rm E}\left[\psi\left(W_{it};\theta_0,\eta_0\right)|\gamma_t\right], \ \ & & \Sigma_g  := \sum_{l=-\infty}^{\infty}{\rm E}[g_tg_{t+l}'] , \\
e_{it} &:= \psi\left(W_{it};\theta_0,\eta_0\right) - a_i - g_t, \ \ & &\Sigma_e := \sum_{l=-\infty}^{\infty}{\rm E}[e_{it}e_{i,t+l}'].
\end{alignat*}
Let $\lambda_{min}[.]$ denote the smallest eigenvalue of a square matrix. The next assumption specifies the non-degeneracy condition and it implies that at least one of the components drives the cluster dependence. 
\begin{assumption_nd}[Non-Degeneracy]
\label{nondeg} Either $\lambda_{min}[\Sigma_a]>0$ or $\lambda_{min}[\Sigma_g]>0$. 
\end{assumption_nd}

The next two assumptions follow the same format as \cite{Chernozhukov2018} but, importantly, they characterize some different rates of convergence required for inferential theory. Let $(\delta_{NT})$ and $(\Delta_{NT})$  be some sequence of positive constants converging to 0 as $N,T\to \infty$. Let $\mathcal{T}_{NT}$ be a nuisance realization set such that it contains $\eta_0$ and that $\widehat{\eta}_{kl}$ belongs to $\mathcal{T}_{NT}$ with probability $1-\Delta_{NT}$ for each $(k,l)$.

\begin{assumption_dml1}[Linear Moment Conditions, Smoothness, and Identification]
\label{linear} \ \ \ \ \ \ \ \ \ \ 
\begin{enumerate}
    \item[(i)] $\psi(W;\theta,\eta)$ is linear in $\theta$: $\psi(W;\theta,\eta)=\psi^a(W,\eta)\theta+\psi^b(W,\eta)$.
    \item[(ii)] $\psi(W;\theta,\eta)$ satisfy the Neyman orthogonality conditions \ref{neyman1} and \ref{neyman2} with respect to the probability measure $P$, or, more generally, \ref{neyman2} can be replaced by a $\lambda_{NT}$ near-orthogonality condition
    \begin{align*}
        \lambda_{NT}:= \sup_{\eta\in \mathcal{T}_{NT}}\left\Vert \partial_{r}{\rm E}[\psi(W;\theta_0,\eta_0+r(\eta-\eta_0))]|_{r=0} \right\Vert\leq \delta_{NT}/\sqrt{N}.
    \end{align*}
    \item[(iii)] The map $\eta\to {\rm E}[\psi(W_{it};\theta,\eta)]$ is twice continuously Gateaux-differentiable on $\mathcal{T}$.
    \item[(iv)] The singular values of the matrix $A_0:= {\rm E}[\psi^a(W_{it};\eta_0)]$ are bounded below by $c_a>0$.
\end{enumerate}    
\end{assumption_dml1}

Assumption \hyperref[linear]{DML1}(i) restricts the focus of this paper to models with linear orthogonal moment conditions, which covers the model in Section \ref{partial_linear_unob}. For nonlinear orthogonal moment conditions, \cite{Chernozhukov2018} has shown that the DML estimator has the same desirable properties under more complicated regularity conditions. Focusing on the linear cases allows us to pay more attention to issues specifically attributed to panel data. Assumption \hyperref[linear]{DML1}(ii) slightly relaxes the orthogonality condition \ref{neyman2} by a near-orthogonality condition, which is useful for the approximate sparse model with approximation errors. Assumption \hyperref[linear]{DML1}(iii) imposes a mild smoothness assumption on the orthogonal moment condition and Assumption \hyperref[linear]{DML1}(iv) is a common condition for identification.

\begin{assumption_dml2}[Moment Regularity and First-Steps] \label{rates} \ \ \ \ \ \ \  \ \ \ 
\begin{enumerate}
    \item[(i)] For all $i\geq 1$, $t\geq 1$, and some $q>2$, $c_m<\infty$, the following moment conditions hold:
    \begin{align*}
        m_{NT}&:=\sup_{\eta\in\mathcal{T}_{NT}}\left({\rm E}\left\Vert \psi(W_{it};\theta_0,\eta)\right\Vert^q\right)^{1/q}\leq c_m, \\
        m'_{NT}&:=\sup_{\eta\in\mathcal{T}_{NT}}\left({\rm E}\left\Vert \psi^a(W_{it};\eta)\right\Vert^q\right)^{1/q}\leq c_m.
    \end{align*}
    \item[(ii)] The following conditions on the statistical rates $r_{NT}$, $r_{NT}'$, $\lambda_{NT}'$ hold for all $i\geq 1$, $t\geq 1$:
    \begin{align*}
        r_{NT}&:= \sup_{\eta\in \mathcal{T}_{NT}} \left\Vert {\rm E}[\psi^a(W_{it};\eta)-\psi^a(W_{it};\eta_0)]\right\Vert \leq \delta_{NT}, \\
        r_{NT}'&:= \sup_{\eta\in \mathcal{T}_{NT}} \left(  {\rm E}\left\Vert\psi(W_{it};\theta_0,\eta)-\psi(W_{it};\theta_0,\eta_0)\right\Vert^2\right)^{1/2} \leq \delta_{NT}, \\
        \lambda_{NT}'&:= \sup_{r\in(0,1),\eta\in\mathcal{T}_{NT}} \left\Vert \partial_r^2 {\rm E}[\psi(W_{it};\theta_0,\eta_0+r(\eta-\eta_0))] \right\Vert \leq \delta_{NT}/\sqrt{N}.
    \end{align*}
\end{enumerate}
\end{assumption_dml2}

Assumption \hyperref[rates]{DML2} regulates the quality of the first-step nuisance estimators. It follows from \cite{Chernozhukov2018} and it can be verified under primitive conditions in the next section. Observe that, if the orthogonal moment function $\psi(W;\theta,\eta)$ is smooth in $\eta$, then $\lambda'_{NT}$ is the dominant rate and it imposes a crude rate requirement of order $\varepsilon_{NT} = o(N^{-1/4})$ on the first-step nuisance parameter in $L^2(P)$ norm, which adimts the proposed two-way cluster-LASSO estimator under suitable sparsity conditions. Furthermore, for some models, e.g., the partial linear model, $\lambda_{NT}'$ can be exactly 0, then it is possible to obtain the weakest possible rate requirement for the first-step estimator, i.e. $\varepsilon_{NT} = o(1)$.

\begin{theorem}[Asymptotic Normality and Variance] \label{normality_1}
    Suppose Assumptions \hyperref[ahk]{AHK}, \hyperref[nondeg]{ND}, \hyperref[linear]{DML1}, \hyperref[rates]{DML2} hold for any ${P}\in\mathcal{P}_{NT}$, then for some $\delta_{NT}\geq N^{-1/2}$, as $(N,T)\to \infty$ jointly and $N/T\to c$,
    \begin{align*}
        \sqrt{N}\left(\widehat{\theta}-\theta_0\right) = -\sqrt{N}A_0^{-1}\sum_{i=1}^N\sum_{t=1}^T\psi(W_{it}; \theta_0, \eta_0) + o_{P}(1) \Rightarrow \mathcal{N}(0,V),
    \end{align*}
   where $ V \coloneq  A_0^{-1} \Omega A^{-1'}_0$ and $\Omega \coloneq \Sigma_a + c \Sigma_g$.
\end{theorem}

We observe that the convergence rate of the two-step estimator $\widehat{\theta}$ resulting from the panel DML procedure is non-standard. It is $\sqrt{N}$-consistent instead of $\sqrt{NT}$-consistent. This is because the cluster dependence introduced by the unit and time components does not decay over time or space. Intuitively, with more persistence, the information carried by data is accumulated more slowly. It is a common feature in the literature of robust inference with cluster dependence\footnote{For example, see \citealp{mackinnon2021wild}, \citealp{Chiang2022},\citealp{CHS_Restat},\citealp{chen2024fixed} among many others.} and it is also related to inferential theory under strong cross-sectional dependence (e.g., \citealp{gonccalves2011moving}).

It is noted that the variance estimator under the cross-fitting is equivalent to estimating the variance in each sub-sample and then averaging across all sub-samples. Since $K,L$ are fixed, the asymptotic analysis is done at the sub-sample level. The next theorem establishes the consistency of this variance estimator under the conventional small-bandwidth assumption. 
\begin{theorem}[Consistent Variance Estimator] \label{consistency_1}
    Assumptions \hyperref[ahk]{AHK}, \hyperref[nondeg]{ND}, \hyperref[linear]{DML1}, \hyperref[rates]{DML2} hold for any ${P}\in\mathcal{P}_{NT}$, and some $q>4$ (defined in Assumption \hyperref[rates]{DML2}), and $M/T^{1/2}=o(1)$. Then, as $N,T\to\infty$ and $N/T\to c$ where $0<c<\infty$, $\widehat{V}_{\rm CHS} = V + o_P(1).$
\end{theorem}

Theorem \ref{consistency_1} can be seen as a generalization of the consistency result for the CHS variance estimator in \cite{CHS_Restat} by allowing for the estimated nuisance parameters in the moment functions. A remaining practical issue is that $\widehat{V}$ is not ensured to be positive semi-definite. It has been shown in \cite{chen2024fixed} that negative variance estimates
happen with a non-trivial number of times under certain data-generating processes. Accordingly, an alternative two-term variance estimator was proposed in \cite{chen2024fixed}. Following the same idea, I propose an alternative variance estimator by dropping the double-counting term $\widehat{\Omega}_{\rm NW}$:
\begin{align*}
    \widehat{V}_{\rm DKA}  = \widehat{A}^{-1}\widehat{\Omega}_{\rm DKA}\widehat{A}^{-1'}, \ \widehat{\Omega}_{\rm DKA} &= \widehat{\Omega}_{\rm A}+\widehat{\Omega}_{\rm DK}.
\end{align*}
The estimator is referred to as the DKA variance estimator because it is a sum of Driscoll-Krray and Arellano variance estimators.\footnote{Note that, the DKA estimator defined in \cite{chen2024fixed} differs from the DKA estimator here by a constant term based on fixed-b asymptotic analysis. Such bias correction is not considered here since the fixed-b properties are not directly applicable in this setting. The conjecture is that the same form of bias correction can be applied here but formally establishing the fixed-b asymptotic results with the presence of estimated nuisance parameters is challenging and out of the scope of this paper, and so is left for future research. } Similar approaches can be found in \cite{mackinnon2021wild}. It relies on the fact that the double-counting term is of small order asymptotically when the panel is two-way clustering. Similar to other two-term cluster-robust variance estimators, it has the computational advantage of guaranteeing positive semi-definiteness but at the cost of inconsistency in the case of no clustering or clustering at the intersection. For theoretical results and more detailed discussions on the trade-off between the ensured positive-definiteness and the risk of being too conservative/losing power, readers are referred to \cite{mackinnon2021wild} and \cite{chen2024fixed}. The next theorem formally shows that the double-counting term is of small order under two-way clustering, and it implies that the $\widehat{V}_{\rm DKA}$ is also consistent for $\Omega$ under two-way clustering. 

\begin{theorem} [Alternative Consistent Variance Estimator]
\label{consistency_2}
    Under the same conditions as Theorem \ref{consistency_1}, we have, as $N,T\to\infty$ and $N/T\to c$ where $0<c<\infty$, $\widehat{V}_{\rm DKA} = \widehat{V}_{\rm CHS} + o_P(1)$.
\end{theorem}
To conclude, in this section, the inferential theory is established for the panel DML estimator, under high-level assumptions on the first-step estimator. Even though the rate of convergence can be slow for the nuisance estimations due to the two-way cluster dependence, the cross-fitting approach for panel models allows for valid inference in a general moment restriction model with growing dimensions in the nuisance parameters. In the next section, I will study a special case of the semiparametric restriction model and consider the complication due to unobserved heterogeneity.

\begin{proof}[\textbf{Proof of Theorem \ref{normality_1}}]
   
    By Assumption \hyperref[rates]{DML2}(i), with probability $1-\Delta_{NT}$, $\widehat{\eta}_{kl}\in \mathcal{T}_{NT}$. So, ${\rm P}(\widehat{\eta}_{kl}\in \mathcal{T}_{NT}, \ \forall (k,l)) \geq 1-KL\Delta_{NT} = 1-o(1)$. Let's denote the event  ${\rm P}(\widehat{\eta}_{kl}\in \mathcal{T}_{\eta}, \ \forall (k,l))$ as $\Eulerconst_\eta$ and the event $\{(W(k,l),W(-k,-l))=(\tilde{W}(k,l),\tilde{W}(-k,-l)),\ {\rm for \ some} \ (k,l) \}$ as $\Eulerconst_{cp}$. By Lemma \ref{indep_couple}, we have ${\rm P}(\Eulerconst_{cp})= 1-o(1)$. By union bound inequality, we have ${\rm P}(\Eulerconst_{\eta}^c \cup \Eulerconst_{cp}^c) \leq {\rm P}(\Eulerconst_{\eta}^c)+ {\rm P}(\Eulerconst_{cp}^c) = o(1)$. So, ${\rm P}(\Eulerconst_{\eta} \cap \Eulerconst_{cp}) = 1- {\rm P}(\Eulerconst_{\eta}^c \cup \Eulerconst_{cp}^c) \geq 1 - o(1)$.

    Let $\widehat{\theta}$ be a solution from equation \ref{dml_exact}. To simplify the notation, we denote
    \begin{align*}
        \widehat{A}_{kl} =\mathbb{E}_{kl} [\psi^a(W_{it},\widehat{\eta}_{kl})], \ \widehat{\bar{A}} &= \frac{1}{KL}\sum_{k=1}^K\sum_{l=1}^L\widehat{A}_{kl}, \ A_0 = {\rm E}[\psi^a(W_{it};\eta_0)], \\
        \widehat{B}_{kl} = \mathbb{E}_{kl}[\psi^b(W_{it},\widehat{\eta}_{kl})], \ \widehat{\bar{B}} &= \frac{1}{KL}\sum_{k=1}^K\sum_{l=1}^L\widehat{B}_{kl}, \ B_0 =  {\rm E}[\psi^b(W_{it};\eta_0)], \\
        \widehat{\bar{\psi}}(\theta)  =  \widehat{\bar{A}}\theta + \widehat{\bar{B}}, \  \bar{\psi}(\theta,\eta) &= \mathbb{E}_{NT} \psi(W_{it};\theta,\eta).
    \end{align*}

\begin{claim}
\label{claim_1}
    On event $\{\Eulerconst_{\eta} \cap \Eulerconst_{cp}\}$,  $\Vert  \widehat{\bar{A}}-A_0\Vert = O_{P}(N^{-1/2}+r_{NT})$.
\end{claim}
    
By Claim \ref{claim_1} and Assumption \hyperref[linear]{DML1}(iv) that all singular values of $A_0$ are bounded below by zero, it follows that all singular values of $\widehat{\bar{A}}$ are also bounded below from zero, on event $\Eulerconst_\eta$. Then, by the linearity in Assumption \hyperref[linear]{DML1}(i), we can write $ \widehat{\theta} = -\widehat{\bar{A}}^{-1}\widehat{\bar{B}}, \ \theta_0 = - A_0^{-1} B_0$. Then, by basic algebra, we have
\begin{align*}
    \sqrt{N}(\widehat{\theta}-\theta_0) = & \sqrt{N}(-\widehat{\bar{A}}^{-1}\widehat{\bar{B}}-\theta_0) = -\sqrt{N}\widehat{\bar{A}}^{-1}(\widehat{\bar{B}}+\widehat{\bar{A}}\theta_0) = -\sqrt{N}\widehat{\bar{A}}^{-1} \widehat{\bar{\psi}}(\theta_0) \\
    =& \sqrt{N}A_0^{-1}\bar{\psi}(\theta_0,\eta_0) + \sqrt{N}A_0^{-1}\left(\widehat{\bar{\psi}}(\theta_0) -\bar{\psi}(\theta_0,\eta_0) \right) \\
    &+\sqrt{N}\left[\left(A_0+\widehat{\bar{A}}-A_0\right)^{-1}-A_0^{-1}\right]\left(\bar{\psi}(\theta_0,\eta_0) + \widehat{\bar{\psi}}(\theta_0) -\bar{\psi}(\theta_0,\eta_0) \right)
\end{align*}

\begin{claim}
\label{claim_2}
    On event $\{\Eulerconst_{\eta} \cap \Eulerconst_{cp}\}$, $\Vert \widehat{\bar{\psi}}(\theta_0) - \bar{\psi}(\theta_0,\eta_0) \Vert = O_{P}(r_{NT}'/\sqrt{N}+\lambda_{NT}+\lambda_{NT}')$.
\end{claim}

By Assumption \hyperref[rates]{DML2}(i) and Jensen's inequality, we have $\Vert A_0 \Vert \leq m_{NT}' \leq c_m$. Then, Claim \ref{claim_2} implies that  
\begin{align*}
    \Vert \sqrt{N}A_0^{-1} \left(\widehat{\bar{\psi}}(\theta_0) -\bar{\psi}(\theta_0,\eta_0) \right) \Vert =& O_{P}(1) O_{P}(\sqrt{N}r_{NT}'+\sqrt{N}\lambda_{NT}+\sqrt{N}\lambda_{NT}') \\
    =& O_{P}(r_{NT}'+\sqrt{N}\lambda_{NT}+\sqrt{N}\lambda_{NT}'),
\end{align*}
Since ${\rm E}[\bar{\psi}(\theta_0,\eta_0)] = 0$, by Lemma \ref{lemma_a2}, we have $\sqrt{N}\bar{\psi}(\theta_0,\eta_0) \overset{d}{\to} \mathcal{N}(0,\Omega)$ where $\Omega = \Sigma_a+c\Sigma_g$ and $\Vert \Omega \Vert<\infty$. By Claims B.1, B.2, and the asymptotic normality of $\sqrt{N}\bar{\psi}(\theta_0,\eta_0)$, we have
\begin{align*}
    & \left\Vert\sqrt{N}\left[\left(A_0+\widehat{\bar{A}}-A_0\right)^{-1}-A_0^{-1}\right]\left(\bar{\psi}(\theta_0,\eta_0) + \widehat{\bar{\psi}}(\theta_0) -\bar{\psi}(\theta_0,\eta_0) \right)\right\Vert \\
    \leq & \left\Vert \widehat{\bar{A}}^{-1} \right\Vert \left\Vert \widehat{\bar{A}}- A_0 \right\Vert  \left\Vert A_0^{-1} \right\Vert \left\Vert \sqrt{N}\left( \bar{\psi}(\theta_0,\eta_0)+ \widehat{\bar{\psi}}(\theta_0) -\bar{\psi}(\theta_0,\eta_0) \right)\right\Vert \\
    = &O_{P}(1) O_{P}\left(N^{-1/2}+r_{NT}\right)  O_{P}(1) \left(O_{P}(1) +O_{P}\left(r_{NT}'+\sqrt{N}\lambda_{NT}+\sqrt{N}\lambda_{NT}'\right) \right) =O_{P}\left(N^{-1/2}+r_{NT}\right),
\end{align*}
and $ \sqrt{N}\left(\widehat{\theta}-\theta_0\right) = A_0^{-1}\mathcal{N}(0,\Omega) + O_{P}\left(N^{-1/2}+r_{NT} + r_{NT}'+\sqrt{N}\lambda_{NT}+\sqrt{N}\lambda_{NT}'\right) \overset{d}{\to} A_0^{-1}\mathcal{N}(0,\Omega).$ The proofs for Claim \ref{claim_1} and Claim \ref{claim_2} are given below.
\end{proof}

\begin{proof}[\textbf{Proof of Theorem \ref{consistency_1}}]

By the same arguments for Theorem \ref{normality_1}, we have
${\rm P}(\Eulerconst_{\eta} \cap \Eulerconst_{cp}) = 1- {\rm P}(\Eulerconst_{\eta}^c \cup \Eulerconst_{cp}^c) \geq 1 - o(1)$. By Claim \ref{claim_1}, we have $\Vert  \widehat{\bar{A}}-A_0\Vert = O_{P}(N^{-1/2}+r_{NT})$ on event $\{\Eulerconst_{\eta} \cap \Eulerconst_{cp}\}$. Therefore, due to $\Vert A_0^{-1} \Vert\leq a_0^{-1}$ ensured by Assumption \hyperref[linear]{DML1}(iv) and $\Omega < \infty$ as shown in Claim \ref{claim_2}, it suffices to show $\Vert \widehat{\Omega}_{\rm CHS} - \Omega\Vert = o_P(1)$. Furthermore, since $K$, $L$ are fixed constants, it suffices to show for each $(k,l)$ that $\Vert \widehat{\Omega}_{{\rm CHS},kl}- \Omega\Vert = o_P(1) $ where
\begin{align*}
{\widehat{\Omega}}_{{\rm CHS},kl}:= & {\widehat{\Omega}}_{a,kl} + {\widehat{\Omega}}_{b,kl} - {\widehat{\Omega}}_{c,kl} + {\widehat{\Omega}}_{d,kl} + {\widehat{\Omega}}_{d,kl}', \\
{\widehat{\Omega}}_{a,kl}:=  &  \frac{1}{N_k T_l^2}\sum\limits_{i\in I_k, t\in S_l, r\in S_l} \psi(W_{it};\widehat{\theta},\widehat{\eta}_{kl})
\psi(W_{ir};\widehat{\theta},\widehat{\eta}_{kl})'  , \\
{\widehat{\Omega}}_{b,kl}:=  &  \frac{K/L}{N_k T_l^2}\sum\limits_{t\in S_l, i\in I_k, j\in I_k} \psi(W_{it};\widehat{\theta},\widehat{\eta}_{kl})
\psi(W_{jt};\widehat{\theta},\widehat{\eta}_{kl})'  , \\
{\widehat{\Omega}}_{c,kl}:=  &  \frac{K/L}{N_k T_l^2}\sum\limits_{i\in I_k, t\in S_l} \psi(W_{it};\widehat{\theta},\widehat{\eta}_{kl})
\psi(W_{it};\widehat{\theta},\widehat{\eta}_{kl})'  , \\
{\widehat{\Omega}}_{d,kl}:=  & \frac{K/L}{N_k T_l^2}  \sum_{m=1}^{M-1} k\left(
\frac{m}{M}\right)  \sum_{t=\lfloor S_l \rfloor}^{\lceil S_l \rceil - m} \sum\limits_{i\in I_k, j\in I_k, j\ne i} \psi(W_{it};\widehat{\theta},\widehat{\eta}_{kl})\psi(W_{j,t+m};\widehat{\theta},\widehat{\eta}_{kl})' .
\end{align*}

Since a sequence of symmetric matrices $\Omega_n$ converges to a symmetric matrix  $\Omega_0$ if and only if $e' \Omega_n e \to  e'\Omega_0 e$ for all comfortable $e$, it suffices to assume without loss of generality that the dimension of $\psi$ to be $1$. To simplify the expression, we denote 
\begin{align*}
    \psi_{it}^{(0)} = \psi(W_{it};\theta_0,\eta_0), \ \ \ \ \widehat{\psi}_{it}^{(kl)} = \psi(W_{it};\widehat{\theta},\widehat{\eta}_{kl})
\end{align*}

\begin{claim}
\label{claim_3}
       On event $\{\Eulerconst_{\eta} \cap \Eulerconst_{cp}\}$,  $ \left\vert{\widehat{\Omega}}_{a,kl} - \Sigma_a\right\vert = O_P\left(N^{-1/2}+r_{NT}'\right)$.
\end{claim}
 
\begin{claim}
\label{claim_4}
    On event $\{\Eulerconst_{\eta} \cap \Eulerconst_{cp}\}$,  $ \left\vert{\widehat{\Omega}}_{b,kl} - c{\rm E}[g_tg_t]\right\vert =O_P\left(N^{-1/2}+r_{NT}'\right)$.
\end{claim}

\begin{claim}
\label{claim_5}
    On event $\{\Eulerconst_{\eta} \cap \Eulerconst_{cp}\}$,  $ \left\vert{\widehat{\Omega}}_{c,kl} \right\vert =  O_P\left(T^{-1}\right)$.
\end{claim}

\begin{claim}
\label{claim_6}
    On event $\{\Eulerconst_{\eta} \cap \Eulerconst_{cp}\}$,  $\left \vert{\widehat{\Omega}}_{d,kl} - c\sum_{m=1}^{\infty}{\rm E}[g_t g_{t+m}]\right\vert = o_{P}(1)$.
\end{claim}

Combining the Claims \ref{claim_3} - \ref{claim_6} completes the proof of Theorem \ref{consistency_1}. 
\end{proof}

\begin{proof}[\textbf{Proof of Theorem \ref{consistency_2}}]
Since $(K,L)$ are fixed constants, it suffices to show for each $(k,l)$ that
\begin{align*}
    \widehat{\Omega}_{{\rm NW},kl}:= \frac{K/L}{N_k T_l^2} \sum\limits_{i\in I_k, t\in S_l, r\in S_l}k\left(  \frac{\left\vert t-r\right\vert }{M}\right)
\psi(W_{it};\widehat{\theta},\widehat{\eta}_{kl})\psi(W_{ir};\widehat{\theta},\widehat{\eta}_{kl})'  =  o_P(1).
\end{align*}
Note that we can rewrite $\widehat{\Omega}_{{\rm NW},kl} =  {\widehat{\Omega}}_{c,kl} + {\widehat{\Omega}}_{e,kl} - {\widehat{\Omega}}_{d,kl}$ where ${\widehat{\Omega}}_{c,kl}$ and ${\widehat{\Omega}}_{d,kl}$ are defined in the proof of Theorem \ref{consistency_1}, and ${\widehat{\Omega}}_{e,kl}$ is defined as follows:
\begin{align*}
{\widehat{\Omega}}_{e,kl}:=  & \frac{K/L}{N_k T_l^2}  \sum_{m=1}^{M-1} k\left(
\frac{m}{M}\right)  \sum_{t=\lfloor S_l \rfloor}^{\lceil S_l \rceil - m} \sum\limits_{i\in I_k, j\in I_k} \psi(W_{it};\widehat{\theta},\widehat{\eta}_{kl})\psi(W_{j,t+m};\widehat{\theta},\widehat{\eta}_{kl})' .
\end{align*}
Observe that by replacing ${\widehat{\Omega}}_{d,kl}$ by ${\widehat{\Omega}}_{e,kl}$, each step in the proof of Claim \ref{claim_6} also follows. It implies that ${\widehat{\Omega}}_{e,kl} = {\widehat{\Omega}}_{d,kl} + o_{P}(1)$. By Claim \ref{claim_5}, we have ${\widehat{\Omega}}_{c,kl} = O_{P}(T^{-1})$. Therefore, we conclude that $\widehat{\Omega}_{{\rm NW},kl} = o_{P}(1)$.
\end{proof}

\begin{proof}[\textbf{Proof of Claim \ref{claim_1}.}]
Fix any $(k,l)$, we have
\begin{align*}
    \left\Vert \widehat{A}_{kl}-A_0 \right\Vert  \leq \left\Vert \widehat{A}_{kl}-{\rm E}[\widehat{A}_{kl}|W(-k,-l)] \right\Vert + \left\Vert {\rm E}[\widehat{A}_{kl}|W(-k,-l)]-A_0 \right\Vert =: \Vert\Delta_{A,1}\Vert + \Vert\Delta_{A,2}\Vert.
\end{align*}
On the event $\{\Eulerconst_{\eta} \cap \Eulerconst_{cp}\}$, we have $\widehat{\eta}_{kl} \in \mathcal{T}_{NT}$ and the independence between $W(-k,-l)$ and $W(k,l)$. So, due to Assumption \hyperref[rates]{DML2}, we have $\left\Vert \Delta_{A,2} \right\Vert \leq r_{NT}$. By iterated expectation, ${\rm E}[\Delta_{A,1}]=0$. To simplify the notation, we denote $\ddot{\psi}^{a,kl}_{it} := \psi^a(W_{it},\widehat{\eta}_{kl}) - {\rm E}[\psi^a(W_{it},\widehat{\eta}_{kl})|W(-k,-l)] $. Consider $\left\Vert \Delta_{A,1} \right\Vert$: 
\begin{align*}
    &{\rm E}\left(\left\Vert \Delta_{A,1} \right\Vert^2| W(-k,-l)\right) = \left(\frac{1}{N_kT_l}\right)^2 {\rm E} \left[ \left\Vert 
\sum_{i\in I_k, t\in S_l}  \ddot{\psi}^{a,kl}_{it}\right\Vert^2 | W(-k,-l)\right] \\
\leq &  \left(\frac{1}{N_kT_l}\right)^2 \sum_{i\in I_k, t\in S_l, r\in S_l} \left\vert {\rm E}\left[\langle\ddot{\psi}^{a,kl}_{it},\ddot{\psi}^{a,kl}_{is} \rangle | W(-k,-l)\right]\right\vert + \sum_{t\in S_l, i\in I_k, j\in I_k}\left\vert {\rm E}\left[\langle\ddot{\psi}^{a,kl}_{it},\ddot{\psi}^{a,kl}_{jt} \rangle | W(-k,-l)\right]\right\vert \\
&  + \sum_{ t\in S_l,i\in I_k} \left\vert {\rm E}\left[\langle\ddot{\psi}^{a,kl}_{it},\ddot{\psi}^{a,kl}_{it} \rangle | W(-k,-l)\right]\right\vert + 2\sum_{m=1}^{ T_l-1}\sum_{t=\min(S_l)}^{\max(S_l)-m}\sum_{i,j\in I_k} \left\vert {\rm E}\left[\langle\ddot{\psi}^{a,kl}_{it},\ddot{\psi}^a_{j,t+m} \rangle | W(-k,-l)\right]\right\vert \\
&  + 2\sum_{m=1}^{ T_l-1}\sum_{t=\min(S_l)}^{\max(S_l)-m} \sum_{i\in I_k} \left\vert {\rm E}\left[\langle\ddot{\psi}^{a,kl}_{it},\ddot{\psi}^a_{i,t+m} \rangle | W(-k,-l)\right]\right\vert =:  \left(\frac{1}{N_kT_l}\right)^2\left( a(1) + a(2) + a(3) +2a(4) +2a(5)\right).
\end{align*}
By conditional Cauchy-Schwarz inequality, for any $i,t,j,s$, we have
\begin{align*}
    \left\vert {\rm E}\left[\langle\ddot{\psi}^{a,kl}_{it},\ddot{\psi}^{a,kl}_{js} \rangle | W(-k,-l)\right]\right\vert & \leq \left({\rm E}\left[\Vert\ddot{\psi}^{a,kl}_{it}\Vert^2 | W(-k,-l)\right] {\rm E}\left[\Vert\ddot{\psi}^{a,kl}_{js}\Vert^2 | W(-k,-l)\right]\right)^{1/2}\\
    &= {\rm E}\left[\Vert\ddot{\psi}^{a,kl}_{it}\Vert^2 | W(-k,-l)\right].
\end{align*}
Therefore, we have
\begin{alignat*}{2}
    a(1)&\leq N_kT_l^2{\rm E}\left[\Vert\ddot{\psi}^{a,kl}_{it}\Vert^2 | W(-k,-l)\right], \ \ \ && a(2)\leq N_k^2 T_l{\rm E}\left[\Vert\ddot{\psi}^{a,kl}_{it}\Vert^2 | W(-k,-l)\right],\\
    a(3)&\leq N_k T_l {\rm E}\left[\Vert\ddot{\psi}^{a,kl}_{it}\Vert^2 | W(-k,-l)\right], \ \ 
 && a(5)\leq N_k T_l^2 {\rm E}\left[\Vert\ddot{\psi}^{a,kl}_{it}\Vert^2 | W(-k,-l)\right].
\end{alignat*}
On the event $\Eulerconst_{\eta}\cap\Eulerconst_{cp}$, we have, for $i\in I_k, t\in S_l$,
\begin{align*}
    &\left({\rm E}\left[\Vert\ddot{\psi}^{a,kl}_{it}\Vert^2 | W(-k,-l)\right]\right)^{1/2} \lesssim \left({\rm E}\left[\Vert \psi^a(W_{it},\widehat{\eta}_{kl}) \Vert^2|W(-k,-l)\right]\right)^{1/2} < \infty,
\end{align*}
where the first inequality follows from expanding the term and applying Jensen's inequality and the second inequality follows from Assumption \hyperref[rates]{DML2}(i). Let $D$ denote the dimension of $\psi^a(W,\eta)$, then we have
\begin{align*}
   a(4)=  a(5) + \sum_{m=1}^{ T_l-1}\sum_{t=\min(S_l)}^{\max(S_l)-m}\sum_{i,j\in I_k, i\ne j} \sum_{d=1}^D {\rm E}\left[\ddot{\psi}^{a,kl}_{d,i,t} \ddot{\psi}^{a,kl}_{d,j,t+m}  | W(-k,-l)\right]
\end{align*}
For each $i\in I_k, t\in S_l$, we can decompose $\ddot{\psi}^{a,kl}_{d,i,t} = a_i^{kl} + g_t^{kl} + e_{it}^{kl}$ where $a_i = {\rm E}[\ddot{\psi}^{a,kl}_{d,i,t}|\alpha_i]$, $g_i = {\rm E}[\ddot{\psi}^{a,kl}_{d,i,t}|\gamma_t]$, and $e_{it} = \ddot{\psi}^{a,kl}_{d,i,t} - a_i - g_t$. Conditional on $W(-k,-l)$, $(a_i^{kl},g_t^{kl}, e_{it}^{kl})$ are mutually uncorrelated, $a_i \ind a_j $ for $i\ne j$, and $g_t^{kl}$ is also beta-mixing with $\beta_g(m)\leq \beta_\gamma(m)$. Therefore, we have
\begin{align*}
    &{\rm E}\left[\ddot{\psi}^{a,kl}_{d,i,t} \ddot{\psi}^{a,kl}_{d,j,t+m}  | W(-k,-l)\right] ={\rm E}\left[g_t^{kl} g_{t+m}^{kl} + e_{it}^{kl} e_{j,t+m}^{kl}| W(-k,-l)\right] \\
    = &{\rm E}\left[g_t^{kl} g_{t+m}^{kl} |W(-k,-l) \right]  + {\rm E}\left[{\rm E}\left[e_{it}^{kl} e_{j,t+m}^{kl}| \alpha_i,\alpha_j, W(-k,-l)\right]|W(-k,-l)\right]
\end{align*}
Note that $\beta$-mixing of $\gamma_t$ implies $\alpha$-mixing with the mixing coefficient $\alpha_{\gamma}(m)\leq \beta_{\gamma}(m)$ for all $m\in \mathbb{Z}^{+}$, and conditional on $W(-k,-l)$ and $\alpha_i$, $e_{it}^{kl}$ is also $\alpha$-mixing with the mixing coefficient not larger than $\alpha_{\gamma}(m)$ by Theorem 14.12 of \cite{hansen2022econometrics}. Then, we have
\begin{align*}
    &\left\vert{\rm E}\left[{\rm E}\left[e_{it}^{kl} e_{j,t+m}^{kl}| \alpha_i,\alpha_j, W(-k,-l)\right]|W(-k,-l)\right]\right\vert \leq  {\rm E}\left[\left\vert{\rm E}\left[e_{it}^{kl} e_{j,t+m}^{kl}| \alpha_i,\alpha_j, W(-k,-l)\right]\right\vert|W(-k,-l)\right] \\
     \lesssim & 8\alpha_{\gamma}(m)^{1-2/q} \left({\rm E}[\vert  \ddot{\psi}^{a,kl}_{d,i,t}\vert^q|W(-k,-l)]\right)^{1/q} \left({\rm E}[\vert  \ddot{\psi}^{a,kl}_{d,j,t+m}\vert^q| W(-k,-l)]\right)^{1/q} \lesssim  32\alpha_{\gamma}(m)^{1-2/q}c_m^2,
\end{align*}
where the first inequality follows from the Jensen's inequality; the second inequality follows from the fact that ${\rm E}[e_{it}^{kl}|\alpha_i,W(-k,-l)] = 0$, and Theorem 14.13(ii) of \cite{hansen2022econometrics}; the last inequality follows from the moment conditions in Assumption \hyperref[rates]{DML2} and that $W(-k,-l)$ is independent of $W(k,l)$ on $\Eulerconst_{cp}$ . Similarly, 
\begin{align*}
    \left\vert {\rm E}\left[g_t^{kl} g_{t+m}^{kl} |W(-k,-l) \right] \right\vert \lesssim  \alpha_{\gamma}(m)^{1-2/q}c_m^2,
\end{align*}
Then, we have
\begin{align*}
   &\frac{1}{N_k^2 T_l} \sum_{m=1}^{ T_l-1}\sum_{t=\min(S_l)}^{\max(S_l)-m}\sum_{i,j\in I_k, i\ne j} \sum_{d=1}^D {\rm E}\left[\ddot{\psi}^{a,kl}_{d,i,t} \ddot{\psi}^{a,kl}_{d,j,t+m}  | W(-k,-l)\right]  \\
   \lesssim & c_m^2 \frac{1}{N_k^2 T_l} \sum_{m=1}^{ T_l-1}\sum_{t=\min(S_l)}^{\max(S_l)-m}\sum_{i,j\in I_k, i\ne j} \sum_{d=1}^D \alpha_{\gamma}(m)^{1-2/q} \nonumber \leq c_m^2 {D} \sum_{m=1}^{\infty} c_{\kappa} exp(-\kappa m)^{1-2/q} \leq   \frac{c_m^2{D} c_{\kappa} }{exp(\kappa(1-2/q))-1}<\infty,
\end{align*}
where the last inequality follows from the geometric sum. Thus, as $(N_k,T_l)\to\infty$ we have
\begin{align*}
   {\rm E}\left(\left\Vert \Delta_{A,1} \right\Vert^2| W(-k,-l)\right)= \left(\frac{1}{N_kT_l}\right)^2 \left[a(1)+a(2)+(3)+2a(4)+2a(5)\right] =O_P(1/T_l) =O_P(1/N).
\end{align*}
where the last step follows from that $L$ is constant and $N/T\to c$ as $N,T\to\infty$. By Markov's inequality, we conclude that conditional on $W(-k,-l)$, $\left\Vert \Delta_{A,1} \right\Vert =  O_P(1/\sqrt{N})$. By Lemma 6.1 that conditional convergence implies unconditional convergence, we have $\left\Vert \Delta_{A,1} \right\Vert =  O_P(1/\sqrt{N})$. To summarize, we have $\left\Vert \widehat{A}_{kl}-A_0 \right\Vert = O_{P}(N^{-1/2}+\delta_{NT})$, which implies $\left\Vert \widehat{\bar{A}}-A_0 \right\Vert = O_{P}(N^{-1/2}+r_{NT})$. 
\end{proof}

\begin{proof}[\textbf{Proof of Claim \ref{claim_2}}]
Since $K$ and $L$ are finite, it suffices to show for any $k,l$, 
\begin{align*}
   \left\Vert \mathbb{E}_{kl} \left[\psi(W_{it};\theta_0,\widehat{\eta}_{kl}) - \psi(W_{it};\theta_0,\eta_0) \right] \right\Vert = O_{P}(r_{NT}'/\sqrt{N_k} + \lambda_{NT} +\lambda_{NT}').
\end{align*}
To simplify the notation, we denote 
\begin{align*}
    \ddot{\psi}^{kl}_{it} &=\psi(W_{it};\theta_0,\widehat{\eta}_{kl})-\psi(W_{it};\theta_0,\eta_0), \\
    \tilde{\ddot{\psi}}^{kl}_{it} & = \ddot{\psi}^{kl}_{it} - {\rm E}[\ddot{\psi}^{kl}_{it}|W(-k,-l)], \\
  b(1) & = \left\Vert 
   \frac{\sqrt{N_k}}{N_kT_l}\sum_{i\in I_k, t\in S_l} \left[\ddot{\psi}^{kl}_{it} - {\rm E}[\ddot{\psi}^{kl}_{it}|W(-k,-l)] \right] \right\Vert \\
   b(2) & = \left\Vert  \frac{1}{N_kT_l} {\rm E}\left[\psi(W_{it};\theta_0,\widehat{\eta}_{kl})|W(-k,-l)\right]-{\rm E}\left[\psi(W_{it};\theta_0,\eta_0)\right]
    \right\Vert.
\end{align*}
We also denote $\tilde{\ddot{\psi}}_{d,it}$ as each element in the vector $\tilde{\ddot{\psi}}^{kl}_{it}$ for $d=1,...,D$, while suppressing the subscripts $k,l$ for convenience. By triangle inequality, we have
    \begin{align*}
    \left\Vert\mathbb{E}_{kl} \left[\psi(W_{it};\theta_0,\widehat{\eta}_{kl}) - \psi(W_{it};\theta_0,\eta_0) \right] \right\Vert  
    \leq  b(1)/\sqrt{N_k}+b(2).
    \end{align*}
To bound $b(1)$, first note that it is mean zero by the iterated expectation argument. 
On the event $\Eulerconst_{\eta}\cap\Eulerconst_{cp}$, we have
\begin{align*}
    &{\rm E}[b(1)^2|W(-k,-l)] \leq \frac{1}{N_k T_l^2}\sum_{i\in I_k, t\in S_l, r\in S_l} \left\vert {\rm E}\left[\langle\tilde{\ddot{\psi}}^{kl}_{it},\tilde{\ddot{\psi}}^{kl}_{is} \rangle | W(-k,-l)\right]\right\vert \\
& +\sum_{t\in S_l, i\in I_k, j\in I_k}\left\vert {\rm E}\left[\langle\tilde{\ddot{\psi}}^{kl}_{it},\tilde{\ddot{\psi}}^{kl}_{jt} \rangle | W(-k,-l)\right]\right\vert  +\sum_{ t\in S_l,i\in I_k}\left\vert {\rm E}\left[\langle\tilde{\ddot{\psi}}^{kl}_{it},\tilde{\ddot{\psi}}^{kl}_{it} \rangle | W(-k,-l)\right]\right\vert \\
&+ 2\sum_{m=1}^{ T_l-1}\sum_{t=\min(S_l)}^{\max(S_l)-m}\sum_{i,j\in I_k}\left\vert {\rm E}\left[\langle\tilde{\ddot{\psi}}^{kl}_{it},\tilde{\ddot{\psi}}^{kl}_{j,t+m} \rangle | W(-k,-l)\right]\right\vert +2\sum_{m=1}^{ T_l-1}\sum_{t=\min(S_l)}^{\max(S_l)-m} \sum_{i\in I_k}\left\vert {\rm E}\left[\langle\tilde{\ddot{\psi}}^{kl}_{it},\tilde{\ddot{\psi}}^{kl}_{i,t+m} \rangle | W(-k,-l)\right]\right\vert\\
&=: c(1) + c(2) + c(3) +2c(4) +2c(5).
\end{align*}
By conditional Cauchy-Schwarz inequality, for any $i,t,j,s$, we have
\begin{align*}
    \left\vert {\rm E}\left[\langle\tilde{\ddot{\psi}}^{kl}_{it},\tilde{\ddot{\psi}}^{kl}_{js} \rangle | W(-k,-l)\right]\right\vert & \leq \left({\rm E}\left[\Vert\tilde{\ddot{\psi}}^{kl}_{it}\Vert^2 | W(-k,-l)\right] {\rm E}\left[\Vert\tilde{\ddot{\psi}}^{kl}_{js}\Vert^2 | W(-k,-l)\right]\right)^{1/2}.
\end{align*}
Applying Minkowski's inequality, Jensen's inequality on the event $\Eulerconst_{\eta}\cap\Eulerconst_{cp}$, we have, for $i\in I_k, t\in S_l$,
\vspace{-5pt}
\begin{align*}
    &\left({\rm E}\left[\Vert\tilde{\ddot{\psi}}^{kl}_{it}\Vert^2 | W(-k,-l)\right] \right)^{1/2} \\
    \leq & \left({\rm E}\left[\Vert \ddot{\psi}^{kl}_{it}\Vert^2 | W(-k,-l)\right] \right)^{1/2} +\left({\rm E}\left[\Vert {\rm E}[\ddot{\psi}^{kl}_{it}|W(-k,-l)]\Vert^2 | W(-k,-l)\right] \right)^{1/2} \\
    \leq & 2 \left({\rm E}\left[\Vert \ddot{\psi}^{kl}_{it}\Vert^2 | W(-k,-l)\right] \right)^{1/2} \leq 2{r'_{NT}}.
\end{align*}
Therefore, we have
\begin{alignat*}{2}
    c(1)&\leq {\rm E}\left[\Vert\tilde{\ddot{\psi}}^{kl}_{it}\Vert^2 | W(-k,-l)\right]= O({r'_{NT}}^2), \ \ \ && c(2)\leq c {\rm E}\left[\Vert\tilde{\ddot{\psi}}^{kl}_{it}\Vert^2 | W(-k,-l)\right]= O({r'_{NT}}^2),\\
    c(3)&\leq \frac{1}{N_k}{\rm E}\left[\Vert\tilde{\ddot{\psi}}^{kl}_{it}\Vert^2 | W(-k,-l)\right]= O({r'_{NT}}^2/N), \ \  \ && c(5)\leq {\rm E}\left[\Vert\tilde{\ddot{\psi}}^{kl}_{it}\Vert^2 | W(-k,-l)\right]]=O({r'_{NT}}^2).
\end{alignat*}
Following similar arguments as for bounding $a(4)$, $c(4)$ is of order $O({r'_{NT}}^2)$. So, we have shown
\begin{align*}
    {\rm E}[b(1)^2|W(-k,-l)] = O_P\left({r_{NT}'}^2\right), 
\end{align*}
which implies $b(1) = O_{P}(r_{NT}')$ by Markov inequality and Lemma 6.1 of \cite{Chernozhukov2018}.

To bound $b(2)$, we first define 
\begin{align*}
    f_{kl}(r) := {\rm E}\left[ \psi(W_{it}, \theta_0, \eta_0+r(\widehat{\eta}_{kl}-\eta_0)|W(-k,-l)\right]-{\rm E}\left[\psi(W_{it};\theta_0,\eta_0)\right], \ r\in[0,1], 
\end{align*}
for some $i\in I_k, t\in S_l$. So, $b(2) = \Vert f_{kl}(1) \Vert$. By expanding $f_{kl}(r)$ around $0$ using mean value theorem and evaluating at $r=1$, we have
\begin{align*}
    f_{kl}(r) = f_{kl}(0) + f_{kl}'(0) + f_{kl}^{''}(\tilde{r})/2, 
\end{align*}
where $\tilde{r}\in(0,1)$. We note that $f_{kl}(0)=0$ on the event $\Eulerconst_{cp}$. On the event $\Eulerconst_{\eta}\cap\Eulerconst{cp}$ and under Assumption \hyperref[linear]{DML1}(ii)(near-orthogonality), we have $ \left\Vert f'_{kl}(0)\right\Vert \leq \lambda_{NT}$ and $ \left\Vert f^{''}_{kl(0)}\right\Vert \leq \lambda_{NT}'$. Therefore, we have shown that $b(2) = O_{P}(\lambda_{NT}) + O_{P}(\lambda'_{NT})$. Combining the bounds for $b(1)$ and $b(2)$ completes the proof of Claim \ref{claim_2}.
    
\end{proof}

\begin{proof}[\textbf{Proof of Claim \ref{claim_3}}]
  By triangle inequality, we have
   \begin{align*}
       \left\vert{\widehat{\Omega}}_{a,kl} -\Sigma_a \right\vert &\leq \left\vert I_{a,1}^{(kl)}\right\vert + \left\vert I_{a,2}^{(kl)}\right\vert + \left\vert I_{a,2}^{(kl)}\right\vert, \\
    I_{a,1}^{(kl)}&:= \frac{1}{N_k T_l^2}\sum\limits_{i\in I_k, t\in S_l, r\in S_l} \left\{\widehat{\psi}_{it}^{(kl)}
\widehat{\psi}_{ir}^{(kl)}  - \psi_{it}^{(0)}\psi_{ir}^{(0)}\right\}  ,\\
     I_{a,2}^{(kl)}&:=   \frac{1}{N_k T_l^2}\sum\limits_{i\in I_k, t\in S_l, r\in S_l} \left\{ \psi_{it}^{(0)}\psi_{ir}^{(0)}- {\rm E}[\psi_{it}^{(0)}\psi_{ir}^{(0)}]\right\}  ,\\
     I_{a,2}^{(kl)}&:= \frac{1}{N_k T_l^2}\sum\limits_{i\in I_k, t\in S_l, r\in S_l} {\rm E}[\psi_{it}^{(0)}\psi_{ir}^{(0)}] - {\rm E}[a_i a_i].
\end{align*}
By law of total covariance and mean-zero property of $\psi_{it}^{(0)}$, we have
\begin{align*}
    {\rm E}\left[\psi_{it}^{(0)}\psi_{ir}^{(0)}\right] = {\rm E}[{\rm E}(\psi_{it}^{(0)},\psi_{ir}^{(0)}|\alpha_i)]+{\rm E}\left({\rm E}[\psi_{it}^{(0)}|\alpha_i]{\rm E}[\psi_{ir}^{(0)}|\alpha_i]\right) 
\end{align*}
Due to the identical distribution of $\alpha_i$ and mean zero, we have
\begin{align*}
    \frac{1}{N_k T_l^2}\sum\limits_{i\in I_k, t\in S_l, r\in S_l} {\rm E}[\psi_{it}^{(0)}\psi_{ir}^{(0)}]= \frac{1}{ T_l^2}\sum\limits_{t\in S_l, r\in S_l} \left \{{\rm E}[{\rm E}(\psi_{it}^{(0)}\psi_{ir}^{(0)}|\alpha_i)]+{\rm E}({\rm E}[\psi_{it}^{(0)}|\alpha_i]{\rm E}[\psi_{ir}^{(0)}|\alpha_i])\right\}
\end{align*}
Conditional on $\alpha_i$, $\{\psi_{it}^{(0)}\}_{t\geq 1}$ is $\beta$-mixing with the mixing coefficient same as $\gamma_t$. Therefore, we can apply Theorem 14.13(ii) in \cite{hansen2022econometrics} and Jensen's inequality:
\begin{align*}
    {\rm E} \left\vert {\rm E}\left[\psi_{it}^{(0)},\psi_{ir}^{(0)}|\alpha_i\right] \right\vert  \leq 8 \left({\rm E}\vert \psi_{it}^{(0)} \vert^{q}\right)^{2/q} \beta_{\gamma}(|t-r|)^{1-2/q}
\end{align*}
Note that $\sum_{t\in S_l,r\in S_l} \beta_{\gamma}(|t-r|)^{1-2/q} \leq \infty $ under Assumption \ref{ahk}. So, $I_{a,2}^{(kl)} = O(1/ T_l^{2})= O(T^{-2})$.

To bound $ I_{a,2}^{(kl)}$, we can rewrite it by triangle inequality as follows:
\begin{align*}
    \left\vert  I_{a,2}^{(kl)}\right\vert \leq &  \left\vert  \frac{1}{N_k} \sum_{i\in I_k} I_{a,2,i}^{(kl)}\right\vert  + \left\vert  \frac{1}{N_k} \sum_{i\in I_k} \tilde{I}_{a,2,i}^{(kl)}\right\vert, \\
    I_{a,2,i}^{(kl)} :=& \frac{1}{ T_l^2}\sum_{t,r\in S_l}\left\{ \psi_{it}^{(0)}\psi_{ir}^{(0)} - {\rm E}\left[\psi_{it}^{(0)}\psi_{ir}^{(0)}|\{\gamma_t\}_{t\in S_l}\right]\right\},\\
    \tilde{I}_{a,2,i}^{(kl)}:= & \frac{1}{ T_l^2}\sum_{t,r\in S_l} \left\{{\rm E}\left[\psi_{it}^{(0)}\psi_{ir}^{(0)}|\{\gamma_t\}_{t\in S_l}\right] - {\rm E}[\psi_{it}^{(0)}\psi_{ir}^{(0)}]\right\}.
\end{align*}
Due to identical distribution of $\alpha_i$, $ \tilde{I}_{a,2,i}^{(kl)}$ does not vary over $i$ so that ${\rm E}\left\vert  \frac{1}{N_k} \sum_{i\in I_k} \tilde{I}_{a,2,i}^{(kl)}\right\vert^2= {\rm E}\left\vert  \tilde{I}_{a,2,i}^{(kl)}\right\vert^2$. Denote $h_i(\gamma_t,\gamma_r)={\rm E}[\psi_{it}^{(0)}\psi_{ir}^{(0)}|\gamma_t,\gamma_r]-{\rm E}[\psi_{it}^{(0)}\psi_{ir}^{(0)}]$. By direct calculation, we have
\begin{align*}
  {\rm E}\left\vert  \tilde{I}_{a,2,i}^{(kl)}\right\vert^2 =&\frac{1}{ T_l^4}\sum_{t,r,t',r'\in S_l} {\rm E}\left[h_i(\gamma_t,\gamma_r)h_i(\gamma_{t'},\gamma_{r'})\right].
\end{align*}
To bound the RHS above, we can apply Lemma 3.4 in \cite{dehling2010central} by verifying the following two conditions: 
\begin{align}
    {\rm E}\left\vert h_i(\gamma_t,\gamma_r) \right\vert^{2+\delta}&<\infty, \label{dehling1} \\
    \int\int \left\vert h_i(u,v) \right\vert^{2+\delta} dF(u) dF(v)&<\infty, \label{dehling2} 
\end{align}
for some $\delta>0$ and $F(.)$ is the common CDF of $\gamma_t$. Consider condition \ref{dehling1}. By Minkowski's inequality, Jensen's inequality, and the law of iterated expectation, we have
\begin{align*}
    \left( {\rm E}\left\vert h_i(\gamma_t,\gamma_r) \right\vert^{2+\delta}\right)^{\frac{1}{2+\delta}} \leq & \left( {\rm E}\left\vert \psi_{it}^{(0)}\psi_{ir}^{(0)} \right\vert^{2+\delta}\right)^{\frac{1}{2+\delta}} + {\rm E}\left\vert \psi_{it}^{(0)}\psi_{ir}^{(0)}\right\vert \leq  \left({\rm E}\left\vert \psi_{it}^{(0)} \right\vert^{4+2\delta}\right)^{\frac{1}{2+\delta}} + {\rm E}\left\vert \psi_{it}^{(0)}\right\vert^2
\end{align*}
where the second inequality follows from H\"{o}lder's inequality and the identical distribution of $\gamma_t$. Let $\delta=\frac{p-4}{2}$, then $\left({\rm E}\left\vert \psi_{it}^{(0)} \right\vert^{4+2\delta}\right)^{\frac{1}{2+\delta}}<c_m$ and ${\rm E}\left\vert \psi_{it}^{(0)}\right\vert^2\leq c_m^2$ follows from Assumption \hyperref[rates]{DML2}(i). Therefore, condition \ref{dehling1} is satisfied. 

Consider condition \ref{dehling2}. By Minkowski's inequality and Jensen's inequality, we have
\begin{align*}
   & \left( \int\int \left\vert {\rm E}[\psi_{it}^{(0)}\psi_{ir}^{(0)}|\gamma_t=u,\gamma_r=v]-{\rm E}[\psi_{it}^{(0)}\psi_{ir}^{(0)}] \right\vert^{2+\delta} dF(u) dF(v) \right)^{\frac{1}{2+\delta}} \\
   \leq &\left( \int\int \left\vert {\rm E}[\psi_{it}^{(0)}\psi_{ir}^{(0)}|\gamma_t=u,\gamma_r=v] \right\vert^{2+\delta} dF(u) dF(v) \right)^{\frac{1}{2+\delta}} +  {\rm E}\left\vert\psi_{it}^{(0)}\psi_{ir}^{(0)}\right\vert \\
   \leq & \left( \int\int  \left({\rm E}\left[\left\vert \psi_{it}^{(0)}\right\vert^2|\gamma_t=u\right]\right)^{\frac{2+\delta}{2}} \left({\rm E}\left[\left\vert \psi_{ir}^{(0)}\right\vert^2|\gamma_r=v\right]\right)^{\frac{2+\delta}{2}} dF(u) dF(v) \right)^{\frac{1}{2+\delta}} +  {\rm E}\left\vert\psi_{it}^{(0)}\right\vert^2 \\
    \leq & \left( \int\int {\rm E}\left[\left\vert \psi_{it}^{(0)}\right\vert^{2+\delta}|\gamma_t=u\right] {\rm E}\left[\left\vert \psi_{ir}^{(0)}\right\vert^{2+\delta}|\gamma_r=v\right] dF(u) dF(v) \right)^{\frac{1}{2+\delta}} +  {\rm E}\left\vert\psi_{it}^{(0)}\right\vert^2 \\
    = & \left({\rm E}\left\vert \psi_{it}^{(0)}\right\vert^{4+2\delta} \right)^{\frac{1}{2+\delta}} + {\rm E}\left\vert\psi_{it}^{(0)}\right\vert^2
\end{align*}
where the second inequality follows from (conditional) H\"{o}lder's inequality and identical distribution of $\gamma_t$; the third inequality follows from Jensen's inequality; the last equality follows from the law of iterated expectation and the identical distribution of $\gamma_t$. Therefore, condition \ref{dehling2} is also satisfied with $\delta=\frac{p-4}{2}$. By Lemma 3.4 in \cite{dehling2010central}, we conclude
\begin{align*}
  {\rm E}\left\vert  \tilde{I}_{a,2,i}^{(kl)}\right\vert^2 =&\frac{1}{ T_l^4}\sum_{t,r,t',r'\in S_l} {\rm E}\left[h_i(\gamma_t,\gamma_r)h_i(\gamma_{t'},\gamma_{r'})\right] = o( T_l^{-1})=o(T^{-1}).
\end{align*}
Therefore, by Markov inequality, we have $ \tilde{I}_{a,2,i}^{(kl)} =o_P(T^{-1/2}) $. Next. consider $\left\vert  \frac{1}{N_k} \sum_{i\in I_k} I_{a,2,i}^{(kl)}\right\vert$. Note that conditional on $\{\gamma_t\}_{t\in S_l}$, $I_{a,2,i}^{(kl)}$ is i.i.d over $i$. So, we have
\begin{align*}
    {\rm E}\left[\left\vert \frac{1}{N_k} \sum_{i\in I_k} I_{a,2,i}^{(kl)}\right\vert^2|\{\gamma_t\}_{t\in S_l}\right] =& \frac{1}{N_k^2} \sum_{i\in I_k} {\rm E}\left[\left\vert  I_{a,2,i}^{(kl)}\right\vert^2|\{\gamma_t\}_{t\in S_l}\right] = \frac{1}{N_k} {\rm E}\left[\left\vert  I_{a,2,i}^{(kl)}\right\vert^2|\{\gamma_t\}_{t\in S_l}\right] 
\end{align*}
By conditional Markov inequality, we have
\begin{align*}
    {\rm P}\left(\left\vert \frac{1}{N_k} \sum_{i\in I_k} I_{a,2,i}^{(kl)}\right\vert>\varepsilon | \{\gamma_t\}_{t\in S_l}\right) = O\left(\frac{1}{N_k} {\rm E}\left[\left\vert  I_{a,2,i}^{(kl)}\right\vert^2|\{\gamma_t\}_{t\in S_l}\right] \right)
\end{align*}
By Minkowski's inequality for infinite sums, Jensen's inequality, and H\"{o}lder's inequality, we have 
\begin{align*}
    \left({\rm E}\left[\left\vert I_{a,2,i}^{(kl)}\right\vert^{2}\right]\right)^{1/2} \lesssim & \frac{1}{ T_l^2} \sum_{t,r\in S_l} \left({\rm E}\left[\psi_{it}^{(0)} \psi_{ir}^{(0)}\right]^{2}\right)^{1/2} \leq \frac{1}{ T_l^2} \sum_{t,r\in S_l} \left({\rm E}\left[\psi_{it}^{(0)} \right]^{4}\right)^{1/2} \leq c_m^2,
\end{align*}
where the last inequality follows from Assumption \hyperref[rates]{DML2}(i) . Then, by law of iterated expectation, we have
\begin{align*}
    {\rm P}\left(\left\vert \frac{1}{N_k} \sum_{i\in I_k} I_{a,2,i}^{(kl)}\right\vert>\varepsilon\right) = O\left(N_k^{-1} \right),
\end{align*}
and $\left\vert \frac{1}{N_k} \sum_{i\in I_k} I_{a,2,i}^{(kl)}\right\vert = O_P\left(N_k^{-1/2}\right) =O_P\left(N^{-1/2}\right) $. Therefore, we have shown $\left \vert I_{a,2}^{kl}\right\vert =O_P\left(N^{-1/2}\right) + o_P\left(T^{-1/2}\right) $.

Next, consider $I_{a,1}^{kl}$. By product decomposition, triangle inequality, and Cauchy-Schwarz inequality, we have 
\begin{align*}
    \left\vert I_{a,1}^{kl}\right\vert \leq &\frac{1}{N_k T_l^2} \sum_{i\in I_k, t\in S_l, r\in S_l} \left\vert\widehat{\psi}_{it}^{(kl)}\widehat{\psi}_{ir}^{(kl)'} - {\psi}_{it}^{(0)}{\psi}_{ir}^{(0)} \right\vert \\
     \leq &\frac{1}{N_k T_l^2} \sum_{i\in I_k, t\in S_l, r\in S_l} \left\{\left\vert\widehat{\psi}_{it}^{(kl)} - {\psi}_{it}^{(0)}\right\vert \left\vert\widehat{\psi}_{ir}^{(kl)} - {\psi}_{ir}^{(0)}\right\vert  +\left\vert {\psi}_{it}^{(0)}\right\vert \left\vert\widehat{\psi}_{ir}^{(kl)} - {\psi}_{ir}^{(0)}\right\vert +\left\vert \widehat{\psi}_{it}^{(kl)} - {\psi}_{it}^{(0)}\right\vert \left\vert\widehat{\psi}_{ir}^{(kl)'} \right\vert\right\} \\
    \lesssim &  R_{kl} \left\{\left\Vert \psi_{it}^{(0)} \right\Vert_{kl,2}+ R_{kl}\right\},
\end{align*}
where $ R_{kl}=\left \Vert \widehat{\psi}_{it}^{(kl)} - {\psi}_{it}^{(0)} \right\Vert_{kl,2}  $. By Markov inequality and under Assumption \hyperref[rates]{DML2}(i), we have 
\begin{align*}
    {\rm E}\left[\frac{1}{N_k T_l}\sum_{i\in I_k, t\in S_l} \left({\psi}_{it}^{(0)}\right)^2\right] = {\rm E}\left\vert\psi(W_{it};\theta_0,\eta_0)\right\vert^2 \leq c_m^2.
\end{align*}
Therefore, $\frac{1}{N_k T_l}\sum_{i\in I_k, t\in S_l} \left({\psi}_{it}^{(0)}\right)^2 = O_P(1)$. To bound $R_{kl}$, note that by Assumption \hyperref[linear]{DML1}(i) (linearity) we have
\begin{align*}
    R_{kl}^2 = &\frac{1}{N_k T_l}\sum_{i\in I_k, t\in S_l} \left(\psi^a(W_{it};\widehat{\eta}_{kl})(\widehat{\theta}-\theta_0) + \psi(W_{it};\theta_0,\widehat{\eta}_{kl})- \psi(W_{it};\theta_0,\eta_0)\right)^2 \\
    \lesssim &  \frac{1}{N_k T_l}\sum_{i\in I_k, t\in S_l} \left\vert\psi^a(W_{it};\widehat{\eta}_{kl})\right\vert^2 \left\vert\widehat{\theta}-\theta_0\right\vert^2 + \frac{1}{N_k T_l}\sum_{i\in I_k, t\in S_l} \left\vert  \widehat{\psi}_{it}^{(kl)} - {\psi}_{it}^{(0)}\right\vert^2 
\end{align*}
By Markov inequality and Assumption \hyperref[rates]{DML2}(i), we have $ \frac{1}{N_k T_l}\sum_{i\in I_k, t\in S_l} \left\vert\psi^a(W_{it};\widehat{\eta}_{kl})\right\vert^2 = O_P(1)$. By Theorem \hyperref[normality_1]{3.1}, $\left\vert\widehat{\theta}-\theta_0\right\vert^2 = O_p(N^{-1})$. Therefore, the first term on RHS is $O_P(N^{-1})$. For the second term on RHS, consider its conditional expectation given the auxiliary sample $W(-k,-l)$. On the event $\Eulerconst_\eta \cap \Eulerconst_{cp}$, we have
\begin{align*}
    &{\rm E}\left[ \frac{1}{N_k T_l}\sum_{i\in I_k, t\in S_l} \left\vert  \widehat{\psi}_{it}^{(kl)} - {\psi}_{it}^{(0)}\right\vert^2 |W(-k,-l) \right] =    {\rm E}\left[ \left\vert  \psi(W_{it};\theta_0,\widehat{\eta}_{kl}) - \psi(W_{it};\theta_0,\eta_0)\right\vert^2 |W(-k,-l) \right] \leq \delta_{NT}^2,
\end{align*}
where the last inequality follows from Assumption \hyperref[rates]{DML2}(ii). Then, by Markov inequality and Lemma 6.1 from \cite{Chernozhukov2018}, we have $R_{kl}^2 = O_P\left(N^{-1}+(r_{NT}')^2\right)$ and so $\left\vert I_{a,1}^{kl} \right\vert = O_P\left(N^{-1/2}+r_{NT}'\right)$. To summarize, we have shown
\begin{align*}
       \left\vert{\widehat{\Omega}}_{a,kl} -\Sigma_a \right\vert  = O_P\left(N^{-1/2}+r_{NT}'\right)+O_P(N^{-1/2})+o_P(T^{-1/2})+ O(T^{-2}) = O_P\left(N^{-1/2}+r_{NT}'\right)
\end{align*}

\end{proof}

\begin{proof}[\textbf{Proof of Claim \ref{claim_4}}]
   By triangle inequality, we have
   \begin{align*}
       &\left\vert{\widehat{\Omega}}_{b,kl} - c{\rm E}[g_t g_t'] \right\vert  \leq \left\vert I_{b,1}^{(kl)}\right\vert + \left\vert I_{b,2}^{(kl)}\right\vert + \left\vert I_{b,3}^{(kl)}\right\vert,\\
    &I_{b,1}^{(kl)}:=  \frac{K/L}{N_k T_l^2}\sum\limits_{t\in S_l,i\in I_k, j\in I_k} \left\{\widehat{\psi}_{it}^{(kl)}
\widehat{\psi}_{jt}^{(kl)}  - \psi_{it}^{(0)}\psi_{jt}^{(0)}\right\}  ,\\
     &I_{b,2}^{(kl)}:=   \frac{K/L}{N_k T_l^2}\sum\limits_{t\in S_l, i\in I_k, j\in I_k} \left\{ \psi_{it}^{(0)}\psi_{jt}^{(0)}- {\rm E}[\psi_{it}^{(0)}\psi_{jt}^{(0)}]\right\}  ,\\
     &I_{b,3}^{(kl)}:=   \frac{K/L}{N_k T_l^2}\sum\limits_{t\in S_l, i\in I_k, j\in I_k} {\rm E}[\psi_{it}^{(0)}\psi_{jt}^{(0)}] - c{\rm E}[g_t g_t'],
\end{align*}
and $ \frac{K/L}{N_k T_l^2} =  \frac{c}{N_k^2 T_l}$. 

Consider $I_{b,3}^{(kl)}$. By the the law of total covariance, we have
 \begin{align*}
      {\rm E}[\psi_{it}^{(0)}\psi_{jt}^{(0)}] = cov(\psi_{it}^{(0)},\psi_{jt}^{(0)}) =  {\rm E}[cov(\psi_{it}^{(0)},\psi_{jt}^{(0)}|\gamma_t)] + cov({\rm E}[\psi_{it}^{(0)}|\gamma_t],{\rm E}[\psi_{jt}^{(0)}|\gamma_t]) = 0 + {\rm E}[g_t g_t'],
 \end{align*}
Due to identical distribution of $\gamma_t$, ${\rm E}[g_t g_t']$ does not vary over $t$ and so $I_{b,3}^{(kl)} = 0$. 

To bound $I_{b,2}^{(kl)}$, we can rewrite it by triangle inequality as follows
\begin{align*}
    &\frac{1}{c} \left\vert I_{b,2}^{kl}\right\vert\leq \left\vert \frac{1}{ T_l} \sum_{t\in S_l} I_{b,2,t}^{(kl)} \right\vert +\left\vert \frac{1}{ T_l} \sum_{t\in S_l} \tilde{I}_{b,2,t}^{(kl)} \right\vert, \\
   & I_{b,2,t}^{(kl)} :=  \frac{1}{N_k^2}\sum_{i,j\in I_k} \left\{ \psi_{it}^{(0)}\psi_{jt}^{(0)}- {\rm E}[\psi_{it}^{(0)}\psi_{jt}^{(0)}|\{\alpha_i\}_{i\in I_k}]\right\} \\
   &\tilde{I}_{b,2,t}^{(kl)} :=  \frac{1}{N_k^2}\sum_{i,j\in I_k} \left\{  {\rm E}[\psi_{it}^{(0)}\psi_{jt}^{(0)}|\{\alpha_i\}_{i\in I_k}] - {\rm E}[\psi_{it}^{(0)}\psi_{jt}^{(0)}]\right\} 
\end{align*}
Due to identical distribution of $\gamma_t$, $ \tilde{I}_{b,2,t}^{(kl)}$ does not vary over $t$ so that $ {\rm E}\left\vert \frac{1}{ T_l} \sum_{t\in S_l} \tilde{I}_{b,2,t}^{(kl)} \right\vert^2 =  {\rm E}\left\vert \tilde{I}_{b,2,t}^{(kl)} \right\vert^2 $. Denote $\zeta_{ij,t} = \psi_{it}^{(0)}\psi_{jt}^{(0)}$. By direct calculation, we have
\begin{align*}
  {\rm E}\left\vert \tilde{I}_{b,2,t}^{(kl)} \right\vert^2 = & \frac{1}{N_k^4} \sum_{i,j\in I_k} \sum_{i',j' \in I_k} {\rm E}\left[ \left({\rm E}[\zeta_{ij,t}|\alpha_i,\alpha_j] - {\rm E}[\zeta_{ij,t}]\right)\left({\rm E}[\zeta_{i'j'}|\alpha_{i'},\alpha_{j'}] - {\rm E}[\zeta_{i'j'}]\right) \right] \\
  \lesssim &  \frac{1}{N_k}  {\rm E}[\zeta_{ij,t}]^2  < \frac{1}{N_k}{\rm E}\left[\psi_{it}^{(0)}\right]^4 =O({1}/{N_k}).
\end{align*}
where the first inequality follows from the assumption that $\alpha_i$ is independent over $i$
and an application of H\"{o}lder's inequality and Jensen's inequality. The second inequality follows from H\"{o}lder's inequality and the last equality follows from Assumption \hyperref[rates]{DML2}(i)  with some $q> 4$. Therefore, by Markov inequality, we have $\left\vert \frac{1}{ T_l} \sum_{t\in S_l} \tilde{I}_{b,2,t}^{(kl)} \right\vert = O_P(N_k^{-1/2})=O_P(N^{-1/2})$.

Now consider $\left\vert \frac{1}{ T_l} \sum_{t\in S_l} I_{b,2,t}^{(kl)} \right\vert$. Note that conditional on $\{\alpha_i\}$, $I_{b,2,t}^{(kl)}$ is also $\beta$-mixing with the mixing coefficient same as $\gamma_t$. Then, with an application of the conditional version of Theorem 14.2 from \cite{davidson1994stochastic}, we have
\begin{align*}
    \left({\rm E}\left[ \left\vert {\rm E}[I_{b,2,t}^{(kl)}|\{\alpha_i\}_{i\in I_k}, \mathcal{F}_{-\infty}^{t-l}]\right\vert^2 | \{\alpha_i\}_{i\in I_k}\right]\right)^{1/2} ]\leq 2(2^{1/2}+1) \beta(l)^{1/2-\frac{2}{q}} \left({\rm E}\left[\vert I_{b,2,t}^{(kl)}\vert^{\frac{q}{2}}|\{\alpha_i\}_{i\in I_k}\right]\right)^{\frac{2}{q}}.
\end{align*}
Then, we can apply the conditional version of Lemma A from \cite{hansen1992consistent} to show that
\begin{align*}
    \left({\rm E}\left[ \left\vert\frac{1}{ T_l}\sum_{t\in S_l}I_{b,2,t}^{(kl)}\right\vert^2| \{\alpha_i\}_{i\in I_k}\right]\right)^{1/2} &\lesssim \frac{1}{ T_l} \sum_{l=1}^{\infty} \beta(l)^{1/2-\frac{2}{q}} \left(\sum_{t\in S_l} \left({\rm E}\left[\vert I_{b,2,t}^{(kl)}\vert^{\frac{q}{2}}|\{\alpha_i\}_{i\in I_k}\right]\right)^{\frac{4}{q}}\right)^{1/2} \\
    &\lesssim \frac{1}{\sqrt{ T_l}} \left({\rm E}\left[\vert I_{b,2,t}^{(kl)}\vert^{\frac{q}{2}}|\{\alpha_i\}_{i\in I_k}\right]\right)^{\frac{2}{q}}
\end{align*}
By conditional Markov inequality, we have
\begin{align*}
    {\rm P}\left(\left\vert \frac{1}{ T_l} \sum_{t\in S_l} I_{b,2,t}^{(kl)}\right\vert>\varepsilon|\{\alpha_i\}_{i\in I_k}\right) = O\left( T_l^{-1} {\rm E}\left[\left\vert I_{b,2,t}^{(kl)}\right\vert^{\frac{q}{2}}|\{\alpha_i\}_{i\in I_k}\right]\right)
\end{align*}
By Minkowski's inequality for infinite sums, Jensen's inequality, and H\"{o}lder's inequality, we have 
\begin{align*}
    \left({\rm E}\left[\left\vert I_{b,2,t}^{(kl)}\right\vert^{\frac{q}{2}}\right]\right)^{\frac{2}{q}} \lesssim & \frac{1}{N_k^2} \sum_{i,j\in I_k} \left({\rm E}\left[\psi_{it}^{(0)} \psi_{jt}^{(0)}\right]^{\frac{q}{2}}\right)^{\frac{2}{q}} \leq \frac{1}{N_k^2} \sum_{i,j\in I_k} \left({\rm E}\left[\psi_{it}^{(0)} \right]^{q}\right)^{\frac{2}{q}} \leq c_m^2,
\end{align*}
where the last inequality follows from Assumption \hyperref[rates]{DML2}(i) . Then, by the law of iterated expectation, we have
\begin{align*}
    {\rm P}\left(\left\vert \frac{1}{ T_l} \sum_{t\in S_l} I_{b,2,t}^{(kl)}\right\vert>\varepsilon\right) = O\left( T_l^{-1/2} \right).
\end{align*}
Therefore, we have shown $\left \vert I_{b,2}^{kl}\right\vert = O_P(N_k^{-1}) + O_P( T_l^{-1/2})= O_P(T^{-1/2})$.

Consider $I_{b,1}^{kl}$. By the similar inequality for $\left\vert I_{a,1}^{kl}\right\vert$, we have
\begin{align*}
    \frac{1}{c} \left\vert I_{b,1}^{kl}\right\vert\lesssim   R_{kl} \left\{\left(\frac{1}{N_k T_l}\sum_{i\in I_k, t\in S_l} \left({\psi}_{it}^{(0)}\right)^2\right)^{1/2}+ R_{kl}\right\},
\end{align*}
where $R_{kl} = \left\Vert \widehat{\psi}_{it}^{(kl)} - {\psi}_{it}^{(0)} \right\Vert_{kl,2}$. We have shown in the proof of Claim \ref{claim_3} that $\left\Vert {\psi}_{it}^{(0)} \right\Vert_{kl,2} = O_P(1)$ and $R_{kl}^2 = O_P\left(N^{-1}+(r_{NT}')^2\right)$. So $\left\vert I_{b,1}^{kl} \right\vert = O_P\left(N^{-1/2}+r_{NT}'\right)$. To summarize 
   \begin{align*}
       \left\vert{\widehat{\Omega}}_{b,kl} - c{\rm E}[g_t g_t'] \right\vert  = O_P\left(N^{-1/2}\right) +O_P\left(T^{-1/2}\right) + O_P\left(N^{-1/2}+r_{NT}'\right) = O_P\left(N^{-1/2}+r_{NT}'\right),
   \end{align*}
which completes the proof of Claim \ref{claim_4}. 
    
\end{proof}

\begin{proof}[\textbf{Proof of Claim \ref{claim_5}}]

 By triangle inequality, we have $\left\vert{\widehat{\Omega}}_{c,kl}\right\vert \leq \left\vert I_{c,1}^{(kl)}\right\vert + \left\vert I_{c,2}^{(kl)}\right\vert + \left\vert I_{c,3}^{(kl)}\right\vert$ where 
\begin{align*}
    I_{c,1}^{(kl)}&:=  \frac{K/L}{N_k T_l^2}\sum\limits_{i\in I_k, t\in S_l} \left\{\widehat{\psi}_{it}^{(kl)}\widehat{\psi}_{it}^{(kl)} - \psi_{it}^{(0)}\psi_{it}^{(0)}\right\}  ,\\
     I_{c,2}^{(kl)}&:=   \frac{K/L}{N_k T_l^2}\sum\limits_{i\in I_k, t\in S_l}\left\{\psi_{it}^{(0)}\psi_{it}^{(0)}- {\rm E}[\psi_{it}^{(0)}\psi_{it}^{(0)}]\right\}  ,\\
     I_{c,3}^{(kl)}&:=  \frac{K/L}{N_k T_l^2}\sum\limits_{i\in I_k, t\in S_l} {\rm E}[\psi_{it}^{(0)}\psi_{it}^{(0)}],
\end{align*}

Consider $I_{c,3}^{(kl)}$. Note that under Assumption \hyperref[rates]{DML2}(i), we have
 \begin{align*}
      {\rm E}[\psi_{it}^{(0)}\psi_{it}^{(0)}] \leq c_m^2.
 \end{align*}
Thus, $I_{c,3}^{(kl)} = O_P(1/ T_l)=O_P(T^{-1})$.

Consider $I_{c,2}^{kl}$. We denote $\xi_{it} =\psi_{it}^{(0)}\psi_{it}^{(0)}- {\rm E}[\psi_{it}^{(0)}\psi_{it}^{(0)}]$. By expanding ${\rm E}\left|I_{c,2}^{kl}\right|^2$ and applying H\"{o}lder's inequality, we have
\begin{align*}
{\rm E}\left|I_{c,2}^{kl}\right|^2 & \leq  \left(\frac{K/L}{N_k T_l^2}\right)^2  \left\{\sum_{i\in I_k, t\in S_l, r\in S_l} {\rm E}\left\vert\xi_{it} \right\vert^2 + \sum_{t\in S_l, i\in I_k, j\in I_k}{\rm E}\left\vert\xi_{it}\right\vert^2 + \sum_{ t\in S_l,i\in I_k}{\rm E}\left\vert\xi_{it}\right\vert^2 \right. \\
 &\left.+ 2\sum_{m=1}^{ T_l-1}\sum_{t=\min(S_l)}^{\max(S_l)-m}\sum_{i,j\in I_k}{\rm E}\left\vert\xi_{it}\right\vert^2+2\sum_{m=1}^{ T_l-1}\sum_{t=\min(S_l)}^{\max(S_l)-m} \sum_{i\in I_k}{\rm E}\left\vert\xi_{it}\right\vert^2\right\}  .
\end{align*}
where the last inequality follows from  Note that for each $i,t$, by  H\"{o}lder's inequality and Assumption \hyperref[rates]{DML2}(i), we have
\begin{align*}
    {\rm E}\left\vert\xi_{it}\right\vert^2 \lesssim {\rm E}\left[\psi(W_{it};\theta_0,\eta_0)^4\right] \leq c_m^4. 
\end{align*}
Thus, ${\rm E}\left|I_{c,2}^{(kl)}\right|^2 = O(T^{-2})$ and so $I_{c,2}^{(kl)}= O_P(T^{-1})$.
 
Now consider $ I_{c,1}^{(kl)}$. Following the same steps for $ I_{b,1}^{(kl)}$, we have
\begin{align*}
    \left\vert I_{c,1}^{(kl)}\right\vert \lesssim \frac{K/L}{ T_l} R_{kl} \left\{\left \Vert {\psi}_{it}^{(0)}\right\Vert_{kl,2}+ R_{kl}\right\},
\end{align*}
where $R_{kl}  = \left \Vert \widehat{\psi}_{it}^{(kl)} - {\psi}_{it}^{(0)}\right\Vert_{kl,2}$.
We have shown in the proof of Claim \ref{claim_3} that $\left \Vert {\psi}_{it}^{(0)}\right\Vert_{kl,2} = O_P(1)$ and $R_{kl}^2 = O_P\left(N^{-1}+(r_{NT}')^2\right)$. So, $\left\vert I_{c,1}^{(kl)} \right\vert = O_P\left(N^{-1/2}/T+r_{NT}'/{T}\right)$. To summarize 
   \begin{align*}
       \left\vert{\widehat{\Omega}}_{c,kl} \right\vert  = O_P\left(T^{-1}\right) + O_P\left(N^{-1/2}/T+r_{NT}'/{T}\right)= O_P\left(T^{-1}\right),
   \end{align*}
which completes the proof of Claim \ref{claim_5}. 
    
\end{proof}

\begin{proof}[\textbf{Proof of Claim \ref{claim_6}}]
   By triangle inequality, we have
   \begin{align*}
       \left\vert{\widehat{\Omega}}_{d,kl} - c\sum_{m=1}^{\infty}{\rm E}[g_t g_t'] \right\vert \leq \left\vert I_{d,1}^{(kl)}\right\vert + \left\vert I_{d,2}^{(kl)}\right\vert + \left\vert I_{d,3}^{(kl)} \right\vert + \left\vert I_{d,4}^{(kl)} \right\vert + \left\vert I_{d,5}^{(kl)} \right\vert +\left\vert I_{d,6}^{(kl)} \right\vert 
   \end{align*}
where 
\begin{align*}
    I_{d,1}^{(kl)}&:=   \frac{K/L}{N_k T_l^2}  \sum_{m=1}^{M-1} k\left(
\frac{m}{M}\right)  \sum_{t=\lfloor S_l \rfloor}^{\lceil S_l \rceil - m} \sum\limits_{i\in I_k, j\in I_k, j\ne i} \left\{\widehat{\psi}_{it}^{(kl)}\widehat{\psi}_{j,t+m}^{(kl)} - {\psi}_{it}^{(0)}{\psi}_{j,t+m}^{(0)} \right\} ,\\
     I_{d,2}^{(kl)}&:=   \frac{K/L}{N_k T_l^2}  \sum_{m=1}^{M-1} k\left(
\frac{m}{M}\right)  \sum_{t=\lfloor S_l \rfloor}^{\lceil S_l \rceil - m} \sum\limits_{i\in I_k, j\in I_k, j\ne i} \left\{{\psi}_{it}^{(0)}{\psi}_{j,t+m}^{(0)} - {\rm E}\left[{\psi}_{it}^{(0)}{\psi}_{j,t+m}^{(0)}\right] \right\} ,\\
     I_{d,3}^{(kl)}&:=   \frac{K/L}{N_k T_l^2}  \sum_{m=1}^{M-1} \left(k\left(
\frac{m}{M}\right)-1\right)  \sum_{t=\lfloor S_l \rfloor}^{\lceil S_l \rceil - m} \sum\limits_{i\in I_k, j\in I_k, j\ne i} {\rm E}\left[{\psi}_{it}^{(0)}{\psi}_{j,t+m}^{(0)}\right]  ,\\
 I_{d,4}^{(kl)}&:=  \frac{K/L}{N_k T_l^2} \sum_{m=M}^{\infty} \sum_{t=\lfloor S_l \rfloor}^{\lceil S_l \rceil - m} \sum\limits_{i\in I_k, j\in I_k, j\ne i} {\rm E}\left[{\psi}_{it}^{(0)}{\psi}_{j,t+m}^{(0)}\right], \\
  I_{d,5}^{(kl)}&:=  \frac{K/L}{N_k T_l^2} \sum_{m=1}^{M-1}\sum_{t=\lfloor S_l \rfloor}^{\lceil S_l \rceil - m} \sum\limits_{i\in I_k, j\in I_k, j\ne i} {\rm E}\left[{\psi}_{it}^{(0)}{\psi}_{j,t+m}^{(0)}\right] - c\sum_{m=1}^{\infty} {\rm E}\left[{\psi}_{it}^{(0)}{\psi}_{j,t+m}^{(0)}\right], \\
  I_{d,6}^{(kl)}&:=c\sum_{m=1}^{\infty} {\rm E}\left[{\psi}_{it}^{(0)}{\psi}_{j,t+m}^{(0)}\right] -c \sum_{m=1}^{\infty} {\rm E}\left[g_{t}g_{t+m}\right]
\end{align*}
and $ \frac{K/L}{N_k T_l^2} =  \frac{c}{N_k^2 T_l}$. 

Consider $I_{d,6}^{(kl)}$. By the law of total covariance, we have
 \begin{align*}
       {\rm E}\left[{\psi}_{it}^{(0)}{\psi}_{j,t+m}^{(0)}\right] &= cov(\psi_{it}^{(0)},\psi_{j,t+m}^{(0)}) \\
       &=  {\rm E}[cov(\psi_{it}^{(0)},\psi_{j,t+m}^{(0)})|\gamma_t,\gamma_{t+m}] + cov({\rm E}[\psi_{it}^{(0)}|\gamma_t],{\rm E}[\psi_{j,t+m}^{(0)}|\gamma_{t+m}])\\
       & = 0+{\rm E}[g_t g_{t+m}'],
 \end{align*}
 where the last equality follows from the properties of Hajek-type decomposition components. Therefore,  $I_{d,6}^{(kl)}=0$.

Consider $I_{d,5}^{(kl)}$. The strict stationarity of $\gamma_t$ implies that $\psi_{it}^{(0)}$ is also strictly stationary over $t$. And under Assumption \ref{ahk}, there is no heterogeneity across $i$. Then, as $M,T\to\infty$, we have $I_{d,5}^{(kl)} = o(1)$.

Consider $I_{d,4}^{(kl)}$. Under Assumption \hyperref[rates]{DML2}(i), $\left({\rm E} \vert{\psi}_{it}^{(0)}\vert^q\right)^{1/q}\leq c_m $ for $q>4$. And conditional on $\alpha_i$,  ${\psi}_{it}^{(0)}$ is $\beta$-mixing with the mixing coefficient not larger than that of $\gamma_t$. Then by Theorem 14.13(ii) in \cite{hansen2022econometrics}, we have
\begin{align*}
    \left\vert {\rm E}\left[{\psi}_{it}^{(0)}{\psi}_{j,t+m}^{(0)}|\{\alpha_i\}_{i\in I_k}\right] \right\vert \leq 8 \left({\rm E}\left[ \vert{\psi}_{it}^{(0)}\vert^q|\alpha_i\right]\right)^{1/q}\left({\rm E}\left[ \vert{\psi}_{j,t+m}^{(0)}\vert^q|\alpha_j\right]\right)^{1/q} \alpha_{\gamma}(m)^{1-2/q}
\end{align*}
By iterated expectation and Jensen's inequality, we have
\begin{align*}
    \left\vert {\rm E}\left[{\psi}_{it}^{(0)}{\psi}_{j,t+m}^{(0)}\right] \right\vert \leq& {\rm E}\left[\left\vert {\rm E}\left[{\psi}_{it}^{(0)}{\psi}_{j,t+m}^{(0)}|\{\alpha_i\}_{i\in I_k}\right] \right\vert\right] \\
    \leq& 8{\rm E}\left[   \left({\rm E}\left[ \vert{\psi}_{it}^{(0)}\vert^q|\alpha_i\right]\right)^{1/q}\left({\rm E}\left[ \vert{\psi}_{j,t+m}^{(0)}\vert^q|\alpha_j\right]\right)^{1/q} \alpha_{\gamma}(m)^{1-2/q}\right] \\
    \leq & 8{\rm E}\left[ \left({\rm E}\left[ \vert{\psi}_{it}^{(0)}\vert^q|\alpha_i\right]\right)^{1/q}\right] {\rm E}\left[ \left({\rm E}\left[ \vert{\psi}_{j,t+m}^{(0)}\vert^q|\alpha_j\right]\right)^{1/q}\right] \alpha_{\gamma}(m)^{1-2/q} \\
    \lesssim & c_m^2\alpha_{\gamma}(m)^{1-2/q}
\end{align*}
where the third inequality follows from that $\alpha_i$ are independent over $i$. Then, as $M\to\infty$,
\begin{align*}
    \left\vert I_{d,4}^{(kl)}\right\vert \leq &\frac{K/L}{N_k T_l^2} \sum_{m=M}^{\infty} \sum_{t=\lfloor S_l \rfloor}^{\lceil S_l \rceil - m} \sum\limits_{i\in I_k, j\in I_k, j\ne i} \left\vert {\rm E}\left[{\psi}_{it}^{(0)}{\psi}_{j,t+m}^{(0)}\right]\right\vert \lesssim \sum_{m=M}^{\infty} \alpha_{\gamma}(m)^{1-2/q} \leq \sum_{m=M}^{\infty} \beta_{\gamma}(m)^{1-2/q} \\
    \leq & c_\kappa \sum_{m=M}^{\infty} exp(-\kappa m) = c_\kappa\left( \frac{1}{1-e^{-k}} - \frac{1-e^{-kM}}{1-e^{-\kappa}} \right)= O\left(e^{-\kappa M}\right).
\end{align*}

Consider $I_{d,3}^{(kl)}$. 
\begin{align*}
   \left\vert I_{d,3}^{(kl)}\right\vert\leq & \frac{K/L}{N_k T_l^2}  \sum_{m=1}^{M-1} \left\vert k\left(
\frac{m}{M}\right)-1\right\vert  \sum_{t=\lfloor S_l \rfloor}^{\lceil S_l \rceil - m} \sum\limits_{i\in I_k, j\in I_k, j\ne i} \left\vert {\rm E}\left[{\psi}_{it}^{(0)}{\psi}_{j,t+m}^{(0)}\right]\right\vert \\
\leq  &c c_m^2\sum_{m=1}^{M-1} \left\vert k\left(
\frac{m}{M}\right)-1\right\vert \alpha_{\gamma}(m)^{1-2/q}.
\end{align*}
Note that for each $m$, $\left\vert k\left(
\frac{m}{M}\right)-1\right\vert\to 0$ as $M\to\infty$. Since $\left\vert k\left(
\frac{m}{M}\right)-1\right\vert \alpha_{\gamma}(m)^{1-2/q} \leq 1$, we can apply dominated convergence theorem to conclude that $ I_{d,3}^{(kl)} =o(1)$.

To bound $I_{d,2}^{(kl)}$, we can rewrite it by triangle inequality as follows
\begin{align*}
    \frac{1}{c} \left\vert I_{d,2}^{(kl)}\right\vert\leq \left\vert \sum_{m=1}^{M-1} \frac{ k\left(
\frac{m}{M}\right)}{ T_l} \sum_{t= \lfloor S_l\rfloor}^{\lceil S_l\rceil -m } I_{d,2,tm}^{(kl)} \right\vert +\left\vert \sum_{m=1}^{M-1} \frac{ k\left(
\frac{m}{M}\right)}{ T_l}  \sum_{t= \lfloor S_l\rfloor}^{\lceil S_l\rceil -m }\tilde{I}_{d,2,tm}^{(kl)} \right\vert,
\end{align*}
where
\begin{align*}
    I_{d,2,tm}^{(kl)} := & \frac{1}{N_k^2}\sum_{i,j\in I_k, i \ne j} \left\{ \psi_{it}^{(0)}\psi_{j,t+m}^{(0)}- {\rm E}[\psi_{it}^{(0)}\psi_{j,t+m}^{(0)}|\{\alpha_i\}_{i\in I_k}]\right\} \\
    \tilde{I}_{d,2,tm}^{(kl)} := & \frac{1}{N_k^2}\sum_{i,j\in I_k, i \ne j} \left\{  {\rm E}[\psi_{it}^{(0)}\psi_{j,t+m}^{(0)}|\{\alpha_i\}_{i\in I_k}] - {\rm E}[\psi_{it}^{(0)}\psi_{j,t+m}^{(0)}]\right\} 
\end{align*}

Due to identical distribution of $\gamma_t$, $ \tilde{I}_{d,2,tm}^{(kl)}$ does not vary over $t$ so that $ {\rm E}\left\vert \sum_{m=1}^{M-1} \frac{ k\left(
\frac{m}{M}\right)}{ T_l}\sum_{t= \lfloor S_l\rfloor}^{\lceil S_l\rceil -m } \tilde{I}_{d,2,tm}^{(kl)} \right\vert^2 \leq  {\rm E}\left\vert \sum_{m=1}^{M-1} k\left(
\frac{m}{M}\right) \tilde{I}_{d,2,tm}^{(kl)} \right\vert^2$. And by Minkowski's inequality, we have
\begin{align*}
    \left({\rm E}\left\vert \sum_{m=1}^{M-1} k\left(
\frac{m}{M}\right) \tilde{I}_{d,2,tm}^{(kl)} \right\vert^2\right)^{1/2} \leq \sum_{m=1}^{M-1} k\left(\frac{m}{M}\right) \left({\rm E}\left[\tilde{I}_{d,2,tm}^{(kl)}\right]^2\right)^{1/2}
\end{align*}

Denote $\zeta_{ijm} = \psi_{it}^{(0)}\psi_{j,t+m}^{(0)}$. By direct calculation, we have
\begin{align*}
  {\rm E}\left\vert \tilde{I}_{d,2,tm}^{(kl)} \right\vert^2 = & \frac{1}{N_k^4} \sum_{i,j\in I_k} \sum_{i',j' \in I_k} {\rm E}\left[ \left({\rm E}[\zeta_{ijm}|\alpha_i,\alpha_j] - {\rm E}[\zeta_{ij,t}]\right)\left({\rm E}[\zeta_{i'j'}|\alpha_{i'},\alpha_{j'}] - {\rm E}[\zeta_{i'j'}]\right) \right] \\
  \lesssim &  \frac{1}{N_k}  {\rm E}[\zeta_{ijm}]^2  < \frac{1}{N_k}{\rm E}\left[\psi_{it}^{(0)}\right]^4 =O({1}/{N_k}).
\end{align*}
where the first inequality follows from the assumption that $\alpha_i$ is independent over $i$
and an application of H\"{o}lder's inequality and Jensen's inequality. The second inequality follows from H\"{o}lder's inequality and the last equality follows from Assumption \hyperref[rates]{DML2}(i)  with some $q> 4$. Therefore, we have 
\begin{align*}
\left({\rm E}\left\vert  \right\vert^2\right)^{1/2} \leq O_P\left(\frac{M}{N^{1/2}}\right)=O_P\left(\frac{M}{T^{1/2}}\right).    
\end{align*}
By Markov inequality, we have $\left\vert \sum_{m=1}^{M-1} \frac{ k\left(
\frac{m}{M}\right)}{ T_l}\sum_{t= \lfloor S_l\rfloor}^{\lceil S_l\rceil -m }\tilde{I}_{d,2,tm}^{(kl)} \right\vert = O_P\left(\frac{M}{T^{1/2}}\right)$.

Now consider $\left\vert \sum_{m=1}^{M-1} \frac{ k\left(
\frac{m}{M}\right)}{ T_l}\sum_{t= \lfloor S_l\rfloor}^{\lceil S_l\rceil -m } I_{d,2,tm}^{(kl)} \right\vert $. By Minkowski's inequality, we have
\begin{align*}
   \left({\rm E} \left\vert \sum_{m=1}^{M-1} \frac{ k\left(
\frac{m}{M}\right)}{ T_l} \sum_{t= \lfloor S_l\rfloor}^{\lceil S_l\rceil -m } I_{d,2,tm}^{(kl)} \right\vert^2 \right)^{1/2} \leq \sum_{m=1}^{M-1}  k\left(
\frac{m}{M}\right)\left({\rm E}\left\vert \frac{1}{ T_l}\sum_{t= \lfloor S_l\rfloor}^{\lceil S_l\rceil -m } I_{d,2,tm}^{(kl)}  \right\vert^2 \right)^{1/2} 
\end{align*}
Following the same steps as for $I_{b,2,tm}^{(kl)}$, we can show 
\begin{align*}
  {\rm E}\left\vert \frac{1}{ T_l}\sum_{t= \lfloor S_l\rfloor}^{\lceil S_l\rceil -m } I_{d,2,tm}^{(kl)}  \right\vert^2 = O\left( T_l^{-1} \right).
\end{align*}
Therefore, $\left\vert \sum_{m=1}^{M-1} \frac{ k\left(
\frac{m}{M}\right)}{ T_l}\sum_{t= \lfloor S_l\rfloor}^{\lceil S_l\rceil -m } I_{d,2,tm}^{(kl)} \right\vert  = O_P\left(\frac{M}{ T_l^{-1/2}}\right)=O_P\left(\frac{M}{T^{-1/2}}\right) $. We have shown $\left \vert I_{b,2}^{(kl)}\right\vert = O_P(1/N_k) + O_P\left(\frac{M}{T^{-1/2}}\right) = O_P\left(\frac{M}{T^{-1/2}}\right)$.

Consider $I_{d,1}^{(kl)}$. Denote 
\begin{align*}
    I_{d,1,m}^{(kl)} = \frac{K/L}{N_k T_l^2} \sum_{t=\lfloor S_l \rfloor}^{\lceil S_l \rceil - m} \sum\limits_{i\in I_k, j\in I_k, j\ne i} \left\{\widehat{\psi}_{it}^{(kl)}\widehat{\psi}_{j,t+m}^{(kl)} - {\psi}_{it}^{(0)}{\psi}_{j,t+m}^{(0)} \right\},
\end{align*} 
for each $m$. Then, $  I_{d,1}^{(kl)} =  \sum_{m=1}^{M-1} k\left(
\frac{m}{M}\right) I_{d,1,m}^{(kl)}$. Following the same steps as for $I_{a,1}^{(kl)}$, we can show
\begin{align*}
    \left\vert I_{d,1,m}^{(kl)} \right\vert = O_P(T^{-1/2}+r_{NT}'),
\end{align*}
for each $m$. Therefore, $ \left\vert I_{d,1}^{(kl)} \right\vert = O_P\left(\frac{M}{T^{-1/2}}+Mr_{NT}'\right)$. Note that $Mr_{NT}' \leq M\delta_{NT} N^{-1/2} = \frac{M}{T^{1/2}} \frac{T^{1/2}}{N^{1/2}} \delta_{NT} = o(1)$.

To summarize 
\begin{align*}
       \left\vert{\widehat{\Omega}}_{d,kl} - c\sum_{m=1}^{\infty}{\rm E}[g_t g_t'] \right\vert& = O_P\left(\frac{M}{T^{-1/2}}+Mr_{NT}'\right) + O_P\left(\frac{M}{T^{1/2}}\right) + o(1) + O(e^{-\kappa M}) + o(1) +0 \\
       &=o_P(1).
\end{align*}
which completes the proof of Claim \ref{claim_6}. 
    
\end{proof}

\subsection{Results for the Partial Linear Model}
\label{app_2}

\begin{proof}[\textbf{Proof of Theorem \ref{thm_asymp_norm}}]
Let $P\in \mathcal{P}_{NT}$ for each $(N,T)$. We denote 
\begin{alignat*}{2}
    A_{NT}& = \frac{1}{NT}(V^Z)'V^D, \ &&\hat{A}_{NT}  = \frac{1}{NT}(Z-{f}\widehat{\zeta}_0)'(D-{f}\widehat{\pi}_0),  \\
  \psi_{NT} &=\frac{1}{NT} (V^Z)'V^g, \ && \hat{\psi}_{NT} = \frac{1}{NT}(Z-f\hat{\zeta}_0)'\left(Y - f\hat{\beta} - (D-f\hat{\zeta})'\theta_0\right).
\end{alignat*}
We can write $\widehat{\theta} - \theta_0 =\hat{A}_{NT}^{-1} \hat{\psi}_{NT}$. By product decomposition, we have
\begin{align*}
    \widehat{\theta} - \theta_0 =& {A}_{NT} ^{-1} {\psi}_{NT}  + {A}_{NT} ^{-1}\left[\hat{\psi}_{NT} - \psi_{NT}\right] + \left[\hat{A}_{NT} ^{-1}-{A}_{NT} ^{-1}\right]\left[\hat{\psi}_{NT}  -\psi_{NT}  \right] + \left[\hat{A}_{NT} ^{-1}-{A}_{NT} ^{-1}\right]\psi_{NT}  
\end{align*}

For the asymptotic normality of $\sqrt{N\wedge T}\left(\widehat{\theta} - \theta_0\right)$, we need to show the following statements: (i) ${A}_{NT} \overset{p}{\to} {A}_0 = {\rm E}[V_{it}^ZV_{it}^D] $; (ii) $\sqrt{N\wedge T}{\psi}_{NT}  \overset{d}{\to} \mathcal{N}(0, {\Omega}_0)$; (iii) $ \sqrt{N\wedge T}\left[\hat{\psi}_{NT} - \psi_{NT}\right] = o(1)$; (iv) $\widehat{A}_{NT} - A_{NT} = o_P(1)$. With statements (i) - (iv) and the identification condition in Assumption \hyperref[regcon_app]{REG-P}(i) such that ${A}_0$ is non-singular, $\sqrt{N\wedge T}\left(\widehat{\theta} - \theta_0\right) \overset{d}{\to} \mathcal{N}\left(0,{A}_0^{-1} {\Omega}_0 {A}_0^{-1'}\right)$. Then, the conclusion of the theorem follows.

First, we note that Assumptions \hyperref[regcon_app]{REG-P}(ii) and \hyperref[ahk]{AHK} imply that $(\bar{F}_i,\bar{F}_t)$ are functions of only $(\alpha_i,\gamma_t,\varepsilon_{it})$, and so are $f_{it}$ and $V_{it}^l$ for $l=g,D,Y,Z$. Therefore, the results based on Hajek-type decomposition are still applicable. Furthermore, under Assumptions \hyperref[regcon_app]{REG-P}(ii), $\bar{F}_i$ is a function of only $(c_i,\epsilon_i^c)$ and $\bar{F}_t$ is a function of only $(d_t,\epsilon_t^d)$, so $f_{it}=(L_{1,it},1)$ is a function of $(X_{it},c_i,d_t,\epsilon_i^c,\epsilon_t^d)$. By definition, ${\rm E}[U^D|X_{it},c_i,d_t]=0$. Given that $(\epsilon_i^c,\epsilon_t^d)$ are independent shocks, ${\rm E}[U^D|X_{it},c_i,d_t,\epsilon_i^c,\epsilon_t^d]=0$. Therefore, ${\rm E}[f_{it}U_{it}^D] = 0$. By definition of $L_{2,it}$ and that $(\epsilon_i^c,\epsilon_t^d)$ are independent shocks, ${\rm E}[f_{it}L_{2,it}]={\rm E}[f_{it}]{\rm E}[L_{2,it}]$. Therefore, ${\rm E}[f_{it}'V_{it}^D] ={\rm E}\left[f_{it} [(L_{2,it} - {\rm E}[L_{2,it}])\eta_{D,2} +U_{it}^D]\right] = 0 $. Similarly, we have ${\rm E}[f_{it}'V_{it}^Y] = 0$.

Statement (i) follows from Lemma \ref{lemma_a1} under Assumptions \hyperref[ahk]{AHK} and \hyperref[regcon_app]{REG-P}(iii). For Statement (ii), we first observe that $V^Z_{it} = Z_{it} ( 1- \zeta_0)$ where $\zeta_0 = \left({\rm E}[f_{it}'f_{it}]\right)^{-1}{\rm E}[f_{it}'Z_{it}]$. Due to the exogeneity condition ${\rm E}[Z_{it}{V}^g] = 0$, we have ${\rm E}[V_{it}^Z{V}^g_{it}] = 0$. With the additional Assumption \hyperref[regcon_app]{REG-P}(iv), Statement (ii) follows from Lemma \ref{lemma_a2}. 

Consider Statement (iii). By product decomposition and triangle inequality, we have
\begin{align}
  NT  \vert \hat{\psi}_{NT} - \psi_{NT}\vert \leq& \vert (f(\zeta_0-\widehat{\zeta}))'(f(\beta_0-\widehat{\beta}) + V^Y+ r^Y - \theta_0 (f(\pi_0-\widehat{\pi}) + V^D+ r^D) )\vert \nonumber \\
  &  + \vert (Z-f\zeta_0)'(\theta_0(f(\widehat{\pi}-\pi_0)) - f(\beta_0 - \widehat{\beta}) + r^g)\vert \nonumber  \\ 
  \leq & \vert (f(\zeta_0-\widehat{\zeta}))'f(\beta_0-\widehat{\beta})\vert +\vert (f(\zeta_0-\widehat{\zeta}))' V^Y\vert +\vert (f(\zeta_0-\widehat{\zeta}))' r^Y \vert  \nonumber  \\
  & + \theta_0\vert (f(\zeta_0-\widehat{\zeta}))' f(\pi_0-\widehat{\pi})\vert + \theta_0\vert (f(\zeta_0-\widehat{\zeta}))'V^D\vert + \theta_0\vert (f(\zeta_0-\widehat{\zeta}))'r^D \vert  \nonumber  \\
  & + \theta_0\vert (V^Z)'f(\widehat{\pi}-\pi_0)\vert + \theta_0\vert (V^Z)' f(\beta_0 - \widehat{\beta})\vert + \theta_0\vert (V^Z)' r^g \vert \label{boundc1}
\end{align}
Under Assumption \hyperref[ahk]{AHK}, the sparse approximation conditions as well as Assumption \hyperref[regcon_app]{REG-P}(ii) - (vii), we can apply Theorem \ref{performance_bound} to obtain that $\Vert f_{it}'(\eta_0-\widehat{\eta})\Vert_{NT,2} = O_P\left( \sqrt{\frac{s\log(p/\gamma)}{N\wedge T}}\right)$, $\Vert \eta_0-\widehat{\eta}\Vert_1 = O_P\left(s\sqrt{\frac{\log (p/\gamma)}{N\wedge T} }\right)$ for $\eta = \zeta,\pi,\beta$, and
$P\left(\max_{j=1,...,p}  \left|\frac{1}{NT}\sum_{i=1}^N\sum_{t=1}^T{\omega}_{j,l}^{-1}{f}_{it,j}V_{it}^l\right| \geq \frac{\lambda}{2c_1NT}\right) \to 0$ for $l=Z,D,Y$ where $\lambda = \frac{2C_\lambda NT}{\sqrt{N\wedge T}}\Phi^{-1}(1-\gamma/2p)$. By Lemma \ref{lemma_a2}, ${\omega}_{j,l} \overset{p}{\to} \frac{A\wedge T}{N} \Sigma_{a,j,l} + \frac{N\wedge T}{T}\Sigma_{g,j,l}$ where $\min_{j\leq p}\Sigma_{a,j}^l >0$ by Assumption \hyperref[regcon_app]{REG-P}(iv) and Lemma \ref{lemma_a1}. Therefore, $\min_j{\omega}_{j,l}^{-1} > 0$, which implies $\frac{1}{NT}\Vert f'V^l\Vert_\infty = O_P(\Phi^{-1}(1-\gamma/2p)/\sqrt{N\wedge T}) = O_P\left(\sqrt{\frac{\log(p/\gamma)}{N\wedge T}}\right)$ for $l=D,Y,Z$. 

Consider the first term in \ref{boundc1}. By Cauchy-Swartz inequality, we have $\frac{\sqrt{N\wedge T}}{NT}\vert (f(\zeta_0-\widehat{\zeta}))'f(\beta_0-\widehat{\beta})\vert \leq \sqrt{N\wedge T} \Vert f_{it}'(\zeta_0-\widehat{\zeta})\Vert_{NT,2}\Vert f_{it}'(\beta_0-\widehat{\beta})\Vert_{NT,2} = O_P\left( \frac{s\log(p/\gamma)}{\sqrt{N\wedge T}}\right)$. Consider the second term in \ref{boundc1}. By Hölder's inequality, we have $\frac{\sqrt{N\wedge T}}{NT}\vert (f(\zeta_0-\widehat{\zeta}))' V^Y\vert \leq \frac{\sqrt{N\wedge T}}{NT}\Vert \zeta_0-\widehat{\zeta} \Vert_1 \Vert f'V^Y\Vert_\infty = O_P\left(s{\frac{\log (p/\gamma)}{\sqrt{N\wedge T}} }\right)$. Consider the third term in \ref{boundc1}. By Cauchy-Swartz inequality, we have $\frac{\sqrt{N\wedge T}}{NT}\vert (f(\zeta_0-\widehat{\zeta}))' r^Y \vert \leq {\sqrt{N\wedge T}}\Vert f_{it}'(\zeta_0-\widehat{\zeta})\Vert_{NT,2}\Vert r_{it}^Y\Vert_{NT,2} =O_P\left( \sqrt{\frac{s\log(p/\gamma)}{N\wedge T}}\right) $. For the last term of \ref{boundc1}, Cauchy-Swartz inequality implies that $\frac{\sqrt{N\wedge T}}{NT}\vert (V^Z)' r\vert\leq  {\sqrt{N\wedge T}}\Vert V_{it}^Y\Vert_{NT,2}\Vert r_{it}^Y\Vert_{NT,2}$. By Assumption \hyperref[regcon_app]{REG-P}(ii), we have $\vert {\rm E}[(V_{it}^Y)^2]^{4(\mu+\delta)}\vert <\infty$. Then we can apply Lemma \ref{lemma_a1} and obtain that $\Vert V_{it}^Y\Vert_{NT,2} \to ({\rm E}[(V_{it}^Y)^2])^{1/2}$. Therefore, we have $\frac{\sqrt{N\wedge T}}{NT}\vert (V^Z)' r\vert = o_P(1)$. The arguments for the rest of the terms in \ref{boundc1} are similar. Under the sparsity condition $s=\frac{\sqrt{N\wedge T}}{\log(p/\gamma)}$, we conclude that $\sqrt{N\wedge T}\vert \hat{\psi}_{NT} - \psi_{NT}\vert = o_P(1)$. 

Consider Statement (iv). By product decomposition, we have 
\begin{align*}
    NT\left\Vert \widehat{A}_{NT} - A_{NT}\right\Vert_1=\left\Vert \left(f(\zeta_0 - \widehat{\zeta})\right)' f(\pi_0-\widehat{\pi}) + \left(f(\zeta_0 - \widehat{\zeta})\right)'(D-f{\pi}_0) + (Z-f\zeta_0)'f(\pi_0-\widehat{\pi})\right\Vert_1 \\
    \leq \left\Vert \left(f(\zeta_0 - \widehat{\zeta})\right)' f(\pi_0-\widehat{\pi}) \right\Vert_1 + \left\Vert\left(f(\zeta_0 - \widehat{\zeta})\right)'(r^D+V^D) \right\Vert_1+ \left\Vert (V^Z)'f(\pi_0-\widehat{\pi})\right\Vert_1
\end{align*}
We observe that, by similar arguments for Statement (iii), $\left\Vert \widehat{A}_{NT} - A_{NT}\right\Vert_1 = o_P(1)$. We have shown Statement (i)-(iv), completing the proof.
\end{proof}

\begin{proof}[\textbf{Proof of Theorem \ref{thm_consistency}}]
    We have shown in the proof of Theorem \ref{thm_asymp_norm} that $\widehat{A}_{NT}-{A}_{NT} = o_P(1)$ and ${A}_{NT}-{A}_0 = o_P(1)$. By triangle inequality, we have $\widehat{A}_{NT}-{A}_0 = o_P(1)$. Then, it suffices to show $\widehat{\Omega}_{\rm CHS} - \Omega = o_P(1)$. We decompose $\widehat{\Omega}_{\rm CHS} $ as follows:
    \begin{align*}
{\widehat{\Omega}}_{{\rm CHS}}:= & {\widehat{\Omega}}_{a} + {\widehat{\Omega}}_{b} - {\widehat{\Omega}}_{c} + {\widehat{\Omega}}_{d} + {\widehat{\Omega}}_{d}', \\
{\widehat{\Omega}}_{a}:=  &  \frac{1}{NT^2} \sum\limits_{i=1}^N  \sum_{t=1}^T \sum_{r=1}^T \psi_{it}(\widehat{\theta},\widehat{\eta})
\psi_{ir}(\widehat{\theta},\widehat{\eta})'  , \ \ 
{\widehat{\Omega}}_{b}:=   \frac{1}{NT^2} \sum\limits_{t=1}^T  \sum_{i=1}^N \sum_{j=1}^N \psi_{it}(\widehat{\theta},\widehat{\eta})
\psi_{jt}(\widehat{\theta},\widehat{\eta})'  , \\
{\widehat{\Omega}}_{c}:=  &  \frac{1}{NT^2} \sum\limits_{i=1}^N  \sum_{t=1}^T \psi_{it}(\widehat{\theta},\widehat{\eta})
\psi_{it}(\widehat{\theta},\widehat{\eta})'  , \ \ 
{\widehat{\Omega}}_{d}:=  \frac{1}{NT^2} \sum_{m=1}^{M-1} k\left(
\frac{m}{M}\right)  \sum_{t=1}^{T - m} \sum\limits_{i=1}^N \sum_{j=1, j \ne i}^N \psi_{it}(\widehat{\theta},\widehat{\eta})\psi_{j,t+m}(\widehat{\theta},\widehat{\eta})' .
\end{align*}
where $\psi_{it}({\theta},{\eta}) = (Z_{it}-f_{it}{\zeta})(Y_{it} - f_{it}{\beta} - {\theta} (D_{it} - f_{it}{\pi}))$ and $\eta=(\zeta,\beta,\pi)$. We need to show $\widehat{\Omega}_a \overset{p}{\to} {\Sigma}_a = {\rm E}[{a}_i^2]$, $\widehat{\Omega}_b \overset{p}{\to} c{\rm E}[{g}_t^2]$, $\widehat{\Omega}_c = o_P(1)$, and $\widehat{\Omega}_d\overset{p}{\to} c\sum_{m=1}^\infty {\rm E}[{g}_t {g}_{t+m}]$. 

First, consider $ \widehat{\Omega}_a - {\rm E}[{a}_i^2] $. By triangle inequality, we have
\begin{align*}
    \left\vert \widehat{\Omega}_a - {\rm E}[a_i^2] \right\vert \leq & \left\vert I_{a,1}\right\vert +  \left\vert I_{a,2}\right\vert + \left\vert I_{a,2}\right\vert, \\
      I_{a,1} :=  & \frac{1}{NT^2} \sum\limits_{i=1}^N  \sum_{t=1}^T \sum_{r=1}^T \left\{\psi_{it}(\widehat{\theta},\widehat{\eta})
\psi_{ir}(\widehat{\theta},\widehat{\eta}) - \psi_{it}({\theta}_0,{\eta}_0)
\psi_{ir}({\theta}_0,{\eta}_0) \right\}, \\
      I_{a,2} :=  & \frac{1}{NT^2} \sum\limits_{i=1}^N  \sum_{t=1}^T \sum_{r=1}^T \left\{{\psi}_{it}({\theta}_0,{\eta}_0){\psi}_{ir}({\theta}_0,{\eta}_0) - {\rm E}[{\psi}_{it}({\theta}_0,{\eta}_0){\psi}_{ir}({\theta}_0,{\eta}_0)]\right\}, \\
      I_{a,2} :=  & \frac{1}{NT^2} \sum\limits_{i=1}^N  \sum_{t=1}^T \sum_{r=1}^T \left\{  {\rm E}[{\psi}_{it}({\theta}_0,{\eta}_0){\psi}_{ir}({\theta}_0,{\eta}_0)] - {\rm E}[{a}_i^2]\right\}.
\end{align*}
Note that in proving Claim \ref{claim_3}, the cross-fitting device is only used to show that $I_{a,1}$ is of small order. Since the arguments for showing $I_{a,2}$ and $I_{a,3}$ to be of small order are basically the same as those in the proof of Claim \ref{claim_3}, they are not repeated here.

Consider $I_{a,1}$. By product decomposition, triangle inequality, and Cauchy-Schwarz inequality, we
have
\begin{align*}
    \left|I_{a,1}\right|\lesssim &R_{NT}\left\{ \left\vert  \psi_{it}({\theta}_0,{\eta}_0) \right\vert_{NT,2} +  R_{NT}\right\}\\
    R_{NT} := &\left\Vert \psi_{it}(\widehat{\theta},\widehat{\eta}) -  \psi_{it}({\theta}_0,{\eta}_0) \right\Vert_{NT,2}
\end{align*}
By Minkowski's inequality, we have
\begin{align*}
    R_{NT} = & \left\Vert \psi^a_{it}({\eta}_0)(\widehat{\theta}-\theta_0) + (\psi^a_{it}({\eta}_0)-\psi^a_{it}(\widehat{\eta}))(\widehat{\theta}-\theta_0) + \psi_{it}({\theta}_0,\widehat{\eta}) - \psi_{it}({\theta}_0,{\eta}_0) \right\Vert_{NT,2}  \\
    \leq& \left\Vert \psi^a_{it}({\eta}_0)(\widehat{\theta}-\theta_0) \right\Vert_{NT,2} + \left\Vert (\psi^a_{it}({\eta}_0)-\psi^a_{it}(\widehat{\eta}))(\widehat{\theta}-\theta_0) \right\Vert_{NT,2} +  \left\Vert  \psi_{it}({\theta}_0,\widehat{\eta}) - \psi_{it}({\theta}_0,{\eta}_0) \right\Vert_{NT,2}\\
    =: & \ R_{a,1} +R_{a,2} + R_{a,3},
\end{align*}
where $\psi^a_{it}({\eta}) := \left(Z_{it}-f_{it}{\zeta}\right)(D_{it}-f_{it}{\pi})$. Under Assumption \hyperref[regcon_app]{REG-P}(ii), we have
\begin{align*}
    {\rm E}[\psi_{it}^a({\eta}_0)]^2={\rm E}[V^Z_{it}(V^D_{it}+r^D_{it})]^2=O_P(1),
\end{align*}
and Markov inequality implies that $\left\Vert  \psi_{it}^a({\eta}_0) \right\Vert_{NT,2} = O_P(1)$. By Theorem \ref{thm_asymp_norm}, we have $ \widehat{\theta} - \theta_0 = O_P\left(\frac{1}{\sqrt{N\wedge T}}\right)$. Therefore, $R_{a,1}\leq \left\Vert {\psi}^a_{it}({\eta}_0) \right\Vert_{NT,2} \vert \widehat{\theta}-\theta_0\vert = O_P\left(\frac{1}{\sqrt{N\wedge T}}\right)$. To bound $R_{a,2}$, we note
\begin{align*}
    \left\Vert {\psi}^a_{it}({\eta}_0)- {\psi}^a_{it}(\widehat{\eta})\right\Vert_{NT,2}= \left\Vert f_{it}'(\widehat{\zeta}-\zeta_0)(D_{it}-f_{it}'\pi_0) +  f_{it}'(\widehat{\zeta}-\zeta_0)f_{it}'(\widehat{\pi}-\pi_0) + (Z_{it}-f_{it}'\zeta_0) f_{it}'(\widehat{\pi}-\pi_0)  \right\Vert_{NT,2}
\end{align*}
Under Assumption \hyperref[regcon_app]{REG-P}(iii), we have ${\rm E}|V_{it}^D|^{8(\mu+\delta)}<\infty$, which implies ${\rm E}\left[\max_{i\leq N,t\leq T} |V_{it}^D|^2\right] \lesssim (NT)^{\frac{1}{4(\mu+\delta)}}$. By Markov inequality, we have $\max_{i\leq N,t\leq T} |V_{it}^D|^2 = O_P((NT)^{\frac{1}{4(\mu+\delta)}})$. As in the proof of Theorem \ref{thm_asymp_norm}, Theorem \ref{performance_bound} can be applied to obtain $\left\Vert f_{it}'(\widehat{\zeta}-\zeta_0)\right\Vert_{NT,2} = O_P\left( \sqrt{\frac{s\log(p/\gamma)}{{N\wedge T}}}\right) $. Then, we have 
\begin{align*}
   R_{a,2}=& \left\Vert f_{it}'(\widehat{\zeta}-\zeta_0)V_{it}^D \right\Vert_{NT,2} \leq \left( \max_{i\leq N,t\leq T} |V_{it}^D|^2 \right)^{1/2}\left\Vert f_{it}'(\widehat{\zeta}-\zeta_0)\right\Vert_{NT,2} \\
    =& O_P((NT)^{\frac{1}{8(\mu+\delta)}}) O_P\left( \sqrt{\frac{s\log(p/\gamma)}{{N\wedge T}}}\right) 
    =  O_P((NT)^{\frac{1}{8(\mu+\delta)}})o_P\left(\frac{1}{(N\wedge T)^{1/4}}\right) =o_P(1).
\end{align*}
Similar arguments can be made to show $R_{a,3}$. Therefore, we have $R_{NT} = o_P(1)$ and so $\widehat{\Omega}_a\overset{p}{\to} \Sigma_a$

It is left to show that $\widehat{\Omega}_b \overset{p}{\to} c{\rm E}[g_t^2]$, $\widehat{\Omega}_c = o_P(1)$, and $\widehat{\Omega}_d\overset{p}{\to} c\sum_{m=1}^\infty {\rm E}[g_t g_{t+m}]$. As is shown in the proofs of Claim \ref{claim_4}, the only step in showing these claims that involves the cross-fitting technique is to show the same term $R_{NT}$ to converge to 0 in probability. Otherwise, the arguments are basically the same and not repeated here. Combining these results, we obtain $\widehat{\Omega} \overset{p}{\to} {\rm E}(a_i^2)+c{\rm E}(g_t^2) + c\sum_{m=1}^{\infty}{\rm E}(g_tg_{t+m}) = \Sigma_a+ c \Sigma_g$.  

To show $\widehat{V}_{\rm DKA} = \widehat{V}_{\rm CHS} + o_P(1)$, it suffices to show $\widehat{\Omega}_{\rm NW} = o_P(1)$. We decompose $\Omega_{\rm NW}$ as follows:
\begin{align*}
    \widehat{\Omega}_{\rm NW} = \widehat{\Omega}_{c} + \widehat{\Omega}_e - \widehat{\Omega}_d,
\end{align*}
where $ \widehat{\Omega}_{c}$ and $ \widehat{\Omega}_{d}$ are defined as above and $\widehat{\Omega}_e $ is defined as follows:
\begin{align*}
    \widehat{\Omega}_e := \frac{1}{NT^2} \sum_{m=1}^{M-1}k\left(\frac{m}{M}\right) \sum_{t=1}^{T-m}\sum_{i=1}^{N}\sum_{j=1}^N \psi(W_{it};\widehat{\theta},\tilde{\eta}) \psi(W_{j,t+m};\widehat{\theta},\tilde{\eta}).
\end{align*}
Following the same arguments as in the proof of Claim \ref{claim_6}, we have $\widehat{\Omega}_e = \widehat{\Omega}_d + o_P(1)$. We have shown $\widehat{\Omega}_c = o_P(1)$. Therefore, we conclude that $\widehat{\Omega}_{NW} = o_P(1)$. So it is proved.
\end{proof}

\end{document}